\documentclass[preprint,9pt,cm]{sigplanconf}

\usepackage{version}
\usepackage{amssymb,amsmath,amscd,amsfonts,latexsym}
\usepackage{tikz}
\usetikzlibrary{automata}
\usetikzlibrary{positioning}
\usetikzlibrary{calc}
\usetikzlibrary{shapes}
\usepackage{pgf}
\usepackage{graphicx}
\usepackage{multirow}
\usepackage{xspace}
\usepackage{booktabs}
\usepackage{yfonts}
\usepackage{mathrsfs}
\usepackage{relsize}
\usepackage{natbib}
\usepackage{amsthm}
\usepackage{mathtools}
\usepackage{ifthen}
\RequirePackage[T1]{fontenc}
\usepackage{thm-restate}

\newcommand{\newextmathcommand}[2]{%
    \newcommand{#1}{\ensuremath{#2}\xspace}
}

\newcommand{\MyPi}{\mathrm{\Pi}}
\newcommand{\MySigma}{\mathrm{\Sigma}}

\newextmathcommand{\A}{\mathcal{A}}
\newextmathcommand{\B}{\mathcal{B}}
\newextmathcommand{\CC}{\mathcal{C}}
\newextmathcommand{\AlB}{\MySigma} %
\newextmathcommand{\tran}{\delta} %
\newextmathcommand{\inc}{+1}
\newextmathcommand{\dec}{-1}
\newextmathcommand{\internal}{0}
\newextmathcommand{\myskip}{\internal}
\newcommand{\zerotest}{\mathrm z}

\newcommand{\langop}[2][OPT]{\ifthenelse{\equal{#1}{OPT}}{\mathcal{L}}{\mathcal{L}_{(#1)}}(#2)}
\newextmathcommand{\Language}{\mathcal{L}}
\newcommand{\sizeop}[1]{|#1|}
\newextmathcommand{\sizeA}{\sizeop{\A}}
\newextmathcommand{\qinit}{q_{\mathsf{init}}}
\newextmathcommand{\qfinal}{q_{\mathsf{final}}}

\newcommand{\C}{\mathcal{C}}

\newextmathcommand{\N}{\mathbb{N}}
\newextmathcommand{\Z}{\mathbb{Z}}

\newcommand{\Powerset}[1]{\mathcal{P}(#1)}

\newtheorem{theorem}{Theorem}
\newtheorem{lemma}[theorem]{Lemma}

\newtheorem{corollary}[theorem]{Corollary}

\newtheorem{claim}[theorem]{Claim}
\theoremstyle{definition}
\newtheorem{definition}[theorem]{Definition}
\newtheorem{remark}[theorem]{Remark}

\newcommand{\lra}[1][]{\stackrel{#1}{\longrightarrow}}

\newcommand{\biwalk}{split run\xspace}
\newcommand{\biwalks}{split runs\xspace}
\newcommand{\bi}[2]{(#1,#2)}
\newextmathcommand{\biw}{\bi\rho\sigma}

\newcommand{\initcounter}[1]{\mathsf{init.counter(#1)}}
\newcommand{\initstate}[1]{\mathsf{init.state}(#1)}
\newcommand{\finalcounter}[1]{\mathsf{final.counter}(#1)}
\newcommand{\finalstate}[1]{\mathsf{final.state}(#1)}
\newcommand{\low}[1]{\mathsf{low}(#1)}
\newcommand{\high}[1]{\mathsf{high}(#1)}
\newcommand{\effect}[1]{\mathsf{effect}(#1)}
\newcommand{\length}[1]{|#1|}
\newcommand{\drop}[1]{\mathsf{drop}(#1)}
\newcommand{\height}[1]{\mathsf{height}(#1)}

\newcommand{\df}[1]{\emph{#1}}
\newcommand{\poly}{\mathrm{poly}}

\newcommand{\card}[1]{\##1}
\newcommand{\sset}{\subseteq}
\newextmathcommand{\NFA}{\mathcal N}
\newcommand{\defeq}{\stackrel{\text{def}}{=}}
\newcommand{\eqdef}{\defeq}

\newcommand{\avail}[1]{\mathrm{avail}(#1)}
\newextmathcommand{\eps}{\varepsilon}
\newcommand{\infini}[1]{\|#1\|_{\infty}}
\newcommand{\infnorm}[1]{\infini{#1}}
\newcommand{\norm}[1]{\|#1\|_1}

\newcommand{\introheading}[1]{\textbf{#1}}
\newcommand{\introbullet}[1]{\smallskip\noindent\hbox to \parindent{\hfil$\bullet$\hfil}}

\newcommand{\dr}[2]{\langle#1,#2\rangle}
\newextmathcommand{\direction}{d}

\newcommand{\dimension}{{|\AlB|}}

\newextmathcommand{\parikhnoarg}{\psi}
\newcommand{\parikh}[1]{\parikhnoarg(#1)}

\newcommand{\Lin}[2]{\mathrm{Lin}(#1; #2)}

\newcommand{\cclass}[1]{\mathrm{#1}}
\newcommand{\newcclass}[2]{\newextmathcommand{#1}{\cclass{#2}}}

\newcclass{\Ptime}{P}
\newcclass{\NL}{NL}
\newcclass{\NP}{NP}
\newcclass{\coNP}{coNP}
\newcclass{\DP}{DP}
\newcclass{\BPP}{BPP}
\newcclass{\PSPACE}{PSPACE}
\newcclass{\SharpP}{\#P}
\newcclass{\SigmaTwoP}{\MySigma_2 P}
\newcclass{\PiTwoP}{\MyPi_2 P}
\newcclass{\NExptime}{NEXP}
\newcclass{\TwoExp}{2EXP}

\newcommand{\polyi}{n^3}

\newcommand{\polyiii}{\polyi}
\newcommand{\polyiiii}{2(\polyi)}
\newcommand{\polyv}{(2n^2+3)(\polyi)} %
\newcommand{\polyvi}{n \polyv + 1} %
\newcommand{\polyvii}{\mathrm{poly_1}} %
\newcommand{\polyMinUnpLen}{n^2 (\polyv)^3}%

\newcommand{\HardAutomaton}[1]{\mathcal{H}_{#1}}

\newcommand{\Nat}{\ensuremath{\mathbb{N}}}

\def\by#1{\mathop{{\hbox{\setbox0=\hbox{$\scriptstyle{#1\quad}$}{$ \mathrel{\mathop{\setbox1=\hbox to \wd0{\rightarrowfill}\ht1=3pt\dp1=-2pt\box1}\limits^{#1}} $}}}}}

\newcommand{\By}[1]{\ensuremath{\xRightarrow{#1}}}

\newcommand{\moves}[1]{\ensuremath{\by{{#1}}}}

\newcommand{\smoves}[1]{\ensuremath{\by{{#1}}_S}}

\newcommand{\Moves}[1]{\ensuremath{\By{{#1}}}}

\newcommand{\subword}{\ensuremath{\preceq}}
\newcommand{\ua}{\ensuremath{\!\!\uparrow}}
\newcommand{\da}{\ensuremath{\!\!\downarrow}}
\newcommand{\SigmaE}{\ensuremath{\Sigma_\varepsilon}}
\newcommand{\nop}{\ensuremath{i}}
\newcommand{\lang}{\mathcal{L}}

\newcommand{\push}{\ensuremath{push}}
\newcommand{\pop}{\ensuremath{pop}}
\newcommand{\etest}{\ensuremath{\bot?}}
\newcommand{\Parikh}{\ensuremath{\parikhnoarg}}

\newcommand{\downop}[1]{{{#1}\da}}
\newcommand{\upop}[1]{{{#1}\ua}}
\newextmathcommand{\downlang}{\downop{\Language}}
\newextmathcommand{\uplang}{\upop{\Language}}

\newcommand{\Asp}{\ensuremath{\mathcal{B}}}
\newcommand{\Null}[1]{{}}
\newcommand{\Au}[1]{\ensuremath{\A_{#1}}}
\newcommand{\iu}[1]{\ensuremath{i_{#1}}}
\newcommand{\Fu}[1]{\ensuremath{F_{#1}}}

\begin{document}

\setlength{\pdfpageheight}{\paperheight}
\setlength{\pdfpagewidth}{\paperwidth}

\conferenceinfo{CONF 'yy}{Month d--d, 20yy, City, ST, Country} 
\copyrightyear{20yy} 
\copyrightdata{978-1-nnnn-nnnn-n/yy/mm} 
\doi{nnnnnnn.nnnnnnn}

\title{Complexity of regular abstractions of one-counter languages}

\authorinfo%
{Mohamed Faouzi Atig%
 \thanks{Supported by DST-VR Project P-02/2014, the Swedish Research Council (VR).}}
{Uppsala University, Sweden}
{mohamed\_faouzi.atig@it.uu.se}

\authorinfo%
{Dmitry Chistikov%
 \thanks{Supported in part by the ERC Synergy award ImPACT.}}
{Max Planck Institute for Software Systems (MPI-SWS), Germany}
{dch@mpi-sws.org}

\authorinfo%
{Piotr Hofman%
 \thanks{Supported by Labex Digicosme, Univ. Paris-Saclay,
         project VERICONISS and by Polish National Science Centre
         grant 2013/09/B/ST6/01575.}}
{LSV, CNRS \& ENS Cachan, Universit\'{e} Paris-Saclay, France}
{piotr.hofman@lsv.ens-cachan.fr}

\authorinfo%
{K Narayan Kumar%
 \thanks{Supported by the DST-VR Project P-02/2014, Infosys Foundation.}}
{Chennai Mathematical Institute, India}
{kumar@cmi.ac.in}

\authorinfo%
{Prakash Saivasan%
 \thanks{Supported by the DST-VR Project P-02/2014, TCS Fellowship.}}
{Chennai Mathematical Institute, India}
{saivasan@cmi.ac.in}

\authorinfo%
{Georg Zetzsche%
 \thanks{Supported by a fellowship within the Postdoc-Program
         of the German Academic Exchange Service (DAAD).}}
{LSV, CNRS \& ENS Cachan, Universit\'{e} Paris-Saclay, France}
{zetzsche@lsv.fr}

\maketitle

\makeatletter
\def \@evenfoot {\scriptsize
               \rlap{\textit{\@preprintfooter}}\hfil
               \thepage \hfil
               \llap{\textit{}}}%
\let \@oddfoot = \@evenfoot
\makeatother

\begin{abstract}
We study the computational and descriptional complexity
of the following transformation: Given a one-counter automaton (OCA) \A,
construct a nondeterministic finite automaton (NFA) \B
that recognizes an abstraction of the language $\langop \A$:
its (1) downward closure, (2) upward closure, or (3) Parikh image.
For the Parikh image over a fixed alphabet and for the upward and
downward closures, we find polynomial-time algorithms that compute
such an NFA.  For the Parikh image with the alphabet as part of
the input, we find a quasi-polynomial time algorithm and prove
a completeness result: we construct a sequence of OCA that admits
a polynomial-time algorithm iff there is one for all OCA.
For all three abstractions, it was previously unknown
if appropriate NFA of sub-exponential size exist.

\end{abstract}

\section{Introduction}

The family of \introheading{one-counter languages} is an %
intermediate class between context-free and regular languages:
it is strictly less expressive than the former and
      strictly more expressive than the latter.
For example,
the language $\{a^m b^m \mid m \ge 0\}$ is one-counter,
but not regular,
and
the set of palindromes over the alphabet $\{a, b\}$
is context-free, but not one-counter.
From the verification perspective,
the corresponding class of automata,
\emph{one-counter automata} \textup{(}OCA\textup{)},
can model some infinite-state phenomena
with its ability to keep track of a non-negative integer counter,
see, e.g.,~%
\cite{ChiticR04},
\cite[Section~5.1]{LafourcadeLT05}, and~%
\cite[Section~5.2]{AlurC11}.

Reasoning about OCA, however, is hardly an easy task.
For example, checking whether two OCA accept some word in common
is undecidable even in the deterministic case;
for nondeterministic OCA even language universality,
as well as language equivalence, is undecidable.
For deterministic OCA, equivalence is \NL-complete;
the proof of the membership in \NL took 40~years%
~\cite{ValiantP75,BohmGJ13}.

This lack of tractability suggests the study of finite-state abstractions for OCA.
Such a transition is a recurrent theme in formal methods:
features of programs beyond finite state
are modeled with infinite-state systems (such as pushdown automata, counter systems,
Petri nets, etc.),
and then finite-state abstractions of these systems come as an important
tool for analysis
(see, e.g.,~\cite{AtigBT08concur,abs-1111-1011,AtigBKS14,LongCMM12,SenV06,AtigBT08fsttcs,EsparzaGP14}).
In our work, we focus on the following three \introheading{regular abstractions},
each capturing a specific feature of a language $\Language \sset \AlB^*$:

\introbullet{$\downop \Language$}
The \df{downward closure} of $\Language$,
denoted $\downop \Language$,
is the set of all subwords (subsequences) of all words $w \in \Language$,
i.e., the set of all words that can be obtained from words in $\Language$
by removing some letters.
The downward closure is always a superset of the original language,
$\Language \sset \downop \Language$, and, moreover, a regular one,
no matter what $\Language$ is, by Higman's lemma~\cite{Hig52}.

\introbullet{$\upop \Language$}
The \df{upward closure} of $\Language$,
denoted $\upop \Language$,
is the set of all superwords (supersequences) of all words $w \in \Language$,
i.e., the set of all words that can be obtained from words in $\Language$
by inserting some letters.
Similarly to $\downop \Language$, the language $\upop \Language$ satisfies
$\Language \sset \upop \Language$ and is always regular.

\introbullet{$\parikh \Language$}
The \df{Parikh image} of $\Language$,
denoted $\parikh \Language$, is the set of all vectors $v \in \N^{|\AlB|}$,
that count the number of occurrences of letters of \AlB in words from $\Language$.
That is, suppose $\AlB = \{a_1, \ldots, a_k\}$, then
every word $w \in \Language$ corresponds to a vector $\parikh w = (v_1, \ldots, v_k)$ such that
$v_i$ is the number of occurrences of $a_i$ in~$w$.
The set $\parikh \Language$ is always a regular subset of $\N^{|\AlB|}$ if \Language is context-free,
by the Parikh theorem~\cite{Parikh66}.

\smallskip
It has long been known that
all three abstractions can be \emph{effectively computed}
for context-free languages (CFL),
by the results of van Leeuwen~\cite{VanLeeuwen78}
and Parikh~\cite{Parikh66}.
Algorithms performing these tasks,
as well as finite automata recognizing these abstractions,
are now widely used as building blocks
in the language-theoretic approach to verification.
Specifically, computing upward and downward closures
occurs as an ingredient in the analysis of systems communicating
via shared memory, see, e.g.,~%
\cite{AtigBT08concur,abs-1111-1011,AtigBKS14,LongCMM12}.
As the recent paper~\cite{TorreMW15} shows, for parameterized networks
of such systems the decidability hinges
on the ability to compute downward closures.
The Parikh image as an abstraction in the verification of infinite-state
systems has been used extensively; see, e.g.,
\cite{%
      DBLP:conf/atva/AbdullaAC13,
      KT10,
      DBLP:conf/cav/HagueL12,
      Esparza:2011:CPV:1926385.1926443,
      SenV06,
      AtigBT08fsttcs,
      EsparzaGP14,
      Ganty:2012:AVA:2160910.2160915,
      AbdullaAMS15%
     }.
For pushdown systems, it is possible to construct a linear-sized
existential Presburger formula that captures the Parikh image~\cite{DBLP:conf/cade/VermaSS05},
which leads, for a variety of problems
(see, e.g.,
\cite{%
      DBLP:conf/atva/AbdullaAC13,
      KT10,
      DBLP:conf/cav/HagueL12,
      Esparza:2011:CPV:1926385.1926443%
     }),
to algorithms that rely on deciding satisfiability
for such formulas (which is in \NP).
Finite automata for Parikh images
are used as intermediate representations, for example,
in the analysis of multi-threaded programs~%
\cite{SenV06,AtigBT08fsttcs,EsparzaGP14}
and in recent work on so-called availability languages~\cite{AbdullaAMS15}.

Extending the scope of these three abstractions
from CFL to other classes of languages has been a natural topic of interest.
Effective constructions for the downward closure have been
developed for Petri nets~\cite{HabermehlMW10}
and stacked counter automata~\cite{Zetzsche15stacs}.
The paper~\cite{Zetzsche15icalp} gives a sufficient condition for
a class of languages to have effective downward closures;
this condition has since been applied to
higher-order pushdown automata~\cite{HagueKO16}.
The effective regularity of the Parikh image is known for
linear indexed languages~\cite{DuskeP84},
phase-bounded and scope-bounded
multi-stack visibly pushdown languages~\cite{TorreMP07,TorreNP14}, and
availability languages~\cite{AbdullaAMS15}.
However, there are also negative results:
for example, it is not possible to effectively
compute the downward closure of languages recognized by lossy channels automata---%
this is a corollary of the fact that,
for the set of reachable configurations of a lossy channel system,
boundedness is undecidable~\cite{Mayr03}.

\subsection*{Our contribution}

We study the construction of nondeterministic finite automata
(NFA) for $\downop \Language$, $\upop \Language$, and $\parikh \Language$,
if $\Language$ is given as an OCA \A with $n$~states: $\Language = \langop \A$.
It turns out that for one-counter languages---a proper subclass of CFL---%
all three abstractions
can be computed much more efficiently than for the entire class of CFL.

\subparagraph*{Upward and downward closures:}
We show, for OCA,
how to construct NFA accepting \uplang and \downlang in polynomial time
(Theorems \ref{thm:upwardclosure} and \ref{thm:DownwardClosure}).
The construction for \uplang
is straightforward, but the one for \downlang is
involved and uses pumping-like techniques from automata theory.

These results are in contrast with the exponential lower bounds
known for both closures in the case of CFL~\cite{VanLeeuwen78}:
Several constructions for \uplang and \downlang have been
proposed in the literature
(see, e.g.,~\cite{VanLeeuwen78,GruberHK07,Courcelle91,BachmeierLS15}),
and the best in terms of the size of NFA
are exponential, due to van~Leeuwen~\cite{VanLeeuwen78}
and Bachmeier, Luttenberger, and Schlund~\cite{BachmeierLS15}, respectively.

\subparagraph*{Parikh image:}
For OCA, the problem of constructing NFA
for the Parikh image turns out to be quite tricky.
While we were unable to solve the problem completely,
we make significant progress towards its solution:

\introbullet{}
For any fixed alphabet $\AlB$ we provide a complete
solution: We find a polynomial-time algorithm that
computes an NFA for $\parikh{\langop{\A}}$ that has size $O(|\A|^{\poly(|\Sigma|)})$
(Theorem~\ref{conf:th:parikh-bounded}).
Two key ingredients of this construction are a sophisticated version of a pumping lemma
(Lemma~\ref{conf:l:unpump-summary};
 cf.\ a standard pumping lemma for one-counter languages,
 due to Latteux~\cite{Latteux83}) and
the classic Carath\'eodory theorem for cones in a multi-dimensional space.

\introbullet{}
We provide a quasi-polynomial solution to this problem
in the general case: We find an algorithm that constructs a suitable
NFA of size $O(|\AlB| \cdot |\A|^{O(\log(|\A|))})$ (Theorem~\ref{thm:parikh-unbounded}).
This construction has two steps,
both of which are of interest. In the first step we show, using a
combination of local and global transformations on runs
(Lemmas~\ref{lem:LMGood} and~\ref{lem:boundrevNew}), that we may focus our
attention on runs with at most polynomially many reversals. In the second step, which also
works for pushdown automata, we turn the bound on reversals, using an argument
with a flavour of Strahler numbers~\cite{EsparzaLS14},
into a logarithmic bound on the stack size
of a pushdown system (Lemma~\ref{lem:rbpdas}).

\introbullet{}
We prove a lower-bound type result
(Theorem~\ref{completeness:result}):
We find a sequence of OCA $(\HardAutomaton{n})_{n \ge 1}$,
where $n$ denotes the number of states, over alphabets of
growing size, that admits
a polynomial-time algorithm for computing an NFA for the Parikh image
if and only if there is such an algorithm for all OCA.
Thus, the problem of transforming an arbitrary OCA~\A
into an NFA for $\parikh{\langop{\A}}$
is reduced to performing this transformation on $\HardAutomaton{n}$,
which enables us to call $\HardAutomaton{n}$ \emph{complete}.
This result also has a counterpart
referring to just the existence of NFA of polynomial size.

\smallskip
For the Parikh image of CFL, a number of constructions
can be found in the literature as well;
we refer the reader to the paper
by Esparza et~al.~\cite{EsparzaGKL11} for a survey
and state-of-the-art results:
exponential upper and lower bounds of the form $2^{\Theta(n)}$
on the size of NFA for $\parikh\Language$.

\subsection*{Applications}

Our results show that for OCA, unlike for pushdown systems,
NFA representations of downward and upward closures and Parikh image (for fixed
alphabet size) have efficient polynomial constructions.
This suggests a possible way around
standard \NP procedures that handle existential Presburger formulas.
This insight also leads to
significant gains when abstractions are used in a nested manner,
as illustrated by the following examples.

Consider a \emph{network of pushdown systems} communicating via a shared memory.  
The reachability problem is undecidable in this setting. In~\cite{AtigBKS14}
a restriction called stage-boundedness, generalizing context-boundedness,
is explored. During a stage, the memory can be written to only by one system.
Reachability along runs with at most $k$~stages is decidable
when all but one pushdown in the network are
counters. The procedure in~\cite{AtigBKS14} uses NFA that accept upward and downward
closures of one-counter languages;
the polynomial-time algorithms developed in the present paper bring the
complexity from \NExptime down to \NP for any network with a fixed number of
components.

\emph{Availability expressions}~\cite{HoenickeMO10} extend regular expressions by an additional
counting operator to express quantitative properties of behaviours.
It uses a feature called \emph{occurrence constraint} to impose a set of linear constraints
on the number of occurrences of alphabet symbols in sub-expressions.
As the paper~\cite{AbdullaAMS15} shows, the emptiness problem for
availability expressions is decidable, and the algorithm
involves nested calls to Parikh-image computation for OCA.
Our quasi-polynomial time algorithm for the Parikh image brings the complexity
from non-elementary down to \TwoExp.

\section{Preliminaries}
\label{s:prelim}

\subsection{One-counter automata}
\label{s:prelim:oca}

A \emph{one-counter automaton \textup{(}OCA\textup{)}}~$\A$
is a 5-tuple $(Q, \AlB, \tran, q_0, F)$
where $Q$ a a finite set of states,
$q_0\in Q$ is an initial state, and $F\subseteq Q$ is a set of final states.
$\AlB$ is a finite alphabet and
$\tran \subseteq
    Q \times
    (\AlB \cup \{\eps\}) \times
    \{\dec, \internal, \inc, \zerotest\} \times
    Q$
is a set of transitions.
Transitions $(p_1, a, s, p_2) \in \tran$ are classified as
\emph{incrementing} ($s = \inc$),
\emph{decrementing} ($s = \dec$),
\emph{internal} ($s = \internal$), or
\emph{tests for zero} ($s = \zerotest$).
The \emph{size} of $\A$, denoted $\sizeA$, is its number
of states, $\sizeop{Q}$.

A \df{configuration} of an OCA is a pair that consists of a state and
a (non-negative) counter value, i.e., $(q,n)\in Q\times \N$.
A pair $(p_1,c_1) \in Q \times \Z$ may evolve to a pair $(p_2,c_2) \in Q \times \Z$ via
a transition $t = (p_1, a, s, p_2) \in \tran$ iff
either $s \in \{\dec, \internal, \inc\}$ and $c_1 + s = c_2$,
or $s = \zerotest$ and $c_1=c_2=0$.
We denote this by $(p_1, c_1) \moves{t} (p_2, c_2)$.

Consider a sequence of the form
$\pi = (p_0,c_0)$, $t_1$, $(p_1,c_1)$, $t_2$, \ldots, $t_m$, $(p_m, c_m)$
where $(p_i, c_i) \in Q \times \Z$ for $0 \le i \le m$
and, whenever $i > 0$, it also holds that
$t_i \in \tran$ and $(p_{i-1},c_{i-1}) \moves{t_i} (p_i,c_i)$.
We say that $\pi$ \df{induces} a word~$w = a_1 a_2 \ldots a_m \in \AlB^*$
where $a_i\in \AlB \cup \{\eps\}$  and $t_i=(p_{i-1}, a_i, s, p_i)$;
we also say that the word $w$ can be \df{read} or \df{observed} along
the sequence $\pi$.
We call the sequence $\pi$:
\begin{itemize}
\item a \df{quasi-run}, denoted $\pi = (p_0, c_0) \Moves{w}_{\A} (p_m, c_m)$,
      if none of $t_i$ is a test for zero;
\item a \df{run}, denoted $\pi = (p_0, c_0) \moves{w}_{\A} (p_m, c_m)$,
      if all $(p_i, c_i) \in Q \times \N$.
\end{itemize}
We abuse notation and write $\Moves{w}$ (resp.\ $\moves{w}$) to mean
$\Moves{w}_{\A}$ (resp.\ $\moves{w}_{\A}$) when it is clear from context.
For $m = 0$, we also use this notation with $w = \eps$.
In addition, for any quasi-run $\pi$ as above,
the sequence of transitions $t_1, \ldots, t_m$
is called a \df{walk} from the state $p_0$ to the state $p_m$.

We will \df{concatenate} runs, quasi-runs, and walks, using
the notation $\pi_1 \cdot \pi_2$ and sometimes dropping the dot.
If $\pi_2$ is a walk and $\pi_1$ is a run, then $\pi_1 \cdot \pi_2$
will also denote a run.
In this and other cases, we will often assume that the counter values in $\pi_2$
are picked or adjusted automatically to match the last configuration of $\pi_1$.

The number $m$ is the \df{length} of $\pi$,
denoted $\length{\pi}$; for a walk, its length is equal to
the length of the sequence.
All concepts and attributes naturally carry over from runs to walks
and vice versa.
Quasi-runs are not used until further sections;
the semantics of OCA is defined just using runs.

A run $(p_0, c_0) \moves{w} (p_m, c_m)$ is called \df{accepting} in \A
if $(p_0, c_0) = (q_0, 0)$ where $q_0$ is the initial state of \A
and $p_m$ is a final state of \A, i.e., $p_m \in F$.
In such a case the word $w$ is \df{accepted} by \A;
the set of all accepted words is called the \df{language} of \A,
denoted $\langop\A$.

\subsection{Regular abstractions}
\label{s:prelim:abstractions}

A \df{nondeterministic finite automaton with $\eps$-transitions
      \textup{(}NFA\textup{)}} is
a one-counter automaton
where all transitions are tests for zero.
Languages of the form $\langop\NFA$, where \NFA is an NFA,
are \df{regular}.
If \A is an OCA,
then $\langop\A$ ---a one-counter language--- is not
necessarily regular.
In what follows, we consider three \df{regular abstractions}
of (one-counter) languages:
downward closures, upward closures, and Parikh-equivalent regular languages.

Let $w, w' \in \Sigma^*$.
We say that the word $w$ is a \df{subword} of the
word $w'$  if $w = a_1 \ldots a_n$ and there are  $x_i \in \Sigma^*$,
$1 \leq i \leq n+1$, such that $w' = x_1 a_1 x_2 a_2 \ldots x_n a_n x_{n+1}$.
We write $w \subword w'$
to indicate this.
For any language $\Language \sset \Sigma^*$,
the \df{upward} and \df{downward closures} of \Language are
the languages
\begin{align*}
\Language\ua & = \{ w' \mid \exists w \in \Language.\ w \subword w'\} \text{\quad and} \\
\Language\da & = \{ w \mid \exists w' \in \Language.\ w \subword w'\},\text{ respectively}.
\end{align*}
Any $w \in \Sigma^*$ defines
a function $\Parikh(w) \colon \Sigma \rightarrow \Nat$,
called the \df{Parikh image} of $w$ (i.e., $\Parikh(w) \in \Nat^\Sigma$
for all $w \in \Sigma^*$).  The value
$\Parikh(w)(a)$ is the number of occurrences of $a$
in $w$. The \df{Parikh image} of a language \Language
is the following subset of $\Nat^\Sigma$:
\begin{equation*}
\Parikh(\Language) = \{ \Parikh(w) \mid w \in \Language \}.
\end{equation*}
In the sequel, we usually identify $\Nat^\Sigma$
and $\Nat^{|\Sigma|}$.

It follows from Higman's lemma~\cite{Hig52} that, for any $\Language \sset \Sigma^*$,
the languages $\Language\ua$ and $\Language\da$ are regular;
since they abstract away some specifics of \Language,
they are regular abstractions of \Language.
For Parikh images, the situation is different:
for example, unary languages $\Language \sset \{ a \}^*$
are essentially unaffected by the Parikh mapping $\Parikh$,
but it is easy to find unary languages that are not even decidable,
let alone regular. However, Parikh's theorem~\cite{Parikh66} states that
if $\Language \sset \Sigma^*$ is a context-free language, then
there exists a regular language $\mathcal R \sset \Sigma^*$ that is
\df{Parikh-equivalent} to \Language, i.e., such that $\Parikh(\Language) = \Parikh(\mathcal R)$.
Hence, such languages $\mathcal R$ are also regular
abstractions of \Language; since all one-counter languages are context-free,
every OCA \A has at least one
regular language that is Parikh-equivalent to $\langop\A$.

\subsection{Convention on OCA}
\label{s:prelim:simple}
To simplify the presentation,
everywhere below we focus our attention on a sublcass of OCA
that we call \df{simple one-counter automata} (\df{simple OCA}).
A simple OCA is defined analogously to OCA and is different
in the following aspects:
(1) there are no zero tests,
(2) there is a unique final state, $F = \{\qfinal\}$,
(3) only runs that start from the configuration $(\qinit, 0)$ and end
    at the configuration $(\qfinal, 0)$ are considered \df{accepting}.
The language of a simple OCA \A, also denoted $\langop\A$,
is the set of words induced by accepting runs.
We now show that this restriction is without loss of generality.

For an OCA $\A=(Q, \AlB, q_0, \tran, F)$ and any $p,q \in Q$,
define a simple OCA $\A^{p,q} \eqdef (Q, \AlB, p, 
\tran^+, \{q\})$ where $\tran^+ \sset \tran$ is the set of all transitions in
$\tran$ that are not tests for zero.

For any simple OCA \A, define a sequence of \df{approximants} $\langop[n]{\A}$, $n \ge 0$:
the language $\langop[n]{\A}$ is the set of all words observed along runs of $\A$ from
$(q_0,n)$ to $(\qfinal,n)$.

\begin{lemma}\label{lem:simple-approx}
Let $K\in\N$ and $\diamondsuit\in \{\upop{},\downop{}, \parikhnoarg\}$ be an abstraction.
Assume that there is a polynomial $g_{\diamondsuit}$ such that
for any OCA \A the following holds: for every $\A^{p,q}$ there is an NFA
$\B_{p,q, \diamondsuit}$
such that $\diamondsuit{\langop{\A^{p,q}}}\subseteq\langop{\B_{p,q, \diamondsuit}}\subseteq\diamondsuit(\langop[K]{\A^{p,q}})$
and $\sizeop{\B_{p,q,\diamondsuit}} \le g_{\diamondsuit}(\sizeA)$.
Then there is a polynomial $f$ such that for any $\A$ there is an
NFA $\B^{\diamondsuit}$ of size at most $f(\sizeA, K)$ with
$\langop{\B^{\diamondsuit}} = \diamondsuit({\langop{\A}})$.
\end{lemma}

\begin{proof}[Proof (sketch)]
We use the following two ideas:
(\textit{i})
to compute the abstraction of $\{w\}$ where $w=w_1\cdot w_2$,
it suffices to concatenate
abstractions of $\{w_1\}$ and $\{w_2\}$;
(\textit{ii})
for any $K \in \N$,
every run of $\A$ can
be described as interleaving of runs below $K$ and
above $K$.
The NFA $\B^{\diamondsuit}$ is constructed as follows:
first encode counter values below $K$ using states,
and then insert NFA $\B_{p,q,\diamondsuit}$ in between $(p,n)$ and $(q,n)$.
\end{proof}

Restriction to simple OCA
is now a consequence of Lemma~\ref{lem:simple-approx} for $K=0$.
\newcommand{\onlyappendix}[1]{}
\newcommand{\onlyappendixpara}[1]{}
\newcommand{\onlymainpaper}[1]{#1}
\newcommand{\onlymainpaperpara}[1]{#1}
\newcommand{\Omittable}[1]{}

\section{Upward and Downward Closures}
\label{app:upward}

\onlyappendix{
We begin by showing that for any OCA $\A$  we can 
effectively construct, in \Ptime, a NFA  that accepts $\lang(\A)\ua$.
This easy construction follows the argument traditionally used
to bound the length of the shortest accepting run in a pushdown automaton.
 We give the details as the ideas used here recur elsewhere.
We use the following notation in what follows: for a run $\rho$ and an 
integer $D$ we write $\rho[D]$ to refer to the quasi-run $\rho'$ obtained
from $\rho$ by replacing the counter value $v$ by $v+D$ in every configuration
along the run.

\begin{restatable}{lemma}{ucbound}
\label{lem:uc-bound}
Let $\A=(Q,\Sigma,\delta,s,F)$ be a OCA and let $w$ be a
word accepted by $\A$. Then there is a word $y \subword w$ in $\lang(\A)$ such
that $y$ is accepted by a run where the value of the counter never exceeds
$|Q|^2+1$.
\end{restatable}
\begin{proof}
We show that for any accepting run $\rho$ reading a word $w$, there
is an accepting run $\rho'$,  reading  a  word $y \subword w$, in which the maximum value of the counter does not exceed $|Q|^2+1$.  We prove this by double
induction on the maximum value of the counter and the number of times this
value is attained during the run $\rho$.

If the maximum value is below $|Q|^2+1$ there is nothing to prove.  Otherwise
let the maximum value $m > |Q|^2+1$ be attained $c$ times along $\rho$.   
We break the run up as a concatenation of subruns 
$\rho = \rho_{0}\rho_{1}\rho_2\rho_3\ldots \rho_{m}\rho'_{m-1}\ldots\rho'_2\rho'_1\rho'_{0}$  where
\begin{enumerate}
\item  $\rho_{0}\rho_1\rho_2\ldots \rho_{m}$ is the shortest prefix of $\rho$ 
after which the counter reaches the value $m$.
\item  $\rho_{0}\rho_1\rho_2\ldots \rho_i$ is the longest prefix
of $\rho_0\rho_1\rho_2\ldots \rho_{m}$ after which the counter value is $i$, $1 \leq i \leq m-1$.
\item  $\rho_{0}\rho_1\rho_2\ldots \rho_{m}.\rho'_{m-1}\ldots \rho'_{i}$, is the shortest prefix of $\rho$ with $\rho_{0}\rho_1\rho_2\ldots \rho_{m}$ as a prefix and after which the counter value is $i$, $0 \leq i \leq m-1$.
\end{enumerate}
Let the configuration reached after the prefix $\rho_{0}\ldots\rho_{i}$  be $(p_i,i)$, for $1 \leq i \leq m$. 
Similarly let the configuration reached after the prefix 
$\rho_{0}\rho_1\rho_2\ldots \rho_{m}.\rho'_{m-1}\ldots\rho'_{i}$ be $(q_i,i)$, for $0 \leq i \leq m-1$.

Now we make two observations: firstly, the value of the counter never falls 
below $i$ during the segment of the run $\rho_{i+1} \ldots
\rho'_{i}$ --- this is by the definition of the $\rho_i$s and $\rho'_i$s. 
 Secondly, there are $i < j$ such that $p_i = p_j$ and
$q_i = q_j$ --- this is because $m \geq |Q|^2 + 1$. 
Together this means that we may shorten the run by deleting the
sequence of transitions corresponding to the segment $\rho_{i+1} \ldots \rho_j$ leading from $(p_i,i)$ to $(p_j,j)$ and the sequence corresponding to the 
segment $\rho'_{j-1} \ldots \rho'_i$ from $(q_j,j)$ to $(q_i,i)$ and still
obtain a valid run of the system. 
That is, $\rho_0\rho_1\ldots\rho_{i}\rho_{j+1}[-d]\rho_{j+2}[-d]\ldots\rho'_{j}[-d]\rho'_{i-1}\ldots\rho'_{0}$
is a valid run,
where $d = j - i$.
Clearly the word accepted by such a run is a subword of $w$, and
further this run has at least one fewer occurrence of the maximal counter value $m$.
The Lemma follows by an application of the induction hypothesis to this
run and using the transitivity of the subword relation.

\end{proof}

The set of words in $\lang(\A)$ accepted along 
runs where the value of the counter does not exceed $|Q|^2+1$ is accepted
by an NFA with $|Q|.(|Q|^2+1)$ states (it keeps the counter values as 
part of the state).  Combining this with the standard construction for
upward closure for NFAs we get
}

\onlymainpaper{
The standard argument used to bound the length of accepting runs
of PDAs (or OCAs) can be adapted easily to show
\begin{restatable}{theorem}{upwardclosure}
\label{thm:upwardclosure}
There is a polynomial-time algorithm that takes as input an OCA $\A
=(Q,\Sigma,\delta,s,F)$ and computes an NFA with $O(\sizeA^3)$ states accepting
$\lang(\A)\ua$.
\end{restatable}
}
\onlyappendix{
\upwardclosure*
}
\onlyappendix{
The construction can be extended to general OCA without any change in the
complexity}

Next we show a polynomial time procedure that constructs an NFA accepting 
the downward closure of the language of any simple OCA. 
  For pushdown automata the construction involves a necessary
exponential blow-up. 
We sketch some observations that lead to our polynomial time construction.

Let $\A = (Q,\Sigma,\delta,s,F) $ be a simple OCA and let $K = |Q|$.
Consider any run $\rho$ of  $\A$ from a configuration
$(p,i)$ to a configuration $(q,j)$.  If the value of the counter increases  
(resp. decreases) by at least $K$ in $\rho$ 
then, it contains a segment that can be \emph{pumped} (or iterated) to 
increase (resp. decrease) the value of the counter. 
\onlyappendix{ If the increase in the value of the
counter in this iterable segment is $k$ then by choosing an appropriate
number of iterations we may increase the value of the counter at the end
of the run by any multiple of this $k$.} Quite clearly, 
the word read along this iterated run will be a superword of word read 
along $\rho$. The following lemmas\onlyappendix{,whose proof is a simplified version
of that of Lemma \ref{lem:uc-bound},} formalize this.

\onlymainpaper{
\begin{restatable}{lemma}{PumpUp}
\label{lem:PumpUp}
Let $(p,i) \moves{x} (q,j)$ with $j-i>K$.  Then, there is an integer
$k > 0$ such that for each $N \geq 0$ there is a run
$(p,i) \moves{w = y_1.(y_2)^{N+1}.y_3} (p',j + N.k)$ with $x = y_1y_2y_3$.
\end{restatable}
}
\onlyappendix{
\PumpUp*
\begin{proof}
Consider the run $(p,i) \moves{x} (q,j)$ and break it up as
\begin{align*}
(p,i) = & (p_i,i) \moves{x_1} (p_{i+1},i+1) \moves{x_2} (p_{i+2},i+2) \ldots \\
&\ldots \moves{x_j} (p_j,j) \moves{x'} (q,j)
\end{align*}
 where the run $(p_i,i)\moves{x_1} \ldots
\moves{x_r} (p_r,r)$ is the shortest prefix after which the value of the counter
 attains the value $r$. Since $j - i > K$ it follows that there are $r,r'$
with $i \leq r < r' \leq j$ such that $p_r = p_{r'}$.
 Clearly one may
iterate the segment of the run from $(p_r,r)$ to $(p_r,r')$ any number
of times, say $N \geq 0$, to get a run $(p,i) \moves{w}  (q,j + (r'-r)N)$.
where $w = x_1\ldots x_r (x_{r+1}\ldots x_{r'})^{N+1} x_{r+1} \ldots x_k$.
Setting $k = r' - r$ yields the lemma.
\end{proof}

An analogous argument shows that if the value of the counter decreases by
at least $K$ in $\rho$ then we may iterate a suitable segment to reduce
the value of the counter by any multiple of $k'$ (where the $k'$ is the
net decrease in the value of the counter along this segment) while reading
a superword. This is formalized as 
}

\onlymainpaper{
\begin{restatable}{lemma}{PumpDown}
\label{lem:PumpDown}
Let $(q',j') \moves{z} (p',i')$ with  $j' - i' > K$.  Then, there is an integer $k' > 0$ such that for every $N \geq 0$ there is
a run $(q',j'+N.k') \moves{w = y_1(y_2)^{N+1}y_3} (p',i')$ with
$z = y_1y_2y_3$. 
\end{restatable}
}
\onlyappendix{
\PumpDown*
\begin{proof}
We break the run into segments as:
\begin{align*}
(q',j') = & (q_{j'},j') \moves{z_{j'-1}} (q_{j'-1},j'-1) \moves{z_{j'-2}} (q_{j'-2},j'-2) \ldots \\
& \ldots \moves{z_{i'}} (q_{i'},i') \moves{z'} (p',i')
\end{align*}
where $(q_{j'},j') \moves{z_{j'-1}} (q_{j'-1},j'-1) \moves{z_{j'-2}}  (q_{j'-2},j'-2) \ldots  $ 
$\moves{z_{t}} (q_{t},t)$ is the shortest prefix after which the value of counter
is $t$. Since $j'-i' > K$ it follows that there are $t,t'$,$j' \geq t > t' \geq
i'$ such that $q_t = q_{t'}$. Then, starting at any configuration 
$(q_{t}, t + (t-t')N)$, $N \in \Nat$ we may iterate the transitions in the run
$(q_t,t) \moves{} (q_{t'},t')$  an additional $N$ times. This yields a run
$(q_t,t+(t-t')N) \moves{z''} (q_{t'},t')$ where
$z'' = (z_{t-1} \ldots z_{t'})^{N+1}$. Observe that $z_{t-1}\ldots z_t'$ is
a subword of $z''$.
Finally, notice that this also means that
$(q',j' + N.(t-t')) \moves{z_1\ldots z_t} (q_t,t+N.(t-t'))\moves{z''}(q_{t'},t') \moves{z_{t'+1}\ldots z_i}(p',i')$. Taking $k'=(t-t')$ completes the
proof.
\end{proof}
}

A consequence of these somewhat innocous lemmas is the following interesting
fact: we can turn a triple consisting of two runs, where the first one 
increases the counter by at least $K$ and the second one decreases the counter
by at least $K$, and a quasi-run that connects them, into a real run provided
we are content to read a superword along the way. 

\onlymainpaper{
\begin{restatable}{lemma}{FreeToOrd}
\label{lem:FreeToOrd}
Let $(p,i) \moves{x} (q,j) \Moves{y} (q',j') \moves{z} (p',i')$, with $j - i > K$ and $j' - i' > K$.  Then, there is a run $(p,i) \moves{w} (p',i')$ such
that $xyz \subword w$.
\end{restatable}
}
\onlyappendix{
\FreeToOrd*
\begin{proof}
Let the lowest value of counter in the entire run be $m$. If $m \geq 0$ then
the given quasi-run is by itself a run and hence there is is nothing to prove.
Let us assume that $m$ is negative.

First we use Lemma \ref{lem:PumpUp}, to get a $k$ and an $x'$ for any $N > 1$ and a run $(p,i) \moves{x'}  (q,j + k.N)$  with $x \subword x'$.
We can then extend this to a run $(p,i) \moves{x'} (q,j+k.N) \Moves{y} (q',j'+k.N)$, by simply choosing $N$ so that $k.N > m$. Then,  we have that the value of the
counter is $\geq 0$ in every configuration of this quasi-run.
 Thus $(p,i) \moves{x'}(q,j+k.N) \moves{y} (q',j'+k.N)$ for any such $N$.
Now, we apply Lemma \ref{lem:PumpDown} to the run $(q',j') \moves{z} (p',i')$ to obtain the $k'$.   We now set our $N$ to be a value divisible by $k'$,
say $k'.I$. Thus,
$(p,i) \moves{x'}(q,j+k.k'.I) \moves{y} (q',j'+k.k'.I)$ and now we may
again use Lemma \ref{lem:PumpDown} to conclude that  $(q',j'+k.k'.I) \moves{z''} (p',i')$ with $x \subword x'$ and $z \subword z''$. This completes
the proof.
\end{proof}
}

Interesting as this may be, this lemma still relies on 
the counter value being  recorded exactly in all the three segments in its
antecedent  and \onlymainpaper{we weaken this next.}\onlyappendix{this is not sufficient.  In the next step, we
weaken this requirement (while imposing
the condition that $q = q'$ and $j = j'$) by releasing the (quasi) middle segment
from this obligation.}

\onlymainpaper{
\begin{restatable}{lemma}{FreeLoop}
\label{lem:FreeLoop}
Let $(p,i) \moves{x} (q,j)$,$(q,j) \moves{z} (p',i')$, with $j - i > K$ and $j' - i' > K$. Let there be a walk from $q$ to $q$ that reads $y$.  Then, there is a run $(p,i) \moves{w} (p',i')$ such that $xyz \subword w$.
\end{restatable}
}
\onlyappendix{\FreeLoop* }
\begin{proof}
Let the given walk result in the quasi-run $(q,j) \Moves{y} (q,j+d)$
(where $d$ is the net effect of the walk on the counter, which may be positive
or negative). Iterating this quasi-run $m$ times yields a quasi-run $(q,j) \Moves{y^m} (q,j + m.d)$, for any $m \geq 0$.
Next, we use Lemma \ref{lem:PumpUp} to find a $k > 0$ such that for each $N > 0$
we have a run $(p,i) \moves{x_N} (q,j + N.k)$ with $x \subword x_N$. Similarly,
we use Lemma \ref{lem:PumpDown} to find a $k' > 0$ such that for each $N' >0$
we have a run $(q,j+N'.k') \moves{y_{N'}} (p',i')$ with $y \subword y_{N'}$.

Now, we pick $m$ and $N$ to be multiples of $k'$ in such a way that 
$N.k + m.d > 0$. This can always be done since $k$ is positive. Thus,
$N.k + m.d = N'.k'$ with $N' > 0$. Now we try and combine the (quasi) runs
$(p,i) \moves{x_N} (q,j+N.k)$, $(q,j+N.k) \Moves{y^m} (q,j+N.k+m.d)$ 
and $(q,j+N'.k') \moves{y_{N'}} (p',i')$ to form a run. We are almost
there, as $j+N.k+m.d = j+N'.k'$. However, it is not guaranteed that
this combined quasi-run is actually a run as the value of the counter 
may turn negative in the segment $(q,j+N.k) \Moves{y^m} (q,j+N.k+m.d)$.
Let $-N''$ be the smallest value attained by the counter in this segment.
Then by replacing $N$ by $N + N''.k'$ and $N'$ by $N' + N''.k$ we can
manufacture a  triple which actually yields a run (since the counter values
are $\geq 0$), completing the proof.
\end{proof}

With this lemma in place we can now explain  how to relax the
usage of counters.

Let us focus on runs that are interesting, that is, those in which the counter
value exceeds $K$ at some point. Any such run may be broken into 3 stages: the
first stage where counter value starts at $0$ and remains strictly below $K+1$,
a second stage where it starts and ends  at $K+1$ and a last stage where the
value begins at $K$ and remains below $K$ and ends at $0$ (the 3 stages are
connected by two transitions, an increment and a decrement).  Suppose, we write
the given accepting run as $(p,0) \moves{w_1} (q,c) \moves{w_2} (r,0)$ where
$(q,c)$ is a configuration in the second stage. If $a \in \Sigma$ is a letter
that may be read in some transition on  some walk from $q$ to $q$. Then,
$w_1aw_2$ is in $\lang(\A)\da$. This is a direct consequence of the above
lemma. It means that in the configurations in the middle stage we may
freely read certain letters without bothering to update the counters.
This turns out to be a crucial step in our construction.  To turn this
relaxation idea into a construction, the following seems a natural.

We make an equivalent,
but expanded version of $\A$. This version has $3$ copies of
the state space: The first copy is used as long as the value of the counter 
stays below $K+1$ and on attaining this value the second copy is entered. 
The second copy simulates $\A$ exactly but nondeterministically chooses to
enter third copy whenever the counter value is moves from $K+1$ to $K$. 
The third copy simulates $\A$ but does not permit the counter value to exceed $K$. For every letter $a$ and state $q$ with a walk from $q$ to $q$ along
which $a$ is read on some transition, we add a self-loop transition to the state
corresponding to $q$ in the second copy that does not affect the counter and reads the letter $a$. This idea  has two deficiencies:
first, it is not clear how to define the transition from the second copy to 
the third copy, as that requires knowing that value of the counter is $K+1$, and
second, this is still an OCA (since the second copy simply faithfully simulates $\A$) and not an NFA.

Suppose we bound the value of the counter by some value $U$ in the second
stage.  Then we can overcome both of these defects and  construct a finite
automaton\onlymainpaper{.  }\onlyappendix{ as follows: The state space of the resulting NFA has stages of the
form $(q,i,j)$ where $j \in \{1,2,3\}$ denotes the stage to which this copy of 
$q$ belongs.  The value $i$ is the value of the counter as maintained
within the state of the NFA. The transitions interconnecting the
stages go from a state of the form $(q,K,1)$ to one of the form  $(q',K+1,2)$ (while simulating
a transition involving an increment) and from a stage of the form $(q,K+1,2)$ to
one of the form $(q',K,3)$ (while simulating a decrement). The value of $i$
is bounded by $K$ if $j \in \{1,3\}$ while it is bounded by $U$ if $j = 2$.
(States of the form $(q,i,2)$ also have self-loop transitions described above.)
}
By using a slight generalization of Lemma \ref{lem:FreeLoop},  which allows for the simultaneous insertion of a number of walks (or by applying the Lemma
iteratively), we can show that any word accepted by such a finite automaton
lies $\lang(\A)\da$.  However, there is no guarantee
that such an automaton will accept every word in $\lang(\A)\da$.
The second crucial point is that we are able to show that if 
$U \geq K^2 + K + 1$ then every word in $\lang(\A)$ is accepted by this
3 stage NFA. We show that for each accepting run $\rho$ in $\A$ there is 
an accepting run in the NFA reading the same word.  The proof is by a
double induction, first on the maximum value attained by the counter and
then on the number of times this value is attained along the run. Clearly, segments of the
run where the value of the counter does not exceed $K^2 + K + 1$ can be
simulated as is. We then show that whenever the counter value exceeds
this number, we can find suitable segments whose net effect on the counter
is $0$ and which can be simulated using the self-loop transitions added to
stage 2 (which do not modify the counters), reducing the maximum value of the counter along the run.

\onlymainpaper{Formalizing this gives:}

\onlyappendix{We now present the formal details. We begin by describing
the NFA $\Au{U}$ where $U \geq K + 1$.
$$\Au{U} = (Q_1\cup Q_2 \cup Q_3, \Sigma,\Delta,\iu{U} ,\Fu{U}\})$$
 where $Q_1 = Q \times \{0 \ldots K\} \times
\{1\}$, $Q_2 = Q \times \{0 \ldots U\} \times \{2\}$ and $Q_3 = Q \times \{0
\ldots K\} \times \{3\}$. We let $\iu{U}=(s,0,1)$ and $\Fu{U} = \{(f,0,1),(f,0,3)\mid f \in F\}$.
The transition relation is the union of the relations $\Delta_1$, $\Delta_2$ and
$\Delta_3$ defined as follows:

\paragraph{Transitions in $\Delta_1$:}
\begin{enumerate}
\setlength{\itemsep}{1pt}
\setlength{\parskip}{0pt}
\item $(q,n,1)  \moves{a} (q',n,1)$  for all $n \in \{0 \ldots K\}$ whenever
$(q,a,\nop,q') \in \delta$.  Simulate an internal move.
\item $(q,n,1) \moves{a} (q',n-1,1)$ for all $n \in \{1 \ldots K\}$ whenever
$(q,a,\dec,q') \in \delta$.  Simulate a decrement.
\item $(q,n,1) \moves{a} (q',n+1,1)$ for all $n \in \{0 \ldots K-1\}$ whenever $(q,a,\inc,q') \in \delta$. Simulate an increment.
\item $(q,K,1) \moves{a} (q',K+1,2)$ whenever $(q,a,\inc,q') \in \delta$. Simulate an increment and shift to second phase.
\end{enumerate}

\paragraph{Transitions in $\Delta_2$:}
\begin{enumerate}
\setlength{\itemsep}{1pt}
\setlength{\parskip}{0pt}
\item $(q,n,2)  \moves{a} (q',n,2)$  for all $n \in \{0 \ldots U\}$ when
$(q,a,\nop,q') \in \delta$.  Simulate an internal move.
\item $(q,n,2) \moves{a} (q',n-1,2)$ for all $n \in \{1 \ldots U\}$ whenever
$(q,a,\dec,q') \in \delta$.  Simulate a decrement.
\item $(q,K+1,2) \moves{a} (q',K,3)$ whenever
$(q,a,\dec,q') \in \delta$.  Simulate a decrement and shift to third phase.
\item $(q,n,2) \moves{a} (q',n+1,2)$ for all $n \in \{0 \ldots U-1\}$ whenever $(q,a,\inc,q') \in \delta$. Simulate an increment move.
\item $(q,n,2) \moves{a} (q,n,2)$ whenever there is a walk from 
$q$ to $q$ on some word $w$ and, $a \subword w$. 
Freely simulate loops.
\label{D3:free}
\end{enumerate}

\paragraph{Transitions in $\Delta_3$:}
\begin{enumerate}
\setlength{\itemsep}{1pt}
\setlength{\parskip}{0pt}
\item $(q,n,3)  \moves{a} (q',n,3)$  for all $n \in \{0 \ldots K\}$ whenever
$(q,a,\nop,q') \in \delta$.  Simulate an internal move.
\item $(q,n,3) \moves{a} (q',n-1,3)$ for all $n \in \{1 \ldots K\}$ whenever
$(q,a,\dec,q') \in \delta$.  Simulate a decrement.
\item $(q,n,3) \moves{a} (q',n+1,3)$ for all $n \in \{0 \ldots K-1\}$ whenever $(q,a,\inc,q') \in \delta$. Simulate an increment move.
\end{enumerate}

The following Lemma, which is easy to prove, states that the first and third phases simulate faithfully any run where the value of the counter is bounded by $K$. 

\begin{lemma}\label{lem:FiThSim} 
\begin{enumerate}
\setlength{\itemsep}{1pt}
\setlength{\parskip}{0pt}
\item If $(q,i,l) \moves{w} (q',j,l)$ in $\Au{U}$ then $(q,i) \moves{w} (q',j)$ in $\A$, for $l \in \{1,3\}$.
\item If $(q,i) \moves{w} (q',j)$ in $\A$ through a run where the value of the counter is
$\leq K$ in all the configurations along the run then $(q,i,l) \moves{w} (q',i,l)$ for $l \in \{1,3\}$.
\end{enumerate}
\end{lemma}

The next Lemma extends this to runs involving the second phase as well. 
All moves other than those simulating unconstrained walks can be simulated by $\A$.
The second phase of $\Au{U}$ can also simulate any run where the counter
is bounded by $U$.  Again the easy proof is omitted.

\begin{lemma}\label{lem:SeSimpSim}
\begin{enumerate}
\item
If  $(q,i,l) \moves{w} (q',j,l')$ is a run of $\Au{U}$ in which
no transition from $\Delta_2$ of type 5 is used then 
$(q,i) \moves{w} (q',j)$ is a run of $\A$.
\item
If $\rho = (q_0,i_0) \moves{a_1} (q_1,i_1) \moves{a_2} \ldots \moves{a_m} (q_m,i_m)$ is a run in $\A$ in which the value of the counter never exceeds $U$ then $\rho' = (q_0,i_0,2) \moves{a_1} (q_1,i_1,2) \moves{a_2} \ldots \moves{a_m} (q_m,i_m,2)$ is a run in $\Au{U}$.
\end{enumerate}
\end{lemma}

Now, we are in a take the first step towards generalizing Lemma 
\ref{lem:FreeLoop} to prove that $\lang(\Au{U}) \subseteq \lang(\A)\da$.

\begin{lemma}\label{lem:SeSim}
Let $(q,i,2) \moves{w} (q',j,2)$ be a run in $\Au{U}$. Then, there is an $N \in \Nat$,  words $x_0, y_0, x_1, y_1, \ldots, x_N$, and integers $n_0, n_1, \ldots, n_{N-1}$ such that, in $\A$ we have:
\begin{enumerate}
\setlength{\itemsep}{1pt}
\setlength{\parskip}{0pt}
\item $w \subword x_0y_0x_1y_1 \ldots x_N$.
\item $(q,i) \Moves{x_0y_0 \ldots x_N} (q',j')$ where $j' = j + n_0 + n_1 \ldots + n_{N-1}$.
\item $(q,i) \Moves{x_0y_0^{m_0}x_1y_1^{m_1} \ldots x_N} (q',j'')$ where
$j'' = j + m_0.n_0 + m_1.n_1 \ldots + m_{N-1}.n_{N-1}$, for any $1 \leq m_0,m_1, \ldots m_{N-1}$.
\end{enumerate}
Note that 2 is just a special case of 3 when $m_r = 1$ for all $r$.
\end{lemma}
\begin{proof}
The run $(q,i,2) \moves{w} (q',j,2)$ in $\Au{U}$ uses only transitions of the types $1,2,4$ and $5$ in $\Delta_2$. Let $N$ be the number of transitions of type $5$
used in the run. We then break up the run as follows:
\begin{align*}
(q,i,2) &\moves{x_0}(p_0,i_0,2)\moves{a_0}(p_0,i_0,2)\moves{x_1}(p_1,i_1,2)
\ldots \\
&\ldots  (p_{N-1},i_{N-1},2) \moves{a_{N-1}}(p_{N-1},i_{N-1},2)\moves{x_N}(q',j,2)
\end{align*}
where the transitions, indicated above, on $a_i$'s are the $N$ moves using  transitions of type $5$ in the run. Let $(p_r,i_r) \Moves{y_r} (p_r,i'_r)$ be a quasi-run with
$a_r \subword y_r$ and let $n_r = i'_r - i_r$.  Clearly $w \subword x_0y_0x_1y_1\ldots x_N$.

It is quite easy to show by induction on $r$, $0 \leq r < N$, by replacing
moves of types $1,2$ and $4$ by the corresponding moves in $\A$ and moves
of type $5$ by the iterations of the quasi-runs identified above that:

\begin{align*}
(q,i) & \moves{x_0} (p_0,i_0) \Moves{y_0^{m_0}}  (p_0,i_0 + m_0.n_0) \Moves{x_1} (p_1,i_1+m_0.n_0)\\
& \Moves{y_1^{m_1}}  (p_1,i_1 + m_0.n_0 + m_1.n_1) \\
& \ldots \\
& \moves{x_r} (p_r,i_r+m_0.n_0 \ldots + m_{r-1}.n_{r-1}) \\
& \Moves{y_r^{m_r}} (p_r,i_r+m_0.n_0 \ldots m_r.n_r)\\
& \moves{x_{r+1}} (p_{r+1},i_{r+1}+m_0n_0\ldots m_r.n_r)
\end{align*}

and with $r = N-1$ we have the desired result.
\end{proof}

Now, we use an argument that  generalizes  Lemma \ref{lem:FreeLoop} in order
to show that:

\begin{lemma}\label{lem:FAtoCA}
Let $w$ be any word accepted by the automaton $\Au{U}$. Then, there
is a word $w' \in \A$ such that $w \subword w'$.
\end{lemma}
\begin{proof}
If states in $Q_2$ are not visited in the accepting run 
of $\Au{U}$ on $w$ then we can use Lemma \ref{lem:FiThSim} to conclude that
$w \in \A$.
Otherwise, we break up the run of $\Au{U}$ on $w$ into three parts
as follows:
\begin{align*}
(s,0,1) &\moves{w_1} (p,K,1) \moves{a_1} (q,K+1,2)\moves{w_2} (r,K+1,2) \\
&\moves{a_2} (t,K,3) \moves{w_3} (f,0,3)
\end{align*}
Using Lemma \ref{lem:FiThSim} we have $(s,0) \moves{w_1} (p,K)$
and $(t,K) \moves{w_3} (f,0)$. We then apply Lemmas \ref{lem:PumpUp} and
\ref{lem:PumpDown} to these two segments respectively to identify $k$ and $k'$. Next we use Lemma \ref{lem:SeSim} to identify  the positive integer $N$, 
integers  $n_0, n_1, \ldots n_{N-1}$ and the quasi-run
$$
 (q,K+1) \Moves{x_0y_0 \ldots x_N} (r,K+1+n_0+n_1 \ldots +n_{N-1})
$$ 
with $w_2 \subword x_0y_0x_1y_1\ldots x_{N-1}y_{N-1}x_N$.
We identify numbers $m, m_0, m_1, \ldots, m_{N-1}$, all $\geq 1$,
 such that $(m-1).k + m_0.n_0 + \ldots m_{N-1}.n_{N-1} = k'.m'$ for some $m' \geq
0$.   
By taking $m-1$ and each $m_i$ to be some multiple  of $k'$ we get the sum
$(m-1).k + m_0.n_0 + \ldots m_{N-1}n_{N-1}$ to be a multiple of $k'$, however
this multiple may not be positive. Since $k > 0$, by choosing $m-1$ to
be a sufficiently large multiple of $k'$ we can ensure that $m' \geq 0$. 
Using these numbers we construct the quasi-run 
\begin{align*}
(q,&K+1 + (m-1).k) \Moves{x_0y_0^{m_0}x_1y_1^{m_1}\ldots x_N} \\
& (r,K+1+(m-1).k + m_0n_0+\ldots m_{N-1}n_N) 
\end{align*}
with $K+1+(m-1).k + m_0n_0+\ldots m_{N-1}n_N = K+1+k'.m'$.
Let $l$ be the lowest value attained in this quasi-run. If $l \geq 0$ then
$$(q,K+1 + (m-1).k) \moves{x_0y_0^{m_0}x_1y_1^{m_1}\ldots x_N} (r,K+1+k'.m')$$
 and using
Lemma \ref{lem:PumpUp} and \ref{lem:PumpDown} we get
\begin{align*}
(s,0) &\moves{w} (p,K+(m-1).k) \\
&\moves{a_1} (q,K+1+(m-1).k)\\
&\moves{x_0y_0^{m_0}x_1y_1^{m_1}\ldots x_N} (r,K+1+k'.m') \\
&\moves{a_2} (t,K+k'.m') \\
&\moves{z} (f,0,3)
\end{align*}
with $w_1 \subword w$, $w_2 \subword x_0y_0^{m_0}x_1y_1^{m_1}\ldots x_N$ and 
$w_3 \subword z$ as required.

Suppose $l < 0$. Then, we let $I$ be a positive integer such that $I.k + l > 0$ and $I = k'.m''$ (i.e. $I$ is divisible by $k'$)
 which must exist since $k > 0$.  Then
$$
(q,K+1 + (m-1).k + I.k) \Moves{x_0y_0^{m_0}x_1y_1^{m_1}\ldots x_N} (r,K+1+I.k+k'.m')
$$
is a quasi-run in which the counter values are always $\geq 0$ and is thus
a run.
Once again, we may use Lemmas \ref{lem:PumpUp}
and \ref{lem:PumpDown} (since $I.k$ is a multiple of $k'$)  to get
\begin{align*}
(s,0) & \moves{w} (p,K+(m-1).k+I.k) \\
& \moves{a_1} (q,K+1+(m-1).k+I.k)\\
& \moves{x_0y_0^{m_0}x_1y_1^{m_1}\ldots x_N} (r,K+1+k'.m')\\
& \moves{a_2} (t,K+k'.m' + I.k) \\
&\moves{z} (f,0,3)
\end{align*}
with $w_1 \subword w$, $w_2 \subword  x_0y_0^{m_0}x_1y_1^{m_1}\ldots x_N$ and 
$w_3 \subword z$. This completes the proof of the Lemma.

\end{proof}

Next, we show that if $U \geq K^2 + K + 1$ then $\lang(\Au) \subseteq \lang(\Au{U})$.

\begin{lemma}\label{lem:CAtoFA} Let $U \geq K^2 + K + 1$.
Let $w$ be any word in $\lang(\A)$. Then, $w$ is also accepted by $\Au{U}$.
\end{lemma}
\begin{proof}
The proof is accomplished by examining runs of the from 
$(s,0) \moves{w} (f,0)$ and showing that such a run may be
\emph{simulated} by $\Au{U}$ transition by transition in a manner
to be described below.
Any run $\rho = (s,0) \smoves{w} (f,0)$ can be broken up into parts
as follow:
$$
(s,0) \moves{x} (h,j) \moves{y} (h',j') \moves{z} (f,0)
$$
where, $\rho_1 = (s,0) \moves{x} (h,j)$ is the longest prefix where the
counter value does not exceed $K$, $\rho_3 = (h',j') \moves{z} (f,0)$,
 is the longest suffix, of what is left after removing $\rho_1$,  in which 
the value of the counter does not exceed $K$,
and $\rho_2 = (h,j) \moves{y} (h',j')$ is what lies in between.
We note that using Lemma \ref{lem:FiThSim} we can conclude that there
are runs $(s,0,1) \moves{x} (h,j,1)$ and $(h',j',3) \moves{z} (f,0,3)$.
Further, observe that if value of the counter never exceeds $K$ then $\rho_2$ and
$\rho_3$ are empty, $x = w$, $h=f$ and $j = 0$. In this case, using
Lemma \ref{lem:FiThSim}, there is a (accepting)  run 
$(s,0,1) \moves{w} (f,0,1)$.  

If the value of the counter exceeds $K$ then $j = j' = K$ and  by Lemma  \ref{lem:FiThSim}, $(s,0,1) \moves{x} (h,K,1)$, $(h',K,3) \moves{z} (f,0,3)$ and $\rho_2$ 
is non-empty.
Further suppose that, $\rho_2$, when written out as a sequence of transitions is of the form
\begin{align*}
\rho_2 = &(h,K) \moves{a} (p,K+1) = (p_0,i_0) \moves{a_1} (p_1,i_1) \\
& \moves{a_2} (p_2,i_2) \ldots \moves{a_n} (p_n,i_n) = (q,K+1) \moves{b} (h',K)
\end{align*}

We will show by double induction on the maximum value of the counter value 
attained in the run $\rho_2$ and the number of times the maximum is attained
that there is a run
\begin{align*}
\rho'_2 = & (h,K,1) \moves{a} (p,K+1,2) = (p'_0,i'_0,2) \moves{a_1} (p'_1,i'_1,2) \\
 & \moves{a_2} (p'_2,i'_2,2) \ldots \moves{a_n} (p'_n,i'_n,2) = (q,K+1,2) \moves{b} (h',K,3)
\end{align*}
such that for all $i$, $0 \leq i < n$, 
\begin{enumerate}
\item either $p_i = p'_i$ and  $p_{i+1} = p'_{i+1}$ and the $i$th transition (on $a_{i+1}$)  is of type $1,2$ or $4$, 
\item  or  $p'_i = p'_{i+1}$, $i'_i = i'_{i+1}$,  $p'_i\Moves{} p_i$ and 
$p_{i+1} \Moves{} p'_i$ so that the $i$th transition (on $a_{i+1}$) 
is a transition of  type $5$.
\end{enumerate}

For the basis, notice that if the maximum value attained is $\leq K^2 + K + 1$ then, by Lemma \ref{lem:SeSimpSim}, there is a run of $\Au{U}$ that simulates
$\rho_2$ such that item 1 above is satisfied for all $i$. 

Now, suppose the maximum value attained along the run is $m > K^2 + K + 1$.
We proceed along the lines of the proof of Lemma \ref{lem:uc-bound}. 
We first break up the run $\rho_2$  as  
\begin{align*}
 (h,K)&\moves{a}  (p_0,K+1) = (q_{K+1},K+1) \moves{y_{K+2}} (q_{K+2},K+2) 
\\
&\moves{y_{K+3}} (q_{K+3},K+3) \ldots \moves{y_m} (q_m,m) \\
&\moves{y'_{m-1}} (q'_{m-1},m-1) \ldots \moves{y'_{K+1}} (q'_{K+1},K+1) \\
&\moves{z} (q,K+1) \moves{b} (h',K)
\end{align*}
where 
\begin{itemize}
\item The prefix upto $(q_m,m)$, henceforth referred to as $\sigma_m$, is the shortest prefix after which the counter value is $m$.
\item The prefix upto $(q_i,i)$, $K+1 \leq i < m$ is the longest prefix of 
$\sigma_m$ after which the value of the counter is $i$.
\item The prefix upto $(q'_i,i)$, $K+1 \leq i < m$ is the shortest prefix of
$\rho_2$ with $\sigma_m$ as a prefix after which the counter value is $i$.
\end{itemize}
By construction, the value of the counter in the segment of the run from
$(q_i,i) \moves{} \ldots \moves{} (q'_i,i)$ never falls below $i$. 
Further, by simple counting, there are $i,j$ with $K+1 \leq i < j \leq m$ such that
$q_i = q_j$ and $q'_i = q'_j$. Thus, by deleting the segment of the runs
from $(q_i,i)$ to $(q_i,j)$ and $(q'_j,j)$ to $(q'_j,i)$  we get a shorter 
run $\rho_d$ which  looks like
\begin{align*}
(h,K)& \moves{a}  (p_0,K+1) = (q_{K+1},K+1) \ldots \\
&\moves{y_i} (q_i,i) \moves{y_{j+1}} (q_{j+1},i+1) \ldots \\
&\moves{y_m} (q_m,m-j+i) \moves{y'_{m-1}} \ldots \\
&\moves{y'_j} (q'_j,i) \moves{y'_{i-1}} (q'_{i-1},i-1) \ldots (q'_{K+1},K+1)\\
& \moves{z} (q,K+1) \moves{b} (h',K)
\end{align*}
This run reaches the value $m$ at least one time fewer than $\rho_2$ and
thus we may apply the induction hypothesis to conclude the existence of
a run $\rho'_d$ of  $\Au{U}$ that simulates this run move for move satisfying the properties
indicated in the induction hypothesis.  Let this run be:
\begin{align*}
(h,K,1)&\moves{a}  (r_0,K+1,2)  \ldots \\
&\moves{y_i} (r_i,c_i,2)) \moves{y_{j+1}} (r_{j+1},c_{j+1},2) \ldots \\
&\moves{y_m} (r_m,c_m,2) \moves{y'_{m-1}} \ldots \\
&\moves{y'_j} (r'_j,c'_j,2) \moves{y'_{i-1}} (r'_{i-1},c'_{i-1},2) \ldots (r'_{K+1},c'_{K+1},2) \\
&\moves{z} (r',K+1,2) \moves{b} (h',K) 
\end{align*}

Now, if $(p_l,i_l) \moves{a_{l+1}} (p_{l+1},i_{l+1})$ was a transition 
in $\rho_2$ in the part of the run from $(q_i,i)$ to $(q_i,j)$ then, 
$q_i \Moves{} p_l$, $p_l \Moves{a_{l+1}} q_i$ and $p_{l+1} \Moves{} q_i$.
Now, either $r_i = q_i$ or $r_i\Moves{}q_i$,
and $q_{j+1}\Moves{}r_i$ and $(q_i,a_{j+1},op, q_{j+1})$ is a transition
for some $op$. In the both cases clearly $r_i \Moves{a_{l+1}} r_i$.
Thus every such deleted transition can be simulated by a transition of the form
$(r_i,c_i,2) \moves{a_{l+1}} (r_i,c_i,2)$.  

A similar argument shows that
every transition of the form $(p_l,i_l) \moves{a_{l+1}} (p_{l+1},i_{l+1})$
 deleted in the segment $(q'_j,j)$ to $(q'_j,i)$ can be simulated by
$(r'_j,c'_j,2) \moves{a_{l+1}} (r'_j,c'_j,2)$. Thus we can extend the
run $\rho'_d$ to a run $\rho'_2$ that simulates $\rho_2$ fulfilling the
requirements of the induction hypothesis. This completes the proof
of this lemma.
\end{proof}

Notice that the size of the state space of $\Au{U}$ is $K . ( K^2 + K + 1)$ when
$U = K^2 + K + 1$.  Since downward closures of NFAs can be constructed by just
adding additional ($\eps$) transitions, Lemmas \ref{lem:FAtoCA} and 
\ref{lem:CAtoFA} imply that:
}

\onlymainpaper{
\begin{restatable}{theorem}{DownwardClosure}
\label{thm:DownwardClosure}
There is a polynomial-time algorithm that takes as input a simple  OCA $\A
=(Q,\Sigma,\delta,s,F)$ and computes an NFA with $O(\sizeA^3)$ states accepting
$\lang(\A)\da$.
\end{restatable}
}
\onlyappendix{\DownwardClosure* }

\onlyappendix{
A closer look reveals that the complexity remains the same even for general OCA. We only need one copy of $\Au{U}$.
}

\section{Parikh image: Fixed alphabet}
\label{s:parikh-bounded}

The result of this section is the following theorem.

\begin{theorem}
\label{conf:th:parikh-bounded}
For any fixed alphabet \AlB there is a polynomial-time algorithm that,
given as input a one-counter automaton over \AlB with $n$~states, computes
a Parikh-equivalent NFA.
\end{theorem}

Note that in Theorem~\ref{conf:th:parikh-bounded} the size of \AlB is fixed.
The theorem implies, in particular, that
any one-counter automaton over \AlB with $n$~states
has a Parikh-equivalent NFA of size $\poly_{\AlB}(n)$,
where $\poly_{\AlB}$ is a polynomial of
degree bounded by $f(\dimension)$ for some computable function $f$.

We now provide the intuition behind the proof
of Theorem~\ref{conf:th:parikh-bounded}.
Our key technical contribution
is capturing the structure of the Parikh image of the language $\langop\A$.

Recall that a set $A \sset \N^\dimension$ is called \df{linear}
if it is of the form
$
\Lin{b}{P} \eqdef \{ b + \lambda_1 p_1 + \ldots + \lambda_r p_r \mid
        \lambda_1, \ldots, \lambda_r \in \N,
        p_1, \ldots, p_r \in P \}
$
for some vector $b \in \N^\dimension$ and some finite set $P \sset \N^\dimension$;
this vector $b$ is called the \df{base} and vectors $p \in P$ \df{periods}.
A set $S \sset \N^d$ is called \df{semilinear} if it is a finite union of linear sets,
$S = \cup_{i \in I} \Lin{b_i}{P_i}$.
Semilinear sets were introduced in the 1960s
and have since received a lot of attention in formal language theory
and its applications to verification. They are precisely the sets
definable in Presburger arithmetic, the first-order theory of natural
numbers with addition.
Intuitively, semilinear sets are a multi-dimensional analogue
of ultimately periodic sets in \N.
For our purposes, of most importance is the following way
of stating the Parikh theorem~\cite{Parikh66}:
the Parikh image of any \emph{context-free} language is a semilinear set;
in particular,
so is the Parikh image of any one-counter language,~$\parikh{\langop\A}$.

Our proof of Theorem~\ref{conf:th:parikh-bounded}
captures the periodic structure of this set~$\parikh{\langop\A}$.
More precisely, we prove
polynomial upper bounds on the number of linear sets in the semilinear
representation of~$\parikh{\langop\A}$ and on the magnitude of periods and base vectors.
Since converting such a semilinear representation into a polynomial-size NFA is easy,
these bounds (subsection~\ref{s:parikh-bounded:repr}) entail
the existence of an appropriate NFA.
After this, we show how to compute, in time polynomial in \sizeA,
this semilinear representation from \A
(subsection~\ref{s:parikh-bounded:algo}).

\subsection{Semilinear representation of $\parikh{\langop\A}$}
\label{s:parikh-bounded:repr}

We now explain where the periodic structure of the set $\parikh{\langop\A}$ comes from.
Consider an individual accepting run $\pi$ and assume that one
can factorize it as $\pi = \rho \cdot \sigma \cdot \tau$ so that
for any $k \ge 0$ the run $\rho \cdot \sigma^k \cdot \tau$ is also
accepting. Values $k > 0$ correspond to pumping the run ``up'',
and the value $k = 0$ corresponds to ``unpumping'' the infix $\sigma$.
If we apply this ``unpumping'' to $\pi$ several times (each time taking
a new appropriate factorization of shorter and shorter runs),
then the remaining part eventually becomes small (short). Its Parikh image will
be a base vector, and the Parikh images of different infixes $\sigma$ will be
period vectors of a linear set in the semilinear representation.

However, this strategy faces several obstacles.
First, the overall reasoning should work
on the level of the whole automaton, as opposed to
individual runs; this means that we need to rely on a form
of a pumping lemma to factorize long runs appropriately.
The pumping lemma for one-counter languages involves, instead of individual
infixes~$\sigma$, their pairs $(\sigma_1, \sigma_2)$, so that the entire
run factorizes as
$\pi = \rho \cdot \sigma_1 \cdot \upsilon \cdot \sigma_2 \cdot \tau$, and
runs
$\pi = \rho \cdot \sigma_1^k \cdot \upsilon \cdot \sigma_2^k \cdot \tau$ are
accepting for all $k \ge 0$.
We incorporate this into our argument, talking about \emph{split runs}
(Definition~\ref{conf:def:bi-walk}).
Here and below, for any run $\zeta$,
$\effect{\zeta}$ denotes the \df{effect} of $\zeta$
on the counter, i.e., the difference between the final
and initial counter value along $\zeta$.

\begin{definition}[split run]
\label{conf:def:bi-walk}
A \emph{\biwalk} is a pair of runs
$\bi{\sigma_1}{\sigma_2}$ such that
$\effect{\sigma_1}\geq 0$ and
$\effect{\sigma_2} \le 0$.
\end{definition}

Second and most importantly, it is crucial for the periodic structure
of the set that individual ``pumpings'' and ``unpumpings'' can be performed
independently. That is, suppose we can insert a copy of a sub-run $\sigma$
into $\pi$, as above; also suppose we can remove from $\pi$ some other
sub-run $\sigma'$. What we need to ensure is that, after removing $\sigma'$
from $\pi$, the obtained run $\pi'$ will have the property that we can
still insert a $\sigma$ in it, as in the original run $\pi$.
In general, of course, this does not have to be the case: even
for finite-state machines, removal of loops can lead to removal of individual
control states, which can, in turn, prevent the insertion of other loops
(in our case the automaton also has a counter that must always stay non-negative).
To deal with this phenomenon, we introduce the concept of ``availability''
(Definition~\ref{conf:def:available}). Essentially, a ``pumpable'' part of the
run---i.e., a ``split walk''---%
defines a \emph{direction} (Definition~\ref{conf:def:direction});
we say that a direction is \emph{available}
at the run $\pi$ if it is possible to insert its copy into $\pi$.
Thus, when doing ``unpumping'', we need to make sure that the set of available
directions does not change: we call such unpumpings \emph{safe}
(Definition~\ref{conf:def:unpumping}).
We show that long accepting runs can always be safely unpumped
(Lemma~\ref{conf:l:unpump-summary}), which will lead us
(Lemma~\ref{conf:l:parikh}) to the semilinear
representation that we sketched at the beginning of this subsection.

We now describe the formalism behind our arguments.

\begin{definition}[direction]
\label{conf:def:direction}
A \emph{direction} is a pair of walks $\alpha$ and $\beta$,
denoted $\direction = \dr\alpha\beta$, such that:
\begin{itemize}
\item $\alpha$ begins and ends in the same control state,
\item $\beta$ begins and ends in the same control state,
\item $0 < |\alpha| + |\beta| < \polyvi$,
\item $0 \le \effect{\alpha} \le \polyiii$,
\item $\effect{\alpha} + \effect{\beta} = 0$,
\item if $\effect{\alpha} = 0$, then either 
      $|\alpha| = 0$ or $|\beta| = 0$.
\end{itemize}
The direction is \df{of the first kind} if $\effect{\alpha} = 0$,
and \df{of the second kind} otherwise.
\end{definition}

One can think of a direction as
a pair of short loops with zero total effect on the counter.
Pairs of words induced by these loops are sometimes known as iterative pairs.
Directions of the first kind are essentially just individual loops;
in a direction of the second kind,
the first loop increases and the second loop decreases the counter value
(this restriction, however, only concerns the total effects of $\alpha$ and $\beta$;
 i.e., proper prefixes of $\alpha$ can have negative effects
 and proper prefixes of $\beta$ positive effects).
The condition that $\effect\alpha \le \polyiii$ is a pure technicality
and is only put in to make some auxiliary statements in the proof more laconic;
in contrast, the upper bound $|\alpha| + |\beta| < \polyvi$ is crucial
(although the choice of larger polynomial is, of course, possible,
 at the expense of an increase in the obtained upper bound on the size of NFA).

\begin{definition}[availability of directions]
\label{conf:def:available}
Suppose $\pi$ is an accepting run.
A direction $\direction = \dr\alpha\beta$ is \df{available} at $\pi$
if there exists a factorization $\pi=\pi_1 \cdot \pi_2 \cdot \pi_3$ such that
$\pi'=\pi_1 \cdot \alpha \pi_2 \beta \cdot \pi_3$
is also an accepting run.
We write
$\pi+\direction$ to refer to $\pi'$.
\end{definition}

Note that for a particular run $\pi$ there can be more than one
factorization of $\pi$ into $\pi_1, \pi_2, \pi_3$
such that $\pi_1 \cdot \alpha  \pi_2  \beta  \cdot \pi_3$ is an accepting run.
In such cases the direction $\direction$ can be introduced at different points
inside $\pi$.
We only use the notation $\pi+\direction$ to refer to a single
run $\pi'$ obtained in this way,
without specifying a particular factorization of $\pi$.

Denote by $\avail{\pi}$ the set of all directions available at $\pi$.

\begin{lemma}[monotonicity of availability]
\label{conf:l:more}
If $\pi$ is an accepting run of an OCA
and $\direction$ is a direction available at $\pi$, then
$\avail{\pi} \subseteq \avail{\pi+\direction}$.
\end{lemma}

\begin{definition}[unpumping]
\label{conf:def:unpumping}
A run $\pi'$ \df{can be unpumped} if
there exist a run $\pi$ and a direction $\direction$
such that $\pi' = \pi + \direction$.
If additionally $\avail{\pi'} = \avail\pi$, then we say
that $\pi'$ \df{can be safely unpumped}.
\end{definition}

Note that $\avail{\pi'}$
is always a superset of $\avail\pi$ by Lemma~\ref{conf:l:more}.
The key part of our argument is the proof that,
indeed, every long run can be
unpumped in a safe way:

\begin{lemma}[safe unpumping lemma]
\label{conf:l:unpump-summary}
Every accepting run $\pi'$ of $\A$ 
of length greater than $\polyMinUnpLen$
can be safely unpumped.
\end{lemma}

\begin{proof}[Proof (sketch)]
We consider two cases, depending on whether the \df{height} (largest counter value)
of $\pi'$ exceeds a certain polynomial in~$n$.
The strategy of the proof is the same for both cases
(although the details are somewhat different).
We first show that sufficiently large parts (runs or split runs) of~$\pi'$
can always be unpumped (as in standard pumping arguments).
We notice that for such an unpumping to be \emph{unsafe},
it is necessary that the part contain a configuration whose removal
shrinks the set of available directions---a reason for non-safety;
this \df{important} configuration cannot appear anywhere else in~$\pi'$.
We prove that the total number of important configurations is at most
$\poly(n)$. As a result, if we divide the run~$\pi'$ into
sufficiently many sufficiently large parts, at least one of the parts
will contain no important configurations and, therefore,
can be unpumped safely.
\end{proof}

Lemma~\ref{conf:l:unpump-summary} ensures
that we faithfully represent
the semilinear structure of the Parikh image of the entire language
when we take Parikh images of short runs as base vectors
and Parikh images of available directions as period vectors
in the semilinear representation:

\begin{lemma}
\label{conf:l:parikh}
For any OCA $\A$,
it holds that
\begin{equation}
\label{conf:eq:parikh}
\parikh{\langop\A} =
\bigcup_{\text{\textup{$\length{\pi}\leq s(n)$}}}
\Lin{\parikh\pi}{\parikh{\avail\pi}},
\end{equation}
where the union is taken over all accepting runs of $\A$
of length at most $s(n) = \polyMinUnpLen$.
\end{lemma}

Finally, to keep the representation of the Parikh image small,
we rely on a Carath\'eodory-style argument ensuring that
the number of the linear sets in the semilinear representation
needs to grow only polynomially in the size of the original automaton,
while the sets of period vectors is also kept small.
For this (and only this) part of the argument, we need the alphabet size,
$\dimension$, to be fixed.

\subsection{Computing the semilinear representation}
\label{s:parikh-bounded:algo}

Lemma~\ref{conf:l:parikh} suggests the following algorithm
for computing the semilinear representation of $\parikh{\langop\A}$.
Enumerate all potential Parikh images $v$ of small accepting runs $\pi$ of \A
and all potential Parikh images of directions.
For every $v$ and for every tuple of $r \le \dimension$~vectors $v_1, \ldots, v_r$
that could be Parikh images of directions in \A,
check if \A indeed has an accepting run $\pi$
and directions $d_1, \ldots, d_r$ available at $\pi$ such that
$\parikh\pi = v$ and $\parikh d_i = v_i$ for all~$i$
(Parikh images of runs and directions are defined
 as Parikh images of words induced by them).
Whenever the answer is yes, take a linear set $\Lin{v}{\{v_1, \ldots, v_r\}}$
into the semilinear representation of $\parikh{\langop\A}$.
Terminate when all tuples $(v, v_1, \ldots, v_r)$ have been considered
for all $r \le \dimension$.

We now explain why this algorithm works in polynomial time.
Recall that the size of the alphabet, $\dimension$, is fixed.
Note that by Definition~\ref{conf:def:direction}
the total length of runs $\alpha_i$ and $\beta_i$ in a direction
$d_i = \dr{\alpha_i}{\beta_i}$ is at most polynomial in~$n$;
similarly, equation~\eqref{conf:eq:parikh} in Lemma~\ref{conf:l:parikh}
only refers to accepting runs $\pi$ of polynomial length.
Therefore, all the components of
all potential Parikh images $v$ and $v_1, \ldots, v_r$
are upper-bounded by polynomials in~$n$ of fixed degree.
The number of such vectors in $\N^\dimension$ is polynomial,
and so is the number of appropriate
tuples $(v, v_1, \ldots, v_r)$, $r \le \dimension$.
It now remains to argue that each tuple can be processed
in polynomial time.

\begin{lemma}
\label{conf:l:tuple}
For every \AlB
there is a polynomial-time algorithm that,
given a simple OCA \A over \AlB
and vectors $v, v_1, \ldots$, $v_r \in \N^\AlB$, $0 \le r \le \dimension$,
with all numbers written in unary,
decides if \A has an accepting run $\pi$
and directions $d_1, \ldots, d_r \in \avail\pi$
with $\parikh\pi = v$ and $\parikh{d_i} = v_i$ for all~$i$.
\end{lemma}

Lemma~\ref{conf:l:tuple} is based on the following
building block:

\begin{lemma}
\label{conf:l:pair}
For every \AlB
there is a polynomial-time algorithm that,
given a simple OCA \A over \AlB,
two configurations $(q_1, c_1)$ and $(q_2, c_2)$
and a vector $v \in \N^\AlB$ with all numbers written in unary,
decides if \A has a run $\pi = (q_1, c_1) \moves{} (q_2, c_2)$
with $\parikh\pi = v$.
\end{lemma}

The algorithm of Lemma~\ref{conf:l:pair} solves a version
of the Parikh membership problem for OCA.
It constructs a multi-dimensional table
by dynamic programming:
for \emph{all} pairs of configurations $(q'_1, c'_1)$, $(q'_2, c'_2)$
with bounded $c'_1, c'_2$ and all vectors $v' \in \N^\AlB$
of appropriate size, it keeps the information whether \A has
a run $(q'_1, c'_1) \moves{} (q'_2, c'_2)$ with Parikh image~$v'$.

This completes our description of how to
compute, from an OCA \A, a semilinear representation of $\parikh{\langop\A}$.
Transforming this representation into an NFA is a simple exercise.

\section{Parikh image: Unbounded alphabet }
\label{app:parikh}

In this section we describe an algorithm to construct an NFA Parikh-equivalent 
to an OCA \A without assumptions $|\Sigma|$. The NFA
has at most $O(|\Sigma|K^{O(\log K)})$ states where $K = \sizeA$,
a significant improvement over
$O(2^{\poly(K,|\Sigma|)})$ for PDA.

We establish this result in two steps. In the first step, we show that we can focus our attention on
computing Parikh-images of words recognized along \emph{reversal bounded} runs.
A reversal in a run occurs when the OCA switches to incrementing the counter
after a non-empty sequence of decrements (and internal moves) or when it
switches to decrementing the counter after a non-empty sequence of
increments (and internal moves).  For a number $R$, a run is $R$ reversal
bounded, if the number of reversals along the run is $\leq R$. Let
us use $\lang_R(\A)$ to denote the set of words accepted by $\A$ along
runs with at most $R$ reversals. 

We construct a new polynomial size simple OCA from $\A$ 
and show that we can restrict our attention to runs with at most $R$ reversals of this OCA, where $R$ is a polynomial in $K$.
In the second step, from any simple OCA $\A$ with $K$ states
  and any integer $R$ we construct an NFA of size $O(K^{O(log(R))})$ 
whose Parikh image is $\lang_R(\A)$.
 Combination of the
two steps gives a $O(K^{O(log K)})$ construction. 

\subsection{Reversal bounding}\onlymainpaper{\label{sec:reversal-bounding}}

We establish that, up to Parikh-image, it suffices to consider 
runs with $2K^2+K$ reversals. We use two constructions: one that eliminates \emph{large} reversals (think of a waveform)  and another that eliminates \emph{small} reversals (think of the noise on a noisy waveform).  
For the large reversals, the idea used is the following: we can reorder
the transitions used along a run, hence preserving Parikh-image, to turn it
into one with few large reversals (a noisy waveform with few reversals).
The key idea used is to
move each simple cycle at state $q$ with a positive (resp. negative) 
effect on the counter to the first (resp. last) occurrence of the state
along the run.  To eliminate the smaller reversals (noise), the idea is to maintain the changes to the counter in the state and transfer it only when necessary
to the counter to avoid unnecessary reversals.

\includegraphics[scale =0.3, angle = 270 ]{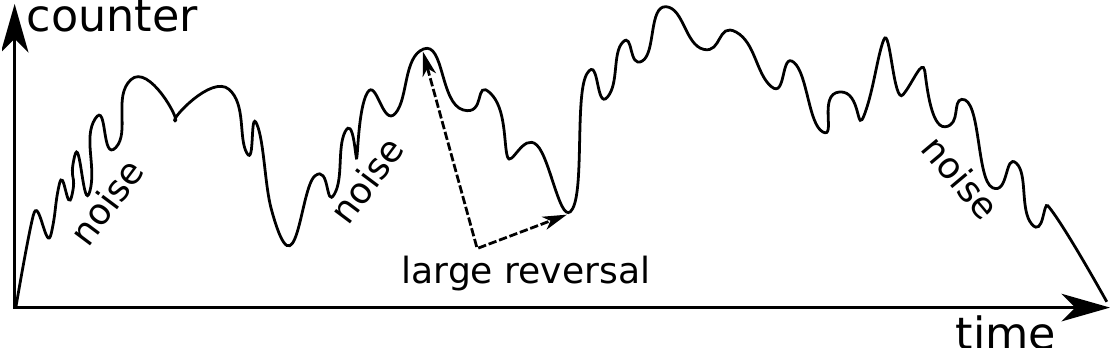}

Consider a run of $\A$ starting at a configuration $(p,c)$ and ending at
some configuration $(q,d)$ such that the value of the counter $e$ in any
intermediate configuration satisfies $c-D \leq e \leq c+D$ (where
$D$ is some positive integer). We refer to such a run as an \emph{$D$-band run}.
 
Reversals along such a run are not important and we get rid of them by 
maintaining the (bounded) changes to the counter within the state.

 We construct a simple OCA $\A[D]$ as follows: its states are $Q \cup Q_1 \cup Q_2$ 
where $Q_1 = Q \times [-D,D]$ and $Q_2 = [-D,D]\times Q$. All transitions
of $\A$ are transitions of $\A[D]$ as well and thus using $Q$ it can simulate
any run of $\A$ faithfully. 
From any state $q \in Q$ the automaton may move nondeterministically to $(q,0)$  in $Q_1$. 
The states in $Q_1$ are used to simulate $D$-band
runs of $\A$ without altering the counter and  by keeping track of the net change to the
counter in the second component of the state. 
\onlyappendix{For instance, 
consider the $D$-band run of $\A$ described above:  $\A[D]$ can move
from $(p,c)$ to $((p,0),c)$ then simulate the run of $\A$ to $(q,d)$ 
to reach $((q,d-c),c)$. At this point it needs to transfer the net effect
back to the counter (by altering it appropriately). The states $Q_2$ are
used to perform this role.}  From a state $(q,j)$ in $Q_1$, $\A[D]$ is
allowed to nondeterministically move to $(j,q)$ indicating that it will
now transfer the (positive or negative) value $j$ to the counter. 
After completing the transfer it reaches a state $(0,q)$ from where it can 
enter the state $q$ via an internal move to continue the simulation of $\A$.

\onlyappendix{The nice feature of this simulated run via $Q_1$ and $Q_2$ is }
\onlymainpaper{Observe} that 
there are no reversals in the simulation and it involves only 
increments (if $d > c$) or only decrements (if $d < c$).
\onlyappendix{
We now formalize the automaton $\A[D]$ and its properties.
The simple OCA $\A[D] = (Q_D,\Sigma,\delta_D,s,F)$ is defined
as follows: 
$$Q_D = Q \cup (Q \times \{-D,\ldots,D\}) \cup  (\{-D,\ldots,D\} \times Q) $$
and $\delta_D$ is defined as follows:
\begin{enumerate}
\setlength{\itemsep}{1pt}
\setlength{\parskip}{0pt}
\item $\delta \subseteq \delta_R$ . Simulate runs of $\A$.
\item $(q,\eps,\nop,(q,0)) \in \delta_D$. Begin a summary phase.
\item $((q,j),a,\nop,(q',j)) \in \delta_D$, if $(q,a,\nop,q') \in \delta$. Simulate an internal move.
\item $((q,j),a,\nop,(q',j+1)) \in \delta_D$, if $(q,a,\inc,q') \in \delta$. 
Simulate an increment.
\item $((q,j),a,\nop,(q',j-1)) \in \delta_D$, if  $(q,a,\dec,q') \in \delta$. Simulate a decrement. 
\item $((q,j),\eps,\nop,(j,q)) \in \delta_D$. Finish summary run. 
\item $((j,q),\eps,\inc,(j-1,q)) \in \delta_D$, if $j > 0$. Transfer a positive effect.
\item $((j,q),\eps,\dec,(j+1,q)) \in \delta_D$, if $j < 0$. Transfer a negative effect.
\item $((0,q),\eps,\nop,q) \in \delta_D$.  Transfer control back to copy of $\A$.
\end{enumerate}

The following Lemma is the first of a sequence that relate $\A$ and $\A[D]$. 
\begin{lemma}\label{lem:MR1}
\begin{enumerate}
\setlength{\itemsep}{1pt}
\setlength{\parskip}{0pt}
\item For any $p,q \in Q$ and any $c,d \in \Nat$, if $(p,c) \moves{w} (q,d)$ in $\A$ then  $(p,c) \moves{w} (q,d)$ in $\A[D]$.
\item For any $p,q \in Q$ and any $c,d \in \Nat$ if $(p,c) \moves{w} (q,d)$ in 
$\A[D]$ then $(p,c+D) \moves{w} (q,d+D)$ in $\A$.
In particular, if $(p,0) \moves{w} (q,0)$ in $\A[D]$ then  $(p,D) \moves{w} (q,D)$ in $\A$.
\end{enumerate}
\end{lemma}
\begin{proof}
The first statement simply follows from the fact that $\delta \subseteq \delta_D$.

Let $\rho = (p,c) \moves{w} (q,d)$ be a run in $\A[D]$.
The second statement is proved by induction on the number of transitions of type 2 taken along $\rho$ (i.e. the number of summary simulations used in $\rho$).
If this number is $0$ then all the transitions used are of type 1 thus
$\rho$ is a run in $\A$ and thus $\rho[D]$ satisfies the requirements of the Lemma. 

Otherwise, let $\rho$ must be of the form
\begin{align*}
 \rho = (p,c) &\moves{w_1} (p_1,c_1) \moves{\eps,\nop} ((p_1,0),c_1) \moves{w_2} ((0,q_1),d_1)\\
& \moves{\eps,\nop} (q_1,d_1) \moves{w_3} (q,d)
\end{align*}
where we have identified the first occurrence of the transition of type 2 and
as well as the first occurrence of a transition of type 9.
Now, by the induction hypothesis, we have runs $(p,c+D) \moves{w_1} (p_1,c_1+D)$ and $(q_1,d_1+D) \moves{w_3} (q,d+D)$  in $\A$.  

From the definition of $\delta_D$, run $((p_1,0),c_1) \moves{w_2} ((0,q_1),d_1)$ must be of the form 
$$ ((p_1,0),c_1) \moves{w_2} ((p_2,c_2),c_1) \moves{\eps} ((c_2,p_2),c_1) \moves{\eps} ((0,p_2),c_1+c_2)$$
with $p_2 = q_1$ and $d_1 = c_1 + c_2$ and where the run $((p_1,0),c_1) \moves{w_2} ((p_2,c_2),c_1)$  involves only transitions of the form 3, 4 or 5.

\paragraph{Claim:} Let $((g,0),e) \moves{x} ((h,i),e)$ be a run in $\A[D]$ using only
transitions of type 3, 4 or 5. Then  $(g,e) \moves{x} (h,e+i)$ in $\A$  for any $e \geq D$.
\paragraph{Proof of Claim:} By induction on the length of the run. The base
case is trivial. For the inductive case, suppose 
$((g,0),e) \moves{x'} ((h',i'),e) \moves{a} ((h,i),e)$ and by the induction
hypothesis $(g,e) \moves{x'} (h',e+i')$ for any $e \geq D$.  Now, if the last
transition  is an internal transition then, $(h',a,\nop,h) \in \delta$ and 
$i = i'$.  Thus  $(h',e+i) \moves{a} (h,e+i)$ in $\A$.
If the last transition is an increment then $(h',a,\inc,h) \in \delta$ and 
$i = i'+1$. Thus, once again we have $(h',e+i) \moves{a} (h,e+i)$ in $\A$.
Finally, if the last transition is a decrement transition then, 
$(h',a,\dec,h) \in \delta$. Then, $i = i'-1$ and $i \geq -D$. Thus, 
$e+i \geq 0$  and thus $(h',e+i') \moves{a} (h,e+i)$ in $\A$, completing
the proof of the claim.\hfill $\square$

Since, $c_1 + D \geq D$, we may apply the claim to conclude that 
$(p_1,c_1 + D) \moves{w_2} (p_2 = q_1,c_1 + D  + c_2 = d_1 + D)$ in $\A$. This
completes the proof of the Lemma.
\end{proof}

Next we show that $\A[D]$ can simulate any $D$-band run 
without reversals.

\begin{lemma}\label{lem:MR3}
Let $(p,c) \moves{w} (q,d)$ be an $D$-band run in $\A$. 
Then, there is a run $(p,c) \moves{w} (q,d)$ in $\A[D]$ in which the 
counter value is never decremented if $c \leq d$ and never 
incremented if $c \geq d$.
\end{lemma}
\begin{proof}
The idea is to simply simulate the run as a summary run in $\A[D]$.
Let the given run be
$$(p,c) = (p_0,c_0) \moves{a_1} (p_1,c_1) \moves{a_2} (p_2,c_2) \ldots \moves{a_n}(p_n,c_n) = (q,d)$$
Then,  it is easy to check that the following is a run in $\A[D]$
\begin{align*}
(p_0,c_0) & \moves{\eps} ((p_0,0),c_0)\moves{a_1} ((p_1,c_1-c_0),c_0) \moves{a_2} \ldots \\
&\moves{a_n} ((p_n,c_n - c_0),c_0) \moves{\eps} ((c_n - c_0,p_n),c_0)
\end{align*}
It is also easy to verify that for any configuration with $((j,p),e)$ 
with $e+j \geq 0$, $((j,p),e) \moves{\eps} (p,e+j)$ is a run in $\A[D]$
consisting only of increments if $j > 0$ and consisting only of decrements
if $j < 0$.  Since $c_n \geq 0$, $(c_n - c_0) + c_0 \geq 0$ and the result follows.
\end{proof}

}
Actually this automaton $\A[D]$ does even better.  Concatenation
of $D$-band runs is often not an $D$-band run but the idea of
reversal free simulation extends to certain concatenations. 
We say that a run $(p_0,c_0) \moves{w} (p_n,c_n)$ is an increasing (resp. decreasing)
\emph{iterated $D$-band run} if it can be decomposed as 
$$(p_0,c_0)\moves{w_1}(p_1,c_1) \moves{w_2} \ldots (p_{n-1},c_{n-1})\moves{w_n} (p_n,c_n)$$ where each $(p_i,c_i)\moves{w_{i+1}} (p_{i+1},c_{i+1})$ is an $D$-band run and $c_i \leq c_{i+1}$ (resp. $c_i \geq c_{i+1})$. 
We say it is an iterated $D$-band run if it is an increasing or decreasing
iterated $D$-band run.

\onlymainpaper{
\begin{restatable}{lemma}{LBand}
\label{lem:LBand}
Let $(p,c) \moves{w} (q,d)$ be an increasing (resp. decreasing)
$D$-band run in $\A$. 
Then, there is a run $(p,c) \moves{w} (q,d)$ in $\A[D]$ along which the
counter value is never decremented (resp. incremented).
\end{restatable}
}
\onlyappendixpara{
\LBand*
\begin{proof}
Simulate each $\rho_i$ by a run that only increments (resp. decrements)
the counter using Lemma \ref{lem:MR3}.
\end{proof}
}

While\onlyappendixpara{, as a consequence of item 1 of Lemma \ref{lem:MR1},}
\onlymainpaperpara{clearly } $\lang(\A) \subseteq \lang(\A[D])$, the
converse is not in general true as along a run of $\A[D]$ the real value
of the counter, i.e. the current value of the counter plus the offset available 
in the state,  may be negative, leading to runs that are not simulations of
runs of $\A$.    The trick\onlyappendixpara{, as elaborated in item 2 of  Lemma \ref{lem:MR1},}  that helps us
get around this is to relate runs of $\A[D]$ to $\A$ 
with a shift in counter values.
\onlymainpaper{The following lemma summarizes this shifted relationship:
\begin{lemma}\label{lem:MR2-L}
Let $p,q \in Q$.  
 If $(p,0) \moves{w} (q,0)$ is a run in $\A[D]$ then $(p,D) \moves{w} (q,D)$
is a run in $\A$.
 \end{lemma}

With these two lemmas we have enough information about $\A[D]$ and its
relationship with $\A$. }
We need a bit more terminology to proceed.

We say that a run of is an $D_\leq$ run (resp. $D_\geq$ run)  if the value of 
the counter is bounded from above (resp. below) by $D$ in every configuration
encountered  along the run. 
We say that a run of $\A$ is an $D_>$ run  if it is 
of the form $(p,D) \moves{w} (q,D)$, it has at least 3 configurations and
the value of the counter at every configuration other than the first and
last is $> D$. Consider any run from a configuration $(p,0)$ to $(q,0)$ in 
$\A$. Once we identify the maximal $D_>$ sub-runs, what is left is a collection of $D_\leq$ subruns.
 
Let $\rho = (p,c) \moves{w} (q,d)$ be a run of $\A$ with $c,d \leq D$. If $\rho$ is a $D_\leq$
run then its \emph{$D$-decomposition} is $\rho$. Otherwise, its $D$-decomposition is given by a sequence of runs $\rho_0,\rho'_0,\rho_1,\rho'_1 \ldots \rho'_{n-1},\rho_n$ with $\rho = \rho_0.\rho'_0.\rho_1.\rho'_1 \ldots. \rho'_{n-1}.\rho_n$, where each $\rho_i$ is a $D_\leq$ run and each $\rho'_i$ is a $D_>$ run
for $0 \leq i \leq n$. Notice, that some of the $\rho_i$'s may be trivial.
Since the $D_>$ subruns are uniquely identified this definition is 
unambiguous.  We refer to the $\rho'_i$'s (resp. $\rho_i$s) as the $D_>$ (resp. $D_\leq$)  components of $\rho$.

Observe that the $D_\leq$ runs of $\A$ can be easily simulated by an NFA.
Thus we may focus on transforming the $D_>$ runs, preserving just the Parikh-image, into a suitable form. For $D,M \in \Nat$, we say that a $D_>$ run $\rho$ is
a \emph{$(D,M)$-good run} (think noisy waveform with few reversals) if there are runs $\sigma_1,\sigma_2 \ldots, \sigma_n,\sigma_{n+1}$ and iterated $D$-band runs $\rho_1,\rho_2, \ldots, \rho_n$ such that $\rho = \sigma_1\rho_1\sigma_2\rho_2 \ldots \sigma_n\rho_n\sigma_{n+1}$ and 
$|\sigma_1|+ \ldots + |\sigma_{n+1}| + 2.n \leq M$. 
Using Lemma \ref{lem:LBand} and that it is a $D_>$ run we show 
\onlymainpaper{
\begin{restatable}{lemma}{LMGood}
\label{lem:LMGood}
Let $(p,D)\moves{w}(q,D)$ be an $(D,M)$-good run of $\A$.  Then,  there is a
run $(p,0) \moves{w} (q,0)$ in $\A[D]$ with atmost $M$ reversals.
\end{restatable}
}
\onlyappendix{
\LMGood*
\begin{proof}
Let the given run be $\rho$. We first shift down $\rho$ to $\rho[-D]$ to
obtain a run from $(p,0)$ to $(q,0)$, which is possible since $\rho$ is $D_>$ run.
We then transform each of the iterated $D$-band runs using Lemma \ref{lem:LBand} so that there are no reversals in the transformed runs. Thus all reversals
occur inside the $\sigma_i[-D]$'s or at the boundary and this gives us the
bound required by the lemma.
\end{proof}
}

So far we have not used the fact that we can ignore the ordering of the
letters read along a run (since we are only interested in the Parikh-image
of $\lang(\A)$). 
We show that for any run $\rho$ of $\A$ we may find another 
run $\rho'$ of $\A$, that is equivalent up to Parikh-image, 
such that every $D_>$ component in the $D$-decomposition of $\rho'$ is $(D,M)$-good, where $M$ and $D$ are polynomially related to $K$.

We fix $D = K$ in what follows. We take $M = 2K^2 + K$ for reasons that will
become clear soon.  We focus our attention on some $D_>$ component $\xi$ of $\rho$ which is not $(D,M)$-good.  Let $X \subseteq Q$ be the set of states of $Q$
that occur in at least two different configurations along $\xi$. For each of the
states in $X$ we identify the configuration along $\xi$ where it occurs 
for the very first time and the configuration where it occurs for the last
time.   There are at most $2|X| (\leq 2K)$ such configurations and these
decompose the run $\xi$ into a concatenation of $2|X| + 1 (\leq 2K+1)$ runs
 $\xi = \xi_1.\xi_2 \ldots \xi_m$ where $\xi_i, 1 < i < m$ is a segment
connecting two such configurations.  Now, suppose one of these $\xi_i$'s has
length $K$ or more. Then it must contain a sub-run $(p,c) \moves{} (p,d)$ 
with at most $K$ moves, for some $p \in X$ (so, this is necessarily a 
$K$-band run).  If $d-c \geq 0$ (resp. $d-c < 0$), then we \emph{transfer} this
subrun from its current position to the first occurrence (resp. last occurrence)
of $p$ in the run. This still leaves a valid run $\xi'$ since $\xi$  begins 
with a $K$ as counter value and $|\xi_i| \leq K$. 
Moreover $\xi$ and $\xi'$ are equivalent upto Parikh-image.

If this $\xi'$ continues to be a $K_>$ run then we again examine if it is
$(K,M)$-good and otherwise, repeat the operation described above. As we
proceed, we continue to accumulate a increasing iterated $K$-band run at the first occurrence of each state and decreasing iterated $K$-band run at the 
last occurrence of each state. We also  ensure that in each 
iteration we only pick a segment that does NOT appear in these  $2|X|$
iterated $K$-bands.  Thus, these iterations will stop when either the segments
outside the iterated $K$-bands are all of length $< K$ and we cannot find
any suitable segment to transfer, or when the resulting run is no longer a
$K_>$ run. In the first case, we must necessarily have a $(K,2K^2+K)$-good run.
In the latter case, the resulting run decomposes as usual in $K_\leq$ and
$K_>$  components, and we have that every $K_>$ component is strictly shorter
than $\xi$\onlymainpaperpara{, allowing us to use an inductive argument to prove the following:}\onlyappendix{. We formalize the ideas sketched above now.

We begin by proving a Lemma which says that any $K_>$ run $\rho$ can 
be transformed into a Parikh-equivalent run $\xi$ which is either a
$K_>$ run which is $(K,2K^2+K)$-good or has a $K$-decomposition each of whose $K_>$ components are strictly shorter than $\rho$.

\begin{lemma}\label{lem:Parikh-main}
Let $\rho = (p,K) \moves{w} (q,K)$ be a $K_>$ run in $\A$. Then, 
there is a run $\xi = (p,K) \moves{w'} (q,K)$ in $\A$, with $|\xi| =
|\rho|$, $\Parikh(w) = \Parikh(w')$  such that one of the following holds:
\begin{enumerate}
\setlength{\itemsep}{1pt}
\setlength{\parskip}{0pt}
\item $\xi$ is not a $K_>$ run. Thus, all $K_>$-components in the $K$-decomposition of $\xi$ are strictly shorter than $\xi$ (and hence $\rho$).
\item $\xi$ is a $K_>$ run and $\xi = \sigma_1\rho_1\ldots \sigma_n\rho_n$ where $n \leq 2K+1$, 
each $\rho_i$ is an iterated $K$-band run and $|\sigma_i| \leq K$ for
each $i$. Thus, $\xi$ is  $(K,2K^2+K)$-good.
\end{enumerate}
\end{lemma}
\begin{proof}
Let $\rho = (p_0,c_0) \moves{a_1} (p_1,c_1) \ldots \moves{a_m} (p_m,c_m)$.
Let $X \subseteq Q$ be the set of controls states that repeat in the run $\rho$.
We identify the first and last occurrences of each state $q \in X$ along the
run $\rho$, and there are $n = 2.|X| \leq 2K$ such positions. We then decompose the
run $\rho$ as follows
\begin{align*}
(p_0,c_0) = &(q_0,e_0)\sigma_1(q_1,e_1)\sigma_2(q_2,e_2)\ldots \\
 & \ldots (q_{n-1},e_{n-1})\sigma_n (q_n,e_n)\sigma_{n+1} (q_{n+1},e_{n+1}) = (q,d)
\end{align*}
where configurations $(q_1,e_1), (q_2,e_2) \ldots (q_n,e_n)$ 
correspond to the first or last occurrence of states from $X$.  
We introduce, for reasons that will become clear in the following, 
an  empty iterated $K$-band run $\rho_i$ following each $(q_i,e_i)$ to get 
\begin{align*}
(q_0,e_0)\sigma_1 & (q_1,e_1)\rho_1(q_1,e_1)\sigma_2(q_2,e_2)\rho_2(q_2,e_2)\ldots \\
& \ldots  (q_{n-1},e_{n-1})\sigma_n (q_n,e_n)\rho_n (q_n,e_n) \sigma_{n+1} (q_{n+1},e_{n+1})
\end{align*}

Let $\xi_0$ be $\rho$ with the decomposition as written above. We shall now
construct a sequence of runs $\xi_i$, $i \geq 0$, from $(p,K)$ to $(q,K)$,
 maintaining the length and  the Parikh image as an invariant, that is, $\Parikh(\xi_i) = \Parikh(\xi_{i+1})$ and $|\xi_i| = |\rho|$. In each step, starting with 
a $K_>$ run $\xi_i$, we shall reduce the length of one of the $\sigma_i$ by
some $1 \leq l \leq K$ and increase the length of  one iterated $K$-band 
runs $\rho_j$ by $l$ to obtain a run $\xi_{i+1}$,
maintaining the invariant.  If this resulting run is not a $K_>$ run then
it has a $K$-decomposition in which every $K_>$ component is shorter than
$\xi_i$ (and hence $\rho$), thus satisfying item 1 of the Lemma
completing the proof.  Otherwise, after sufficient number of iterations
of this step, we will be left satisfying item 2 of the Lemma. 
Let  the $K_>$ run $\xi_i$ be given by
\begin{align*}
(q_0,e_0)&\sigma^i_1(q_1,e^i_1)\rho^i_1(q_1,f^i_1)\sigma^i_2(q_2,e^i_2)\rho^i_2(q_2,f^i_2)\ldots\\
&\ldots (q_{n-1},e^i_{n-1})\sigma^i_n (q_n,e^i_n)\rho^i_n (q_n,f^i_n) \sigma^i_{n+1} (q_{n+1},e^i_{n+1})
\end{align*}
where each $\rho^i_j$ is an iterated $K$-band run. If the length of 
$|\sigma^i_j| \leq K$ for each $j \leq n+1$  then, we have already fulfilled
item  2 of the Lemma, completing the proof. Otherwise, there is some $j$
such that $|\sigma^i_j| \geq K$. Therefore, we may decompose $\sigma^i_j$ 
as $$ (q_{j-1},f^i_{j-1})\chi_1(r,g)\chi_2(r,g')\chi_3 (q_j,e^i_j)$$
where $(r,g)\chi_2(r,g')$ is a run of length $\leq K$  and $r \in X$. 
There are two cases to consider, depending on whether $g' - g \geq 0$ or
$g'-g < 0$.

Let $(q_B,e^i_B)$ and $(q_E,f^i_E)$ be the first and last occurrences of $r$ 
in $\xi_i$. We will remove the segment of the run given by $\chi_2$ and
add it to $\rho^i_B$ if $g' \geq g$ and add it to $\rho^i_E$ otherwise. 
First of all, since the first and last occurrences of $r$ are distinct,
the $\rho^i_B$ will remain a increasing iterated $K$-band run while $\rho^i_E$
remains a decreasing iterated $K$-band run.
Clearly, such a transformation preserves the Parikh image of the word
read along the run.
It is easy to check that, since $\xi_i$ is a $K_>$ run and the length of 
$\chi_2$ is bounded by $K$, the resulting sequence $\xi_{i+1}$ 
(after adjusting the counter values) will be a valid run, since the counter
stays $\geq 0$. However, it may no longer be a $K_>$ run.  (This may happen, if 
$e^i_B < g$ and there is a prefix of $\chi_2$ whose net effect is to 
reduce the counter by more than $e^i_B - K$.)   However, in this
case we may set $\xi_{i+1}$ is a run from $(p,K)$ to $(q,K)$, with the
same length as $\xi_i$ and thus every $K_>$ component in its $K$-decomposition
is necessarily shorter than $\xi_i$. Thus, it satisfies 
item 1 of the Lemma.  

If $\xi_{i+1}$ remains a $K_>$ run then we observe that 
$|\sigma^i_1\ldots \sigma^i_n| > |\sigma^{i+1}_1\ldots \sigma^{i+1}_n|$
and this guarantees the termination of this construction with a $\xi$
satisfying one of the requirements of the Lemma.
\end{proof}

Starting with any run, we plan to apply Lemma \ref{lem:Parikh-main}, 
to the $K_>$ components, preserving Parikh-image, till we reach one
in which every $K_>$ component satisfies item 2 of Lemma \ref{lem:Parikh-main}.
To establish the correctness of such an argument we need the following 
Lemma.

\begin{lemma}\label{lem:Kdecuniq}
Let $\rho = (p,0)\moves{w}(q,0)$ be a run. 
If $\rho = \rho_1 (r,K) \rho_2$ then every $K_>$ component in the decomposition of $\rho$ is a $K_>$ component of $\rho_1$ or $\rho_2$ and vice versa.
In particular, if $\rho = \rho_1 (r,K) \rho_2 (r',K)\rho_3$ 
then, $K_>$ components of the $K$-decomposition of $\rho$ are exactly the
$K_>$ components of the runs $\rho_1$, $\rho_2$ or $\rho_3$.
\end{lemma}
\begin{proof}
By the definition of $K_>$ run and $K$ decompositions.
\end{proof}

We can now combine Lemmas \ref{lem:Kdecuniq} and  \ref{lem:Parikh-main}
to obtain:
}

\onlymainpaper{
\begin{restatable}{lemma}{boundedrevNew}
\label{lem:boundrevNew}
Let $\rho = (p,0) \moves{w} (q,0)$ be any run in $\A$. Then, there
is a run $\rho' = (p,0) \moves{w'} (q,0)$ of $\A$ with 
$\Parikh(w) = \Parikh(w')$ such that every $K_>$ component $\xi$ in the 
canonical decomposition of $\rho'$ is $(K,2K^2+K)$-good.
\end{restatable}
}
\onlyappendix{
\boundedrevNew*
\begin{proof}
The proof is by double induction, on the length of the longest $K_>$ 
component in $\rho$ that is not $(K,2K^2+K)$-good and  the 
number of components of this size that violate it. 
For the basis case, observe that any $K_>$ component whose length 
is bounded by $2K^2 + K$ is necessarily $(K,2K^2+K)$-good.

For the inductive case, we pick a $K_>$ component $\xi$ in $\rho$ of 
maximum size apply Lemma \ref{lem:Parikh-main} and replace $\xi$ by 
$\xi'$ to get $\rho'$. If $\xi'$ is $(K,2K^2+K)$-good 
we have reduced the number of components of the maximum size that are
not $(K,2K^2+2)$-good in $\rho'$.
Otherwise, $\xi'$ satisfies item 2 of Lemma \ref{lem:Parikh-main} and thus
by Lemma \ref{lem:Kdecuniq} the number of $K_>$ components in  the decomposition of $\rho'$ of the size of $\xi$ that are not $(K,2K^2+K)$-good is one less
than that in $\rho$. This completes the inductive argument.
\end{proof}

\begin{remark} 
We note that the above proof can be formulated slightly
differently. The reason we work with $D_>$-runs (which is also incorporated
in the definition of $(D,M)$-good runs) is that such runs of $\A$ 
from, say $(p,D)$ to $(q,D)$, can be simulated faithfully by $\A[D]$ from 
$(p,0)$ to $(q,0)$ while runs of $\A[D]$ from $(p,0)$ to $(q,0)$ can
be simulated by $\A$ along runs from $(p,D)$ to $(q,D)$.  Since, segments
of any run of $\A$ where the counter value lies below $D$ can be easily
simulated by an NFA, $D$-decompositions and the above inductive argument
come naturally.  

A slightly different argument is the following: 
If we begin with $(2D)_\geq$ run of $\A$ then, we can carry out the above
inductive argument without bothering about whether it remains a $(2D)_>$
run at each step, for it is guaranteed to remain a $D_\geq$ run. 
Then, instead of the automaton $A[D]$ we use a slight variant $B[D]$
which does the following: it simulates $A[D]$ and at every point where
it reverses from decrements to increments it \emph{verifies} that
the counter is at least $D$ by decrementing the counter $D$ times and
then incrementing the counter $D$ times. Then, we can relate $\A$ and
$\B[D]$ without a level shift as follows: for any $C \geq 2D$,  
there is a run $(p,C) \moves{w} (q,C)$ in $\A$  if and only if there is a
run $(p,C) \moves{w} (q,C)$ in $\B[D]$. 

We have preferred the argument where the automaton is simpler to define
 and it does not track reversals. The two proofs are of similar difficulty.
\end{remark}
}

\onlymainpaper{
Let $\Asp^{pq}$, $p,q \in Q$, be NFA  Parikh-equivalent to $\lang_{2K^2+K}$
$(\A[K]^{p,q})$ where $A[K]^{pq}$ is $\A[K]$ with $p$ as the only initial
and $q$ as the only final state.
As a consequence of Lemmas \ref{lem:boundrevNew},\ref{lem:LMGood} and \ref{lem:MR2-L}, we can obtain an NFA $\Asp$ such that
$\Parikh(\lang(\Asp)) = \Parikh(\lang(\A))$.  }
\onlyappendix{
Let $\A^K$  be the NFA simulating the simple OCA $\A$ when the 
counter values lie in the range $[0,K]$, by maintaining the counter values
in its local state.  This automaton is of size $O(K^2)$.  
Now, suppose for each pair of states $p,q \in Q$ we have
an NFA $\Asp^{pq}$ which is Parikh-equivalent to $\lang_{2K^2+K}(\A[K]^{p,q})$,
where $\A[K]^{p,q}$ is the automaton $\A[K]$ with $p$ as the only 
initial state and $q$ as the only accepting state.
We combine these automata (there are $K^2$ of them)  with $\A^K$  by taking
their disjoint union and adding the following additional (internal) transitions.
We add transitions from the states of the from $(p,K)$ of $\A^K$, for $p \in Q$
 to the initial state of state of all the $\Asp^{pq}$, $q \in Q$. Similarly,
from the accepting states of $\Asp^{pq}$ we add internal transitions
to the state $(q,K)$ in $\A^K$. Finally we deem $(s,0)$ to be the only
initial state and $(f,0)$ to be the only final state of the combined automaton. We call this NFA $\Asp$.

} 
\onlyappendix{The next lemma confirms that $\Asp$ is the automaton we are after.

\begin{lemma}\label{lem:Correct}
$\Parikh(\lang(\Asp)) = \Parikh(\lang(\A))$
\end{lemma}
\begin{proof}
Let $\rho$ be an accepting run of $\A$ on a word $w$. We first apply Lemma
\ref{lem:boundrevNew} to construct a run $\rho'$ on a $w'$, with $\Parikh(w) =
\Parikh(w')$, in whose $K$-decomposition, every $K_>$
component is $(K,2K^2+K)$-good. Let
$\chi = (p,K) \moves{x} (q,K)$ be such a component.  
Then, by Lemma \ref{lem:LMGood}, there is a run $\chi':(p,0) \moves{x} (q,0)$ 
in $\A[K]$ with at most $2K^2+K$ reversals.
Thus, there is a $x' \in \lang(\Asp^{pq})$ with $\Parikh(x) = \Parikh(x')$. 
If $(s,0) \moves{x} (p,K)$ is a
$K_\leq$ component of $\rho'$ then $(s,0) \moves{x} (p,K)$ in $\A^K$.  
If $(p,K) \moves{x} (q,K)$ is a $K_\leq$ component of $\rho'$ then 
$(p,K) \moves{x} (q,K)$ in $\A^K$ and finally if $(p,K) \moves{x} (f,0)$ 
is a $K_\leq$ component of $\rho'$ then $(p,K) \moves{x} (f,0)$ in $\A^K$. 
Putting these together we get a run from $(s,0)$ to $(f,0)$ in $\Asp$ 
on a word Parikh-equivalent to $w'$ and hence $w$.

For the converse, any word in $\lang(\Asp)$ is of the form 
$x.u_1.v_1.u_2.v_2 \ldots u_nv_n.y$ where $(s,0)\moves{x}(p_1,K)$ in $\A^K$,
$(q_n,K)$ $\moves{y}(f,0)$ in $\A^K$, $u_i \in \lang(\Asp^{p_iq_i})$ and $(q_i,K) \moves{v_i} (p_{i+1},K)$ in $\A^K$, for each $1 \leq i \leq n$. 
By construction, there is a run $(s,0)\moves{x}(p_1,K)$ in $\A$ and
$(q_n,K)\moves{y}(f,0)$ in $\A$. Further for each $i$, there is a run $(q_i,K) \moves{v_i} (p_{i+1},K)$ in $\A$ as well.  
Since $u_i \in \lang(\Asp^{p_iq_i})$, by construction of $\Asp^{p_iq_i}$, there is a 
run $(p_i,0) \moves{u'_i} (q_i,0)$ (with a bound on the number of reversals, but
that is not important here) in $\A[K]$ with $\Parikh(u_i) = \Parikh(u'_i)$. 
But then, by the second part of Lemma \ref{lem:MR1},  there is a run $(p_i,K) \moves{u'_i} (q_i,K)$ in $\A$. Thus we can put together these different segments now to obtain an accepting run in $\A$ on the word $x.u'_1.v_1.u'_2.v_2\ldots u'_nv_n$.  Thus, $\Parikh(\lang(\Asp)) \subseteq \Parikh(\lang(\A))$, completing the proof of the Lemma.
\end{proof} } 
The number of states in the automaton $\Asp$ is $\sum_{p,q \in Q}|\Asp^{pq}| + K^2$. 
What remains to be settled is the size of the automata $\Asp^{pq}$. 
\onlyappendix{ That is,
computing an upper bound on the size of an NFA which is Parikh-equivalent to
 the language of words accepted by an OCA (in this case $\A[K]$) along runs
with at most $R$ (in this case $K^2 + K$) reversals.} This problem is solved
in the next subsection and the solution (see Lemma \ref{lem:rbpdas}) implies
that that the size of $\Asp^{pq}$ is bounded by $O(|\Sigma|K^{O(log K)})$. Thus we have

\onlymainpaper{
\begin{restatable}{theorem}{parikhunbounded}
\label{thm:parikh-unbounded}
There is an algorithm, which given an OCA with $K$ states and alphabet $\Sigma$,
  constructs a Parikh-equivalent NFA with $O(|\Sigma|.K^{O(\log K)})$ states.
\end{restatable}
}
\onlyappendix{\parikhunbounded*}

\subsection{Parikh image under reversal bounds}
\label{app:rbpdas}

Here we show that, for an OCA $\A$, with $K$ states and whose alphabet is $\Sigma$, and any $R \in \Nat$, an NFA Parikh-equivalent to
$\lang_R(\A)$ can be constructed with size $O(|\Sigma|.K^{O(log K)})$. As a matter of
fact, this construction works even for pushdown systems and not just OCAs.

Let $\A$ be a simple OCA. It will be beneficial to think
of the counter as a stack with a single letter alphabet, with pushes for
increments and pops for decrements. Then, in any run from $(p,0)$ to $(q,0)$, 
we may relate an increment move  uniquely with its  \emph{corresponding} decrement move, the pop that removes the value inserted by this push.

Now, consider a \emph{ one reversal run } $\rho$ of $\A$
from say $(p,0)$ to $(q,0)$ involving two phases, a first phase $\rho_1$ 
with no decrement moves and a second phase $\rho_2$ with no increment moves. 
Such a run can be simulated, up to equivalent Parikh
image (i.e. upto reordering of the letters read along the run) by an NFA
as follows: simultaneously simulate the first phase ($\rho_1$) from the 
source and the second phase, in reverse order ($\rho_2^{rev}$), from the target.
 (The simulation of $\rho_2^{rev}$ uses the transitions in the \emph{opposite}
direction, moving from the target of the transition to the source of the transition).
The simulation  matches increment moves of $\rho_1$  against 
decrement moves in $\rho_2^{rev}$  (more precisely, matching the $i$th
increment $\rho_1$ with the $i$th decrement in $\rho_2^{rev}$)
while carrying out moves that do not alter the counters 
independently in both directions.  The simulation terminates (or potentially
terminates)  when a common state, signifying the boundary between $\rho_1$ and
$\rho_2$  is reached from both ends. 

The state space of such an NFA will need pairs of states from $Q$, to maintain
the current state reached by the forward and backward simulations. Since, only
one letter of the input can be read in each move, we will also need two moves
to simulate a matched increment and decrement and will need states of the form
$Q \times Q \times \Sigma$ for the intermediate state that lies between the two
moves.

Unfortunately, such a naive simulation would not work if the run had
more \emph{reversals}.  For then the $i$th increment in the simulation from the 
left need not necessarily correspond to the $i$th decrement in the reverse
simulation from the right.  In this case, the run $\rho$  can be written
as follows:
\begin{align*}
(p,0) \rho_1 (p_1,c) & \moves{\tau_1} (p'_1,c+1)  \rho_3 (p'_2,c+1)\\
&\moves{\tau_2}(p_2,c)\rho_4(q_1,c)\rho_5(q,0)
\end{align*}
  where, the increment $\tau_1$
corresponds to the decrement $\tau_2$ and all the increments in $\rho_1$ are exactly 
matched by decrements in $\rho_5$.  Notice that the increments
in the run $\rho_3$ are exactly matched by the decrements in that run and
similarly for $\rho_4$.  Thus, to simulate such a well-matched run from $p$ to $q$, 
after simulating $\rho_1$ and $\rho_5^{rev}$ simultaneously matching corresponding 
increments and decrements, and reaching the state $p_1$ on the left and $q_1$ on the right, 
we can choose to now simulate matching runs from $p_1$ to $p_2$
and from $p_2$ to $q_1$ (for some $p_2$). Our idea is to choose one of these
pairs and simulate it first, storing the other in a stack. We call such
pairs \emph{obligations}.   The simulation of the chosen obligation may produce further such obligations which are also stored in the stack.  The simulation of an obligation  succeeds when  the state reached from
the left and right simulations  are identical, and at this point 
we  we may choose to close this simulation and pick up the next obligation from the stack or continue simulating the current pair further. The entire simulation terminates when no obligations are left. Thus, to go from a single
reversal case to the general case, we have introduced a stack into which 
states of the NFA used for the single reversal case are stored. 
This can be formalized to show that the resulting PDA is Parikh-equivalent 
to $\A$.

\onlyappendix{
We also add that the order in which the obligations are verified is not important, however, the use of a stack to do this simplifies the arguments.
Observe that in this construction each obligation inserted into the stack
corresponds to a reversal in the run being simulated, as a matter of fact, it
will correspond to a reversal from decrements to increments.  Thus it is quite
easy to see that the stack height of the simulating run can be bounded by the
number of reversals in the original run.  }

But a little more analysis shows that there is a  simulating
run where the height of the stack is bounded by $log(R)$ where $R$ is the
number of reversals in the original run.  
Thus, to simulate all runs of $\A$ 
with at most $R$ reversals, we may bound the stack height of the PDA by $log(R)$.

  We show that if the stack height is $h$ then we can 
choose to simulate only runs with at most $2^{log(R) - h}$  reversals for the obligation on hand. Once we
show this, notice that when $h = log(R)$ we only need to simulate runs with $1$ reversal which can be done without any further obligations being generated. 
Thus, the overall height of the stack is bounded by $log(R)$.  Now, we explain
why the claim made above holds.  Clearly it holds initially when $h=0$. 
Inductively, whenever we  split an obligation, we choose
the obligation with fewer reversals to simulate first, pushing the other
obligation onto the stack.  Notice that this obligation with fewer reversals
is guaranteed to contain at most half the number of reversals of the 
current obligation (which is being split).
Thus, whenever the stack height increases by $1$, the number of reversals
to be explored in the current obligation falls at least by half as required. On the
other hand, an obligation $(p,q)$ that lies in the stack at position $h$ from the bottom, was placed there while executing (earlier) an obligation $(p',q')$
 that only required $2^{k-h+1}$ reversals. Since the obligation $(p,q)$ 
contributes only a part of the obligation $(p',q')$, its number of
reversals is also bounded by $2^{k-h+1}$.  And when $(p,q)$ is removed from
the stack for simulation, the stack height is $h-1$. Thus, the invariant
is maintained. \onlymainpaperpara{%
Once we have this bound on the stack, for a given
$R$, we can simulate it by an exponentially large NFA.  This yields the following lemma: }

\onlyappendix{
We now describe the formal construction of the automaton establish its correctness now.  We establish the result directly for a pushdown system. A pushdown
system is a tuple $\A = (Q,\Sigma,\Gamma,\bot,\delta,s,F)$ where $\Gamma$ is the
stack alphabet and $\bot$ is a special bottom of stack symbol.  The transitions
in $\delta$ are of the form $(q,a,push(x),q')$ denoting a move where the 
letter $x \in \Gamma$ is pushed on the stack while reading $a \in \SigmaE$,
or $(q,a,pop(x),q')$ denoting a move where the letter $x \in \Gamma$ is popped
from the stack while reading $a \in \SigmaE$ or $(q,a,\nop,q')$ where the
stack is ignored while readin $a \in \SigmaE$. A configuration of such
a pushdown is a pair $(q,\gamma)$ with $q \in Q$ and $\gamma = \Gamma^*\bot$.
The notion of move $(q,\gamma) \moves{\tau} (q',\gamma')$ using some 
$\tau \in \delta$ and $(q,\gamma) \moves{a} (q',\gamma')$ where $a \in \SigmaE$
are defined as expected and we omit the details here. 

Observe first of all that if $\Gamma$ is a singleton we have exactly a
simple OCA.  The push moves correspond to increments, pop moves to 
decrements and there are no \emph{emptiness tests} here as there are no
zero tests in simple OCAs, and the correspondence between configurations
is obvious. We remark that as far as PDAs go, the lack
of an emptiness test is not a real restriction as we can push  a special
symbol right at the beginning of the run and subsequently  simulate an 
emptiness test by popping and pushing this symbol back on to the stack. 
Thus, we lose no generality either.  Having said this, we use emptiness
test in the PDA we construct as it simplifies the presentation (while
omitting it from the one given as input w.l.o.g.)

Given a PDA $\A = (Q,\Sigma,\Gamma,\delta,\bot,s,F)$ we construct a new
PDA $\A_P$ which simulates runs of $\A$, upto Parikh-images, and
does so using runs where the stack height is bounded by $log(R)$ where
$R$ is the number of reversals in the run of $\A$ being simulated. 
$\A_P = (\Gamma_P \cup \{s_P,t_P\},\Sigma,\Gamma_P,\delta_P,s_P,t_P)$ is defined as follows.
The set of $\Gamma_P$ is given by   $(Q \times Q)  \cup  (Q \times Q \times
\Sigma)$. States of the form $(p,q)$ are charged with 
simulating a well matched run from $(p,\bot)$ to $(q,\bot)$. 
While carrying out a matched push from the left and a pop from the right, 
as we are only allowed read one letter of $\Sigma$ in a single move, we are
forced to have an intermediary state to allow for the reading of the letters
corresponding to both the transitions being simulated. 
The states of the form $(p,q,a)$, $a \in \Sigma$, are used for this purpose. 
The transition relation $\delta_P$ is described below: 
\begin{enumerate}
\item $(s_P,\eps,\nop,(s,t)) \in \delta_P$. Initialize the start and target states.
\item $((p,q),a,\nop,(p',q)) \in \delta_P$ whenever $(p,a,\nop,p') \in \delta$. Simulate an internal move from the left.
\item $((p,q),a,\nop,(p,q')) \in \delta_P$ whenever $(q',a,\nop,q) \in \delta$. Simulate an internal move from the right.
\item $((p,q),a,\nop,(p',q',b))\in \delta_P$ whenever $$(p,a,\push(x),p'), (q',b,\pop(x),q) \in \delta_P$$ for some $x \in \Gamma$. 
Simulate  a pair of matched moves, a push from the source and the corresponding
pop from the target, first part.
\item $((p,q,b),b,\nop,(p,q))\in \delta_P$ whenever $b \in \Sigma$. 
Second part of the move described in previous item.
\item $((p,q),\eps,\push((q',q)),(p,q')) \in \delta_P$ for every state $q' \in Q$. Guess a intermediary state where a pop to push reversal occurs.  Simulate first half first and push the second as an obligation on the stack.
\item $((p,q),\eps,\push((p,q')),(q',q)) \in \delta_P$ for every state $q' \in Q$. Guess a intermediary state where a pop to push reversal occurs.  Simulate
second half first and push the first as an obligation on the stack.
\item $((p,p),\eps,\pop((p',q')),(p',q')) \in \delta_P$. Current 
obligation completed, load next one from stack.
\item $((p,p),\eps,\etest,t_P)\in \delta_P$. All segments completed successfully, so accept.
\end{enumerate}

The following Lemma shows that every run of $\A_P$ simulates some run
of $\A$ upto Parikh-image. In what follows, we say that a run $\rho$ is a
$\gamma$-run for some $\gamma \in \Gamma^*\bot$ if $\gamma$ is a suffix of the
stack contents in every configuration in $\rho$.  
\begin{lemma}\label{lem:BRpdsToPDS}
Let $\beta \in \Gamma_P^*\bot$.  Let $((p,q),\beta) \moves{w} ((r,r),\beta)$
be a $\beta$-run in $\A_P$, for some $p,q$ and $r$ in $Q$.  Then, for every
$\gamma \in \Gamma^*\bot$ there is a run $(p,\gamma) \moves{w'} (q,\gamma)$ in
$\A$ such that $\Parikh(w') = \Parikh(w)$. Thus, if $w \in \lang(\A_P)$ then 
there is a $w'$ in $\lang(\A)$ with $\Parikh(w) = \Parikh(w')$.
  \end{lemma}
\begin{proof}

For purpose of the proof, we will prove the following  claim.

\begin{claim}\label{no-stack}
If there is a run of the form $((p,q),\beta) \moves{v} ((p',q'),\\ \beta)$ in $\A_{P}$, then for every $\gamma \in \Gamma^*$,  there are runs of the form $(p,\gamma\bot) \moves{v_1} (p',\alpha\gamma\bot)$ and $(q',\alpha\gamma\bot) \moves{v_2} (q,\gamma\bot)$, such that $\Parikh(v) = \Parikh(v_1.v_2)$.
\end{claim}
\begin{proof}

We will now prove this by inducting on stack height reached and on length of the run. Suppose the stack was never used (always remained $\beta$), then the proof is easy to see.

Let us assume that stack was indeed used, then the run $((p,q),\beta) \moves{v} ((p',q'),\beta)$ can be split as 

\begin{multline*}
((p,q),\beta) \moves{v_1} ((p_1,q_1),\beta) \moves{} ((p_2,q_2),(t_1,t_2)\beta) \moves{v_2}\\
((r_1,r_1),(t_1,t_2) \beta) \moves{}  ((t_1,t_2),\beta) \moves{v_3} ((p',q'),\beta)
\end{multline*}

We have two cases to consider, either $q_1 = t_2$ or $p_1 = t_1$. We will consider the case where $q_1 = t_2$, the other case is analogous. In this case, clearly $p_2 = p_1$ and $t_1 = q_2$. Hence the run is of the form 

\begin{multline*}
((p,q),\beta) \moves{v_1} ((p_1,q_1),\beta) \moves{} ((p_1,q_2),(q_2,q_1)\beta) \moves{v_2}\\
((r_1,r_1),(q_2,q_1) \beta) \moves{} ((q_2,q_1),\beta) \moves{v_3} ((p',q'),\beta)
\end{multline*}

\noindent
\paragraph{}Now consider the sub-run of the form 
$$((p,q),\beta) \moves{v_1} ((p_1,q_1),\beta)$$

clearly such a run is shorter and hence by induction we have a corresponding runs of the form $(p,\gamma\bot) \moves{v'_1} (p_1,\alpha''\gamma\bot)$ and $(q_1,\alpha''\gamma\bot) \moves{v''_1} (q,\gamma\bot)$, for some $\alpha'' \in \Gamma^*$ and such that $\Parikh(v_1) = \Parikh(v'_1.v''_1)$.

\noindent
\paragraph{}Consider the sub-run of the form 
$$((p_1,q_2),(q_2,q_1)\beta) \moves{v_2}((r_1,r_1),(q_2,q_1) \beta)$$

clearly stack height of such a run is shorter by $1$. Hence by induction, we have a corresponding runs of the form, $(p_1,\alpha\gamma\bot) \moves{v'_2} (r_1,\alpha'\alpha\gamma\bot)$ and $(r_1,\alpha'\alpha\gamma\bot) \moves{v''_2} (q_1,\alpha\gamma\bot)$ for some $\alpha' \in \Gamma^*$, such that $\Parikh(v_2) = \Parikh(v'_2.v''_2)$.

\noindent
\paragraph{} consider the sub-run of the form 
$$((q_2,q_1),\beta) \moves{v_3} ((p',q'),\beta)$$

clearly such a run is shorter in length, hence by induction, we have corresponding runs $(q_2,\gamma\bot) \moves{v'_3} (p',\alpha\gamma\bot)$ and $(q',\alpha\gamma\bot) \moves{v''_3}(q_1,\gamma\bot)$, for some $\alpha \in \Gamma^*$ and such that $\Parikh(v_3) = \Parikh(v'_3.v''_3)$.

Now combining these sub-runs, we get the required run.
\end{proof}

It is easy to see that the proof of Lemma follows directly once this claim is in place.

\end{proof}

In the other direction, we show that every run of $\A$ is simulated
upto Parikh-image by $\A_P$ with a stack height that is logarithmic
in the number of reversals.  Let 
$$(p,\alpha) = (p_0,\alpha_0) \moves{\tau_1} (p_2,\alpha_1) \moves{\tau_2} \ldots  \moves{\tau_n} = (p_n,\alpha_n) = (q,\alpha)$$ 
be a run in $\A$.  A reversal in such a run is a sequence of the from 
\begin{align*}
(p_i,\alpha_i) & \moves{a_{i+1},pop(x_i)} (p_{i+1},\alpha_{i+1})~\moves{\tau_{i+2}\ldots\tau_{j-1}}~(p_{j-1},\alpha_{j-1}) \\
& \moves{a_j,push(x_j)} (p_j,\alpha_j) 
\end{align*}
where none of the transitions
$\tau_{i+2} \ldots \tau_{j-1}$ are push or pop moves.  The next lemma shows
how $\A_R$ simulates runs of $\A$ and provides bounds on stack size in terms of
the number of reversals of the run in $\A$. 

\begin{lemma}\label{lem:PDSToBRpds}
Let $(p,\alpha) \moves{w} (q,\alpha)$ be a $\alpha$-run of $\A$ with $R$ reversals with $\alpha \in \Gamma^*.\bot$.
Then, for any $\gamma \in \Gamma_P^*\bot$, there is a $\gamma$-run 
$((p,q),\gamma) \moves{w'} ((r,r),\gamma)$ with $\Parikh(w) = \Parikh(w')$.
Further for any configuration along this run the height of the stack is 
no more than $|\gamma| + log(R + 1)$.
\end{lemma}
\begin{proof}
The  proceeds by a double induction, first on the number of reversals and
then on the length of the run.

For the base case, suppose $R = 0$.  
If the length of the run is $0$ then the result follows trivially.
Otherwise, let the $\alpha$-run $\rho$, $\alpha \in \Gamma^*\bot$
 be of the form:
$$(p,\alpha) = (p_0,\alpha_0) \moves{\tau_1} (p_1,\alpha_1) \moves{\tau_2} \ldots  \moves{\tau_n}  (p_n,\alpha_n) = (q,\alpha)$$
If $\tau_1$ is an internal move $(p_0,a_1,\nop,p_1)$ then $((p_0,p_n),a_1,\nop$ $,(p_1,p_n))$
is a transition $\delta_P$ (of type 2). Thus 
$$((p_0,p_n),\gamma) \moves{a_1} ((p_1,p_n),\gamma)$$
 is a valid move in $\A_P$. Let $w = a_1w_1$. Then, 
 by induction hypothesis, there is a $\gamma$-run 
$$((p_1,p_n),\gamma) \moves{w'_1} ((r,r),\gamma))$$ with $\Parikh(w'_1) = \Parikh(w_1)$, whose
stack height is bounded by $|\gamma|$. Putting these two together we get a
$\gamma$-run  $$((p_0,p_n),\gamma) \moves{a_1.w'_1} ((r,r),\gamma)$$ with
$\Parikh(w) = \Parikh(a_1.w'_1)$ whose stack height is bounded by $|\gamma|$ as required.
 
If $\tau_n$ is an internal transition $(p_{n-1},a_n,\nop,p_n)$ then 
$$((p_0,p_n),a_n,\nop,(p_0,p_{n-1})) \in \delta_P$$ 
is a transition of of type 3.
Thus, $((p_0,p_n),\gamma) \moves{a_n} ((p_0,p_{n-1}),\gamma)$ is a move in
$\A_P$. Further, by the induction hypothesis, there is a word $w_2$ with 
$w = w_2.a_n$ and a $\gamma$-run $((p_0,p_{n-1}),\gamma) \moves{w'_2}
((r,r),\gamma)$ with $\Parikh(w_2) = \Parikh(w'_2)$. Then, since
$\Parikh(a_n.w'_2) = \Parikh(w_2.a_n)$, we can put these two together to get
the requisite run. Once again the stack height is bounded by $|\gamma|$.

Since the  given run is a $\alpha$-run, the only other case left to be considered is when $\tau_1$ is a push move and $\tau_n$ is
a pop move. Thus, let $\tau_1 = (p_0,a_1,push(x_1),p_1)$ and $\tau_n =
(p_{n-1},a_n,pop(x_n),p_n)$. We claim that $x_1 = x_n$ and as a matter fact
the value $x_1$ pushed by $\tau_1$ remains in the stack all the way till end
of this run and is popped by $\tau_n$.
If the $x_1$ was popped earlier in the run than the last step, then the stack height would have necessarily reached $|\alpha|$ at this pop, and therefore there will necessarily be a subsequent push of $x_n$. But this contradicts the 
fact that $R = 0$.  Thus, we have the following moves in $\A_P$.
\begin{align*}
((p_0,p_n),\gamma) &\moves{((p_0,p_n),a_1,\nop,(p_1,p_{n-1},a_n))} ((p_0,p_{n-1},a_n),\gamma) \\
&\moves{((p_1,p_{n-1},a_n),a_n,\nop,(p_1,p_{n-1})} ((p_1,p_{n-1}),\gamma)
\end{align*}
Let $w = a_1 w_3 a_n$. Then applying the induction hypothesis we get a
$\gamma$-run $((p_1,p_{n-1}),\gamma) \moves{w'_3} ((r,r),\gamma)$  where
the
stack height is never more than $|\gamma|$. Combining these two gives
us a $\gamma$-run $((p_0,p_n),\gamma)\moves{a_1a_nw'_3}$ $ ((r,r),\gamma)$ where
the stack height is never more than $|\gamma|$. Observing that $\Parikh(a_1a_nw'_3) = \Parikh(a_1w_3a_n)$ gives us the desired result.

Now we examine runs with $R \geq 1$. And once again we proceed by induction on
the length $l$ of runs with $R$ reversals. For $R \geq 1$ there are no runs
of length $l = 0$ and so the basis holds trivially. As usual, let 
$$(p,\alpha) = (p_0,\alpha_0) \moves{\tau_1} (p_1,\alpha_1) \moves{\tau_2} \ldots  \moves{\tau_n} = (p_n,\alpha_n) = (q,\alpha)$$
be an $\alpha$-run with $R$ reversals. If either $\tau_1$ or $\tau_n$ is an internal move then the proof can proceed by induction on $l$ exactly along the same lines
as above and the details are omitted. Otherwise,  since this is a $\alpha$-run, $\tau_1$ is a push move and $\tau_n$ is a pop
move.  Let $\tau_1 = (p_0,a_1,push(x_1),p_1)$ and $\tau_n = (p_{n-1},a_n,pop(x_n),p_n)$.  Now we have two possibilities. 

\paragraph{Case 1:} The value $x_1$ pushed in $\tau_1$ is popped only by $\tau_n$. This is again easy, as we can apply the same argument as in the case $R = 0$
to conclude that,  
\begin{align*}
((p_0,p_n),\gamma) &\moves{((p_0,p_n),a_1,\nop,(p_1,p_{n-1},a_n))} ((p_0,p_{n-1},a_n),\gamma) \\
&\moves{((p_1,p_{n-1},a_n),a_n,\nop,(p_1,p_{n-1})} ((p_1,p_{n-1}),\gamma)
\end{align*}
Again, with $w = a_1w_3a_2$, and applying the induction hypothesis to the shorter
run $(p_1,\alpha_1) \moves{w_3} (p_{n-1},\alpha_{n-1})$ with exactly $R$
reversals, we obtain a $\gamma$-run  $$((p_1,p_{n-1}),\gamma) \moves{w'_3}
((r,r),\gamma)$$ in which the height of the stack is bounded by 
$|\gamma| + log(R + 1)$. Combining these gives us the $\gamma$-run with 
stack height bounded by $|\gamma| + log(R+1)$, $((p_0,p_n),\gamma) \moves{a_1a_nw'_3} ((r,r),\gamma)$ as required. 

\paragraph{Case 2:} The value $x_1$ pushed in $\tau_1$ is popped by some $\tau_j$ with $j < n$. Then we break the run into two $\alpha$-runs, 
$\rho_1 = (p_0,\alpha_0) \moves{a_1\ldots a_j} (p_j,\alpha_j)$  and
$\rho_2 = (p_j,\alpha_j) \moves{a_{j+1}\ldots a_n} (p_n,\alpha_n)$. Note that
$\alpha = \alpha_0 = \alpha_j = \alpha_n$. Let $a_1 \ldots a_j = w_1$
and $a_{j+1} \ldots a_n = w_2$. Let the number of reversals
of $\rho_1$ and $\rho_2$ be $R_1$ and $R_2$ respectively. First of all, 
we observe that $R_1 + R_2 + 1 = R$. Thus $R_1,R_2 < R$ and  further
either $R_1 \leq R/2$ or $R_2 \leq R/2$. 

Suppose $R_1 \leq R/2$. Then,  by the induction hypothesis, there is an
$((p_j,p_n)\gamma)$-run $$\rho'_1~=~ (((p_0,p_j),(p_j,p_n).\gamma) \moves{w'_1} ((r',r'),(p_j,p_n).\gamma))$$ with $\Parikh(w_1) = \Parikh(w'_1)$ and   whose stack height is bounded by 
\begin{align*}
|(p_j,p_n).\gamma| + log(R_1 + 1) &~=~ |\gamma| + 1 + log(R_1 + 1) \\
&~\leq~ |\gamma| + 1 + log(R + 1) - 1 \\
&~=~ |\gamma| + log(R + 1) 
\end{align*}
Similarly, by the induction hypothesis, there is an $\gamma$-run  $\rho'_2 = 
((p_j,p_n),\gamma) \moves{w'_2} ((r,r),\gamma)$ whose number of reversals
is bounded by $|\gamma| + log(R_2 + 1) ~\leq~ |\gamma| + log(R + 1)$ and for which
$\Parikh(w'_2) = \Parikh(w_2)$.

We have everything in place now. We construct the desired run by first using
a transition of type 6, following by $\rho'_1$, followed by a transition of
type 8, followed by a simulation of $\rho'_2$ to obtain the following:
\begin{align*}
((p_0,p_n),\gamma) &\moves{((p_0,p_n),\eps,push((p_j,p_n)),(p_0,p_j))} ((p_0,p_j),(p_j,p_n).\gamma) \\
&\moves{w'_1} ((r',r'),(p_j,p_n)\gamma)\\
&\moves{((r',r'),\eps,pop((p_j,p_n)),(p_j,p_n))} ((p_j,p_n),\gamma)\\
&\moves{w'_2} ((r,r),\gamma)
\end{align*}
This runs satisfies all the desired properties. The case where $R_2 \leq R/2$ 
is handled similarly using moves of type 7 instead of type 6 and using the fact
the $\Parikh(w'_2.w'_1) = \Parikh(w'_1.w'_2)$. This completes the proof of
the Lemma. 
\end{proof}

As we did for OCAs we let $\lang_R(\A)$ refer to the language of words 
accepted by $\A$ along runs with atmost $R$ reversals. Now, for a given
$R$, we can simulate runs of $\A_P$ where stack height is bounded by $log(R)$, using an NFA by keeping the stack as part of the state.  The size of
such an NFA is $O(|Q_P| |\Gamma_P|^{O(log(R))}) = O(|\Sigma||Q|^{O(log(R))})$. 
Let $\A_R$ be such an NFA.  Then by Lemma \ref{lem:BRpdsToPDS}, we have  
$\Parikh(\lang(\A_R)) \subseteq \Parikh(\lang(\A))$ and by 
Lemma \ref{lem:PDSToBRpds}  we also have 
$\Parikh(\lang_R(\A)) \subseteq \Parikh(\lang(\A_R))$.
By keeping track of the reversal count in the state, we may
construct an $\A'$ with state space size $O(R.|Q|)$ such that 
that $\lang(\A') = \lang_R(\A') = \lang_R(\A)$.  Thus, we have
}
\onlymainpaper{
\begin{restatable}{lemma}{rbpdas}
\label{lem:rbpdas}
There is a procedure that takes a simple OCA $\A$ with $K$ states and
whose alphabet is $\Sigma$, and a number $R \in \Nat$ and returns 
an NFA Parikh-equivalent to $\lang_R(\A)$ of size $O(|\Sigma|.(RK)^{O(log(R))})$.
\end{restatable}
}
\onlyappendix{
\rbpdas*
}

\subsection{Completeness result}
\newcommand{\sCount}{a}
\newcommand{\sConn}{c}

\newcommand{\ocaedge}[2][ZZZ]{\ifthenelse{\equal{#1}{ZZZ}}{#2}{#2| #1}}
\begin{figure}
\scalebox{.5}{
\centering
\begin{tikzpicture}[every state/.style={minimum size=10pt}]
\clip (-2,-2.5) rectangle (15,4);
\node[state, initial, initial text=]	(q1) at (0,0) {$q_1$};
\node[state]				(q2) at (4,0) {$q_2$};
\node[state]				(q3) at (8,0) {$q_3$};
\node                                   (dots) at (11,0) {$\cdots$};
\node[state,accepting by arrow]		(qn) at (13,0) {$q_{n}$};
\path[->] (q1) edge [above]      node {$\ocaedge{\sConn_{1,2}}$} (q2)
          (q2) edge [above]      node {$\ocaedge{\sConn_{2,3}}$} (q3)
          (q1) edge [bend right, above]      node {$\ocaedge{\sConn_{1,3}}$} (q3)
          (q1) edge [bend right, above]      node {$\ocaedge{\sConn_{1,n}}$} (qn)
          (q2) edge [bend right, above]      node {$\ocaedge{\sConn_{2,n}}$} (qn)
          (q3) edge [bend right, above]      node {$\ocaedge{\sConn_{3,n}}$} (qn)
          (q1) edge [loop above] node {$\ocaedge[1]{\sCount_{1,1}}$} (q1)
          (q1) edge [loop above, out=130, in=50, looseness=15] node {$\ocaedge[2]{\sCount_{1,2}}$} (q1)
          (q1) edge [loop above, out=135, in=45, looseness=28] node {$\ocaedge[n]{\sCount_{1,n}}$} node[below]  {$\vdots$} (q1)
          (q2) edge [loop above] node {$\ocaedge[\!-\!1]{\sCount_{2,1}}$} (q2)
          (q2) edge [loop above, out=130, in=50, looseness=15] node {$\ocaedge[\!-\!2]{\sCount_{2,2}}$} (q2)
          (q2) edge [loop above, out=135, in=45, looseness=28] node {$\ocaedge[\!-\!n]{\sCount_{2,n}}$} node[below] {$\vdots$} (q2)
          (q3) edge [loop above] node {$\ocaedge[1]{\sCount_{3,1}}$} (q3)
          (q3) edge [loop above, out=130, in=50, looseness=15] node {$\ocaedge[2]{\sCount_{3,2}}$} (q3)
          (q3) edge [loop above, out=135, in=45, looseness=28] node {$\ocaedge[n]{\sCount_{3,n}}$} node[below] {$\vdots$} (q3)
          (qn) edge [loop above] node {$\ocaedge[(-1)^{n+1}]{\sCount_{n,1}}$} (qn)
          (qn) edge [loop above, out=145, in=35, looseness=15] node {$\ocaedge[(-1)^{n+1} 2]{\sCount_{n,2}}$} (qn)
          (qn) edge [loop above, out=145, in=35, looseness=28] node {$\ocaedge[(-1)^{n+1} n]{\sCount_{n,n}}$} node[below] {$\vdots$} (qn);
\end{tikzpicture}}
\caption{The one-counter automaton $\HardAutomaton{n}$}
\label{completeness:figure}
\end{figure}

In this subsection, we present a simple sequence of OCA that is complete with
respect to small Parikh-equivalent NFAs. This means, if the OCA in this
sequence have polynomial-size Parikh-equivalent NFAs, then all OCA have
polynomial-size Parikh-equivalent NFAs.

It will be convenient to slightly extend the definition of OCA.  An
\emph{extended OCA} is defined as an OCA, but in its transition $(p,a,s,q)$,
the entry $s$ can assume any integer (in addition to $\zerotest$). Of course
here, the number of states is not an appropriate measure of size. Therefore,
the \emph{size} of a transition $t=(p,a,s,q)$ of $\A$ is $|t|=\max(0, |s|-1)$
if $s\in\Z$ and $0$ if $s=\zerotest$. If $\A$ has $n$ states, then we define
its \emph{size} is $|\A| = n+\sum_{t\in \tran} |t|$. Given an extended OCA of
size $n$, one can clearly construct an equivalent OCA with $n$ states.
Furthermore, if one consideres an (ordinary) OCA as an extended OCA, then
its size is the number of states.

The complete sequence $(\HardAutomaton{n})_{n\ge 1}$ of automata consists of extended
OCA and is illustrated in Figure~\ref{completeness:figure}. The automaton $\HardAutomaton{n}$
has $n$ states, $q_1,\ldots,q_n$. On each $q_i$ and for each $k\in[1,n]$, there
is a loop reading $\sCount_{i,k}$ and adding $(-1)^{i+1}\cdot k$ to the counter.
Moreover, for $i,j\in[1,n]$ with $i<j$, there is a transition reading
$\sConn_{i,j}$ that does not use the counter. For each $k\in[1,n]$, $\HardAutomaton{n}$ has 
$n$ transitions of size $k-1$. Since it has $n$ states, this results in a size of
$n+\sum_{k=1}^n (k-1)=\frac{1}{2}n(n+1)$.

The result of this section is the following.
\begin{restatable}{theorem}{completenessResult}
\label{completeness:result}
There are polynomials $p$ and $q$ such that the following holds.  
\begin{enumerate} 
\item If for each $n$, there is a Parikh-equivalent NFA for $\HardAutomaton{n}$
with $h(n)$ states, then for every OCA of size $n$, there is a
Parikh-equivalent NFA with at most $q(h(p(n)))$ states. 
\item If there is an algorithm that computes a Parikh-equivalent NFA for
$\HardAutomaton{n}$ in time $O(h(n))$,  then one can compute a
Parikh-equivalent NFA for arbitrary OCA in time $O(q(h(p(n))))$.
\end{enumerate} 
\end{restatable}

Explicitly, we only prove the first statement and keep the analogous statements
in terms of time complexity implicit. Our proof consists of three steps
(Lemmas~\ref{completeness:rba}, \ref{completeness:loop-counting}, and~%
\ref{completeness:hardaut}). Intuitively, each of them is one algorithmic
step one has to carry out when constructing a Parikh-equivalent NFA for a given
OCA. 

\newcommand{\rba}{RBA}
\newcommand{\rbas}{RBAs}
\newcommand{\anrba}{an}
\newcommand{\Anrba}{An}
\newcommand{\prerun}{walk}
\newcommand{\aprerun}{a}
\newcommand{\Aprerun}{A}
\newcommand{\preruns}{walks}
\newcommand{\posTran}{\tran_+}
\newcommand{\negTran}{\tran_-}

For the first step in our proof, we need some terminology.
Let $\A=(Q,\AlB,\tran,q_0,F)$ be an extended  OCA. Recall that a word
$(p_1,a_1,s_1,p'_1)\cdots (p_n,a_n,s_n,p'_n)$ over $\delta$ is called \aprerun\
\emph{\prerun} if $p'_i=p_{i+1}$ for every $i\in [1,n-1]$.  The \prerun\ $u$ is
called a \emph{$p_1$-cycle} (or just \emph{cycle}) if $p'_n=p_1$. If, in
addition, $i\ne j$ implies $p_i\ne p_j$, then $u$ is called \emph{simple}.  A
cycle as above is called \emph{proper} if there is some $i\in[1,n]$ with
$p_i\ne p_1$. We say that $\A$ is \emph{acyclic} if it has no proper cycles,
i.e.  if all cycles consist solely of loops. Equivalently, an OCA is acyclic if
there is a partial order $\le$ on the set of states such that if a transition
leads from a state $p$ to $q$, then $p\le q$.

A transition $(p,a,s,q)$ is called \emph{positive} (\emph{negative}) if $s>0$
($s<0$).  We say that \aprerun\ \prerun\ \emph{contains $k$ reversals} if it
has a scattered subword of length $k+1$ in which positive and negative
transitions alternate.  An (extended) OCA is called
\emph{($r$-)reversal-bounded} if none of its \preruns\ contains $r+1$
reversals.  Observe that an acyclic (extended) OCA is reversal-bounded if and
only if on each state, there are either no positive transitions or no negative
transitions. We call such automata \emph{\rba}
(\emph{reversal-bounded acyclic automata}).

Recall that we have seen in section~\ref{sec:reversal-bounding} that constructing
Parikh-equivalent NFAs essentially reduces to the case of reversal-bounded
simple OCA.  The first construction here takes a reversal-bounded automaton and
decomposes it into an RBA and a regular substitution. This means, if we can
find Parikh-equivalent OCA for RBAs, we can do so for arbitrary OCA: Given an
NFA for the RBA, we can replace every letter by the finite automaton specified
by the substitution.  Here, the \emph{size} of the substitution $\sigma$ is the
maximal number of states of an automaton specified for a language $\sigma(a)$,
$a\in\Sigma$.
\begin{restatable}{lemma}{completenessRBA}
\label{completeness:rba}
Given an $r$-reversal-bounded simple OCA $\A$ of size $n$, one can construct
\anrba\ \rba\ $\B$ of size $6n^5(r+1)$ and a regular substitution $\sigma$ of size
$n(n+1)$ such that $\parikh{\sigma(\langop{\B})}=\parikh{\langop{\A}}$.
\end{restatable}
We prove this lemma by showing that runs of reversal-bounded simple OCA
can be `flattened': Each run can be turned into one with
Parikh-equivalent input that consists of a skeleton of polynomial length in
which simple cycles are inserted flat, i.e. without nesting them. The RBA $\B$
simulates the skeleton and has self-loop which are replaced by $\sigma$ with a
regular language that simulates simple cycles.

In the next construction of our proof (Lemma~\ref{completeness:loop-counting}),
we employ a combinatorial fact.  A \emph{Dyck sequence} is a sequence
$x_1,\ldots,x_n\in\Z$ such that $\sum_{i=1}^k x_i\ge 0$ for every $k\in[1,n]$.
We call a subset $I\subseteq[1,n]$ \emph{removable} if removing all $x_i$,
$i\in I$, from the sequence yields again a Dyck sequence and $\sum_{i\in I}
x_i=\sum_{i=1}^n x_i$.  We call the sequence \emph{$r$-reversal-bounded }if
there are at most $r$ alternations between positive numbers and negative
numbers.
\begin{restatable}{lemma}{completenessDyck}
\label{completeness:dyck}
Let $N\ge 0$ and $x_1,\ldots,x_n$ be an $r$-reversal-bounded Dyck sequence with
$x_i\in[-N,N]$ for each $i\in[1,n]$ such that $\sum_{i=1} x_i\in[0,N]$.  Then
it has a removable subset $I\subseteq [1,n]$ with $|I|\le 2r(2N^2+N)$.
\end{restatable}

We now want to make sure that the self-loops on our states are the only
transitions that use the counter.  Note that this is a feature of
$\HardAutomaton{n}$.  \Anrba\ \rba\ is said to be \emph{loop-counting} if its
loops are the only transitions that use the counter, i.e. all other transitions
$(p,a,s,q)$ have $s=0$.
\begin{restatable}{lemma}{completenessLoopCounting}
\label{completeness:loop-counting}
Given \anrba\ \rba\ $\A$, one can construct a loop-counting \rba\  $\B$
of polynomial size such that 
$\langop{\A}\subseteq\langop{\B}\subseteq\langop[K]{\A}$
for a polynomially bounded $K$.
\end{restatable}
Here, the idea is to add a counter that tracks the counter actions of non-loop
transitions.  However, in order to show that the resulting automaton can still
simulate all runs while respecting its own counter (i.e. it has to
reach zero in the end and cannot drop belo zero), we use
Lemma~\ref{completeness:dyck}. It allows us to `switch' parts of the
run so as not to use $\B$'s counter, but the internal one in the state.
Note that according to Lemma~\ref{lem:simple-approx}, it suffices to construct
Parikh-equivalent NFAs that fulfill the approximation relation in the lemma
here.

We are now ready to reduce to the NFA $\HardAutomaton{n}$.
\begin{restatable}{lemma}{completenessHardAut}
\label{completeness:hardaut}
Given a loop-counting \rba\ $\A$ of size $n$, one can construct a regular
substitution $\sigma$ of size at most $2$ such that
$\parikh{\sigma(\langop{\HardAutomaton{2n+2}})}=\parikh{\langop{\A}}$.
\end{restatable}
Here, roughly speaking, we embed the partial order on the set of states into
the one in $\HardAutomaton{2n+2}$. Then, the substitution $\sigma$ replaces
each symbol in $\HardAutomaton{2n+2}$ by the outputs of the correponding
transition in $\A$. Here, we have to deal with the fact that in $\A$, there
might be loops that do not read input, but those do not exist in
$\HardAutomaton{2n+2}$.  However, we can replace the symbols on non-loops by
regular languages so as to produce the output of their neighboring loops.

\bibliographystyle{abbrvnat}
\bibliography{picl}

\newpage
\appendix
\section{Equivalence between problems of characterizations of the three 
regular abstractions for OCA and simple OCA}

Due to the following Lemma~\ref{app:lem:simple-approx},
without loss of generality we can restrict our attention to
computing the regular abstractions of a subclass OCA, namely 
\emph{simple OCA}, defined analogously to OCA but different
in the following aspects:
(1) there are no zero tests,
(2) there is a unique final state, $F = \{\qfinal\}$,
(3) only runs that start from the configuration $(\qinit, 0)$ and end
    at the configuration $(\qfinal, 0)$ are considered \df{accepting}.
The language of a simple OCA \A, also denoted $\langop\A$,
is the set of words induced by accepting runs.

\begin{definition}\label{app:def:approximants}
Let $\A$ be a simple OCA. We define a sequence of $\langop[n]{\A}$ called approximants 
defined as language of words observed along runs of $\A$ from a configuration $(q_0,n)$ to a configuration $(\qfinal,n)$.
\end{definition}

Let OCA $\A=(Q, \AlB, q_0, \tran, F)$, for any $p,q \in Q$ 
we define $\A^{p,q} \eqdef (Q, \AlB, p, 
\tran^+, \{q\})$ where $\tran^+$ is the subset of all transitions in 
$\tran$ which are not test for zero transitions.

\begin{lemma}\label{app:lem:simple-approx}
Let $K\in\N$ and $\diamondsuit\in \{\upop{},\downop{}, \parikhnoarg\}$ be an abstraction.
Assume that there is a polynomial $g_{\diamondsuit}$ such that
for any OCA \A following holds: for every $\A^{p,q}$ there is an NFA
$\B_{p,q, \diamondsuit}$
such that $\diamondsuit{\langop{\A^{p,q}}}\subseteq\langop{\B_{p,q, \diamondsuit}}\subseteq\diamondsuit(\langop[K]{\A^{p,q}})$
and $\sizeop{\B_{p,q,\diamondsuit}} \le g_{\diamondsuit}(\sizeA)$.
Then there is a polynomial $f$ such that for any $\A$ there is an
NFA $\B^{\diamondsuit}$ of size at most $f(\sizeA, K)$ with
$\langop{\B^{\diamondsuit}} = \diamondsuit({\langop{\A}})$.
\end{lemma}
\begin{remark}
Restricting ourself to simple OCA, mentioned in the beginning of the section, 
is a consequence of Lemma~\ref{app:lem:simple-approx} for $K=0$. 
Indeed, if we can build polynomially bounded NFA for an abstraction of any given NFA,
then we can do it for any simple OCA as well. On the other hand, if we can do it,
for any simple OCA than due to Lemma~\ref{app:lem:simple-approx} we can construct
a polynomially bounded NFA that is abstraction of any given OCA.
\end{remark}

In order to prove the lemma, consider the following definition.

\begin{definition}
Let $\lang,\lang'\subseteq \AlB^*$. By a \emph{concatenation of languages}
denoted $\lang\cdot \lang'$ we mean 
a set of all words that can be obtained by concatenation of a word from $\lang$
and a word $\lang'$ i.e.
$\{w\cdot w'|w\in \lang, w'\in \lang'\}$.
\end{definition}

Let $\Pi$ be a finite alphabet. 
A \emph{substitution} is a function $\phi\colon \Pi\to 2^{\AlB^*}$.
Suppose $\lang\subseteq\Pi^*$ is a language. For $w=w_1\cdots w_n$
with $w_1,\ldots,w_n\in\Pi$, the set $\phi(w)$ contains all words 
$x_1\cdots x_n$ where $x_i\in \phi(w_i)$. Lifting this notation to languages,
we get $\phi(\lang)=\bigcup_{w\in \lang} \phi(w)$.

\begin{lemma}\label{lem:constuctionOfSimpleAutomata}
Let $\A=(Q,\AlB, \delta, q_0, F)$ be an OCA and suppose we are given, for each $\A^{p,q}$, an OCA
$\B_{p,q}$ with $\langop{\A^{p,q}}\subseteq\langop{\B_{p,q}}\subseteq\langop[n]{\A^{p,q}}$. Then there is
an NFA $\B$ over an alphabet $\Pi$ 
and substitution $\phi:\Pi\lra 2^{\Sigma^{*}}$ such that
$\langop{\A}=\phi(\langop{\B})$ and that for every letter $a\in\Pi$ hold
$\phi(a) \in \AlB$ or $\phi(a)$ equals
$\langop{\B_{p,q}}$ for some $p,q\in Q$.
\end{lemma}
\begin{proof}
First, for any $n\in \N$ holds that, an
accepting run $\pi$ of $\A$ can be decomposed 
into $\pi=\pi_1\pi_1'\pi_2\pi_2'\ldots \pi_{n-1}'\pi_n$ where
each $\pi_i$ is a run in OCA such that counter values along $\pi_i$ 
stay below $n$, and where 
each run $\pi_i'$ induces a word from a language 
$\langop[n]{\A^{\initstate{\pi_i'}, \finalstate{\pi_i'}}}$.

Thus we define $\B$ as follows:
$\Pi$ is a \AlB plus the set of triples $(n, p, q)$ where $p,q \in Q$, 
the set of states is equal $Q\times \{0\ldots n\}$, 
initial state is $(q_0, 0)$ and final states are
$F\times \{0\}$. The last thing is set of transitions;
for $p,q \in Q$ it contains transition  
$((p,n),(n , p , q), 0, (q,n))$ and for any $i,j\in \{0\ldots n\}$
the transition $((p,i) a, 0, (q, j))$ if and only if there is a move
$(p,i)\lra[a](q,j)$ in $\A$.

The $\phi$ is simply induced by its definition on elements of $\Pi$:
for every $a \in \Pi$ we have that
if $a \in \AlB$ then $\phi(a) = \{a\}$ else if  
$a=(n, p, q)$ and $\phi(a)=\langop[n]{\A^{p,q}}$.

It is easy check that $\langop{\A}=\phi(\langop{\B})$. 
\end{proof}

Next, we show how use $\B$ if we are interested in 
$F$-abstractions of $\langop{\A}$. We start from a following 
simple observations; for any language $\lang= \lang_1 \cdot \lang_2$ 
we have that:
\begin{itemize}
\item $\downop{\lang}=\downop{\lang_1}\cdot\downop{\lang_2}$
\item $\upop{\lang} = \upop{\lang_1} \cdot \upop{\lang_2}$
\item $\parikh{\lang}=\parikh{\lang_1} \cdot \parikh{\lang_2}$
\end{itemize}

A natural consequence is the following lemma.

\begin{lemma}\label{lem:homo-image}
Let $\B$ be an automaton like one defined in 
Lemma~\ref{lem:constuctionOfSimpleAutomata} for a given OCA $\A$.
Now let $\phi_{\>\downop {}}, \phi_{\>\upop{ }}, \phi_{\>\parikhnoarg}$
 are defined as follows: If $a\in \AlB$ then 
$\phi_{\>\downop{ }}(a) = \downop{\{a\}}$,
$\phi_{\>\upop{}}(a) = \upop{\{a\}}$, 
$\phi_{\>\parikhnoarg}(a) = \{a\}$,
on the other hand if $a$ is of the form $(n, p, q)$ i.e. 
$a\in\Pi\setminus \AlB$ then
$\phi_{\>\downop{ }}( (n, p, q) ) = \downop{\langop[n]{\A^{p,q}}}$,
$\phi_{\>\upop{}}( (n, p, q) ) = \upop{\langop[n]{\A^{p,q}}}$,
$\phi_{\>\parikhnoarg}( (n , p, q)) = \parikh{\langop[n]{\A^{p,q}}}$.

Then the following equalities hold
\begin{itemize}
\item $\downop{\langop{\A}}=\phi_{\>\downop{}}(\langop{\B})$,
\item $\upop{\langop{\A}}=\phi_{\>\upop{}}(\langop{\B})$,
\item $\parikh{\langop{\A}}=\phi_{\>\parikhnoarg}(\langop{\B})$.
\end{itemize}
\end{lemma} 
\begin{proof}
Lemma contains three claims which have very similar proofs, thus we present 
only one of them, for downward closure.

The first inclusion. 
Let $w'\in \downop{\langop{\A}}$ thus there is $w\in \langop{\A}$ such that
$w'\subword w$. Due to Lemma~\ref{lem:constuctionOfSimpleAutomata} we know
that there is a word $u=u_1\ldots u_k$, where $u_i\in \Pi$, 
such that $w\in \phi(u)$ ($\phi$ 
is defined in the proof of Lemma~\ref{lem:constuctionOfSimpleAutomata}). 
We claim
that $w'\in \phi_{\>\downop{}}(u)$. Indeed, observe that 
$\phi_{\>\downop{}}(u) = \phi_{\>\downop{}}(u_1) \cdots \phi_{\>\downop{}}(u_k) = 
\downop{\phi(u_1)} \cdots \downop{\phi(u_k)}$; further recall that if 
$\lang=\lang_1\cdot \lang_2$
implies $\downop{\lang} = \downop{\lang_1} \cdot \downop{\lang_2}$ so
we conclude that $\phi_{\>\downop{}}(u) = 
\downop{(\phi(u_1) \cdots \phi(u_k))} = \downop{\phi(u)}\supseteq \downop{w} 
$. As $w'\in \downop{w}$, this gives us what we need, so 
$w'\in\phi_{\>\downop{}}(u)$.  

In the opposite direction. Let $w'\in \phi_{\>\downop{}}(\langop{\B})$.
Then $w' \in \phi_{\>\downop{}}(u)$ for some $u\in \langop{\B}$.
Using similar calculation like previously we can show that
$w' \in \downop{\phi(u)}$ but this mean that $w' \in \downop{w}$ for 
some $w\in \phi(u)$. Now as we know that $\phi(u)\in \langop{\A}$ we
have that $w'\in \downop{\langop{\A}}$ which ends the proof.

\end{proof}

Finally we can prove Lemma~\ref{app:lem:simple-approx}.
\begin{proof}
We start from Lemma~\ref{lem:homo-image}. We take automaton $B$ and 
choose suitable substitution $\phi_{\diamondsuit}$, where $\diamondsuit\in \{\upop{}, \downop{}, 
\parikhnoarg{}\}$, depending on the 
abstraction that we are 
interested in. The idea is to substitute every transition in $B$, say 
$\tau=(p, a , 0, q)$, by a suitable 
designed automaton $\B^{p,q, \diamondsuit}=(Q_{p,q}, \AlB, \qinit, \delta_{p,q}, \{\qfinal\})$
 that accepts the language $\phi_{\diamondsuit}(a)$.

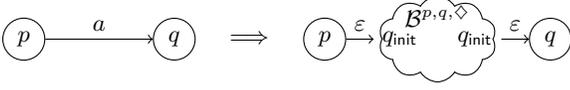
\begin{figure}
\begin{tikzpicture}[every state/.style={minimum size=10pt}]
\node[state]    			(p1) at (0,0) {$p$};
\node[state]                            (q1) at (2,0) {$q$};
\node 					(X) at (3,0) {$\implies$};
\node[state]                            (p2) at (4,0) {$p$};
\node[]				(r2) at (5,0) {$\qinit$};
\node [cloud, draw,cloud puffs=10,cloud puff arc=120, aspect=1.6, inner ysep=1em]  (B) at (5.5, 0) {};
\node []				(B1) at (5.5, 0.3) {$\B^{p,q, \diamondsuit}$};
\node[]				(r22) at (6,0) {$\qinit$};
\node[state]			        (q2) at (7,0) {$q$};
\path[->] (p1) edge [above]      node {$a$} (q1)
          (p2) edge [above]      node {$\eps$} (r2)
          (r22) edge [above]      node {$\eps$} (q2);
\end{tikzpicture} 
\caption{The substitution operation}
\end{figure}

If $a\in \Pi\setminus \AlB$ then automaton is $\B^{p,q}$, else if
$a \in \AlB$ then automaton that we glue is one of following: 

\begin{tikzpicture}[every state/.style={minimum size=14pt, inner sep=0pt}]
\node                           (p) at (0, 0) {for $\phi_{\>\upop{}}$};
\node[state]                            (p2) at (4,0) {\scriptsize{$\qinit$}};
\node[state]			        (q2) at (6,0) {\scriptsize{$\qfinal$}};
\path[->] (p2) edge [above]      node {$a$} (q2)
          (p2) edge [loop above] node {$\AlB$} (p2)
          (q2) edge [loop above] node {$\AlB$} (q2);
\end{tikzpicture},

\begin{tikzpicture}[every state/.style={minimum size=14pt, inner sep=0pt}]
\node                           (p) at (0, 0) {for $\phi_{\>\downop{}}$};
\node[state]                            (p2) at (4,0) {\scriptsize{$\qinit$}};
\node[state]			        (q2) at (6,0) {\scriptsize{$\qfinal$}};
\path[->] (p2) edge [above, bend left]      node {$a$} (q2)
          (p2) edge [above, bend right] node {$\eps$} (q2);
\end{tikzpicture},

\begin{tikzpicture}[every state/.style={minimum size=14pt, inner sep=0pt}]
\node                           (p) at (0, 0) {for $\phi_{\>\parikhnoarg}$};
\node[state]                            (p2) at (4,0) {\scriptsize{$\qinit$}};
\node[state]			        (q2) at (6,0) {\scriptsize{$\qfinal$}};
\path[->] (p2) edge [above]      node {$a$} (q2);
\end{tikzpicture}.

It is obvious that this construction provides a required NFA. The last 
question is about its size. According to 
Lemma~\ref{lem:constuctionOfSimpleAutomata} we have that $\sizeop{B}=
\sizeA\cdot K$. So to estimate the size of the final construction
we need to add at most 
$\sizeop{\AlB} \cdot 2\sizeop{\B} + 
(\sizeA)^2 \cdot g(\sizeA)$ states, which proves the polynomial
bound on the size of $\B^{\diamondsuit}.$
\end{proof}

\begin{remark}
 It is easy observation that above construction can be preformed in
 polynomial time in $K, \sizeA$ and time needed to construct all 
 automata $\B^{p,q, \diamondsuit}$ where 
 $\diamondsuit\in \{\upop{}, \downop{}, \parikhnoarg\}$.
\end{remark}

\renewcommand{\onlyappendix}[1]{#1}
\renewcommand{\onlyappendixpara}[1]{#1}
\renewcommand{\onlymainpaper}[1]{}
\renewcommand{\onlymainpaperpara}[1]{}

\section{Parikh image: Bounded alphabet}
\label{a:s:parikh-bounded}

The result of this section is the following theorem.

\begin{theorem}
\label{a:th:parikh-bounded}
For any fixed alphabet \AlB there is a polynomial-time algorithm that,
given as input a one-counter automaton over \AlB with $n$~states, computes
a Parikh-equivalent NFA.
\end{theorem}

Note that in Theorem~\ref{a:th:parikh-bounded} the size of the alphabet \AlB is fixed.
The theorem implies, in particular, that
any one-counter automaton over \AlB with $n$~states
has a Parikh-equivalent NFA of size $\poly_{\AlB}(n)$,
where $\poly_{\AlB}$ is a polynomial of
degree bounded by $f(\dimension)$ for some computable function $f$.

The numbering sequence from the conference version is mapped as follows:

\begin{itemize}
\item Theorem \ref{conf:th:parikh-bounded} is Theorem \ref{a:th:parikh-bounded},
\item Definition \ref{conf:def:bi-walk} is Definition \ref{def:bi-walk},
\item Definition \ref{conf:def:direction} is Definition \ref{def:direction},
\item Definition \ref{conf:def:available} is Definition \ref{def:available},
\item Lemma \ref{conf:l:more} is Lemma \ref{l:more},
\item Definition \ref{conf:def:unpumping} is Definition \ref{def:unpumping},
\item Lemma \ref{conf:l:unpump-summary} is Lemma \ref{l:unpump-summary},
\item Lemma \ref{conf:l:parikh} is Lemma \ref{l:parikh},
\item Lemma \ref{conf:l:tuple} is Lemma \ref{l:parikh-membership},
\item Lemma \ref{conf:l:pair} is Lemma \ref{l:parikh-membership:walk}.
\end{itemize}

\begin{remark}
We start with yet another simplifying assumption (in addition to
that of subsection~\ref{s:prelim:simple}). In one part of our proof
(subsection~\ref{a:s:parikh-bounded:algo} below) we will need to rely
on the fact that short input words can only be observed along short runs;
this is, of course, always true if all OCA in question have to $\eps$-transitions.
We note here that, for the purpose of computing the Parikh image,
we can indeed assume without loss of generality that this is the case.
To see this, replace all $\eps$ on transitions of an OCA \A
with a fresh letter $e \not\in \Sigma$;
i.e., increase the cardinality of $\Sigma$ by~$1$.
Now construct an NFA Parikh-equivalent to the new OCA over the extended alphabet;
it is easy to see that replacing all occurrences of $e$ in the NFA
by $\eps$ will give us an appropriate NFA.
In this NFA $\eps$-transitions can be eliminated in a standard way.
\end{remark}

\subsection{Basic definitions}

We start from a sequence of definitions which are necessary to describe
pumping schemas for one-counter automata, that are crucial to capture
the linear structure of the Parikh image of the language accepted by 
a given one-counter automaton.

\begin{definition}[attributes of runs and walks]
\label{def:attributes}
For a run 
$\pi=(p_0,c_0), t_1, (p_1,c_1), t_2, \ldots, t_m, (p_m, c_m)$
or a walk (if it has sense)\footnote{To define attributes for walk we take a quasi run such that its sequence of transitions is equal to the walk. All attributes except of $\initcounter{}, \finalcounter{}, \high{}, \low{}$
are well defined for walk, as their value is purely determined by a sequence of transitions.}
we define the following attributes:
\begin{gather*}
\begin{aligned}
\length\pi       &= m,   &\text{(\df{length})}\\
\initstate\pi    &= p_0, &\text{(\df{initial control state})}\\
\finalstate\pi   &= p_m, &\text{(\df{final control state})}\\
\initcounter\pi  &= c_0, &\text{(\df{initial counter value})}\\
\finalcounter\pi &= c_m, &\text{(\df{final counter value})}\\
\end{aligned}
\\
\begin{aligned}
\high\pi   &= \max \{c_i \mid 0 \le i \le m\}, \\
\low\pi    &= \min \{c_i \mid 0 \le i \le m\}, \\
\drop\pi   &= \initcounter\pi - \low\pi, \\
\height\pi &= \high\pi - \low\pi, \\
\effect\pi &= \finalcounter\pi - \initcounter\pi.
\end{aligned}
\end{gather*}
We also use the following terms for quasi runs and walks:
\begin{itemize}

\item
the \df{induced word} is $w =a_1a_2\ldots a_m \in \AlB^*$ where
$t_i=(p_{i-1},a_i, s, p_i) \in \tran$ with $a_i \in \AlB \cup \{\eps\}$,

\item
the \df{Parikh image}, denoted $\parikh{\pi}$, is
the Parikh image of the induced word.

\end{itemize}
\end{definition}

Note that a run of length~$0$ is a single configuration.

\begin{definition}[concatenation]
\label{def:concat}
The \df{concatenation} of a run $\pi_1$ and a quasi run $\pi_2$
\begin{align*}
\pi_1 &=
    (p_0,c_0), t_1, (p_1,c_1), t_2, \ldots,
    t_m, (p_m, c_m) \text{\ and} \\
\pi_2 &=
    (q_0,\bar c_0), \bar t_1, (q_1,\bar c_1), \bar t_2, \ldots,
    \bar t_k, (q_k, \bar c_k)
\end{align*}
where $p_m=q_0$ and 
$c_m \eqdef \finalcounter{\pi_1} \geq \drop{\pi_2}$ is
the sequence
\begin{multline*}
\pi_3 =
    (p_0,c_0), t_1, (p_1,c_1), t_2, \ldots,
    t_m, (p_m, c_m),\\
    \bar t_1, (q_1,\bar c_1'), \bar t_2, \ldots,
    \bar t_k, (q_k, \bar c_k')
\end{multline*}
where $\bar c_i'= \bar c_i - \bar c_0 + c_m$;
note that this sequence $\pi_3$ is a run.
The concatenation of $\pi_1$ and $\pi_2$
is denoted by $\pi_3=\pi_1 \cdot \pi_2$;
we also write $\pi_3 = \pi_1 \pi_2$ when we want
no additional emphasis on this operation.
\end{definition}

Note that in our definition of concatenation the value
$\finalcounter{\pi_1}$ can be different from $\initcounter{\pi_2}$,
in which case the counter values all configurations in $\pi_2$
are adjusted accordingly.
The condition $\finalcounter{\pi_1} \ge \drop{\pi_2}$ ensures
that all the adjusted values stay non-negative, i.e., that
the concatenation $\pi_1 \cdot \pi_2$ is indeed a run.

We extend  Definition~\ref{def:concat} of concatenation to a concatenation of a 
run $\pi_1$ and a walk $\alpha$. Let $\alpha$ be a sequence of transitions
in some quasi-run $\pi_2$. The concatenation $\pi_1$ and $\alpha$ is allowed 
only if $\pi_1\cdot \pi_2$ is well-defined, and the effect of
$\pi_1\cdot \alpha \eqdef \pi_1 \cdot \pi_2 $, so it is a run.
To sum up, we can concatenate runs, quasi-runs, and walks, using
the notation $\pi_1 \cdot \pi_2$ and sometimes dropping the dot.
If $\pi_2$ is a walk and $\pi_1$ is a run, then $\pi_1 \cdot \pi_2$
will also denote a run.
In this and other cases, we will often assume that the counter values in $\pi_2$
are picked or adjusted automatically to match the last configuration of $\pi_1$.
However, whenever we introduce parts of the run,
e.g., by writing ``suppose $\pi = \pi_1 \cdot \pi_2$'', we always
assume that $\pi_2$ is just a sub-run of $\pi$, that is,
no implicit shifting occurs in this particular concatenation.

We say that
a run $\pi_2$ is \emph{in} 
$\pi$ if $\pi=\pi_1\pi_2\pi_3$ for some runs 
$\pi_1,\pi_3$ and $\finalcounter{\pi_1}=\initcounter{\pi_2}$, $\finalcounter{\pi_3}=\initcounter{\pi_3}$.

We say that runs
$\pi_1,\pi_2,\ldots, \pi_k$ are \emph{disjoint} in $\pi$
if $\pi=\pi_1'\pi_1\pi_2'\pi_2\pi_3'\ldots$ $
\pi_k'\pi_k\pi_{k+1}'$
for some runs $\pi_1',\pi_2'\dots \pi_{k+1}'$, where $\finalcounter{\pi_i'} = \initcounter{\pi_i}$ and
$\finalcounter{\pi_i} = \initcounter{\pi_{i+1}'}$ for all $1\leq i\leq k.$

\subsection{Semilinear representation of $\parikh{\langop\A}$}
\label{a:s:parikh-bounded:algo}

\begin{definition}[split run]
\label{def:bi-walk}
A \emph{\biwalk} is a pair of runs
$\bi{\rho}{\sigma}$ such that
$\effect{\rho}\geq 0$ and
$\effect{\sigma} \le 0$.
\end{definition}

In fact, we can even drop these inequalities in the definition,
but we believe it's more visual this way.

\begin{definition}[direction]
\label{def:direction}
A \emph{direction} is a pair of walks $\alpha$ and $\beta$,
denoted $\direction = \dr\alpha\beta$, such that:
\begin{itemize}
\item $\initstate{\alpha} = \finalstate{\alpha},$
\item $\initstate{\beta} = \finalstate{\beta},$
\item $0 < |\alpha| + |\beta| < \polyvi$,
\item $0 \le \effect{\alpha} \le \polyiii$,
\item $\effect{\alpha} + \effect{\beta} = 0$,
\item if $\effect{\alpha} = 0$, then either 
      $|\alpha| = 0$ or $|\beta| = 0$.
\end{itemize}
\end{definition}

One can think of a direction as
a pair of short loops with zero total effect on the counter.
Pairs of words induced by these loops are sometimes known as iterative pairs.
Directions of the first kind are essentially just individual loops;
in a direction of the second kind,
the first loop increases and the second loop decreases the counter value
(even though the values $\drop\alpha$ and $\drop\beta$
 are allowed to be strictly positive).
The condition that $\effect\alpha \le \polyiii$ is a pure technicality
and is only exploited at a later stage of the proof;
in contrast, the upper bound $|\alpha| + |\beta| < \polyvi$ is crucial.

We also use the following terms:
\begin{itemize}

\item
the Parikh image of a \biwalk $\bi{\rho}{\sigma}$ is
$\parikh{(\rho,\sigma)}\eqdef \parikh{\rho}+\parikh{\sigma}$,

\item
similarly, the Parikh image of a direction $\dr{\alpha}{\beta}$ is
$\parikh{\dr{\alpha}{\beta}} \eqdef \parikh{\alpha}+\parikh{\beta}$,

\item
\biwalks $\bi{\rho_1}{\sigma_1}\ldots\bi{\rho_k}{\sigma_k}$ 
are \emph{disjoint in} a run $\pi$ iff
$\rho_1, \sigma_1, \ldots, \rho_k, \sigma_k$ form
disjoint subsequences of transitions in the sequence $\pi$;

\end{itemize}

Remark: the number of directions can be exponential,
but the number of their Parikh images is at most $(\polyvi)^\dimension$.
Since $\dimension$ is fixed, this is polynomial in $n$.

\begin{definition}[availability of directions]
\label{def:available}
Suppose $\pi$ is a run.
A direction $\direction = \dr\alpha\beta$ is \df{available} at $\pi$
if there exists a factorization $\pi=\pi_1 \cdot \pi_2 \cdot \pi_3$ such that
$\pi'=\pi_1 \cdot \alpha \pi_2 \beta \cdot \pi_3$
is also a run.
\end{definition}

In the context of the definition above we write
$\pi+\direction$ to refer to $\pi'$.
This plus operation is non-commutative and binds from the left.
Whenever we use this notation,
we implicitly assume that the direction $\direction$ is available at $\pi$.

Note that for a particular run $\pi$ there can be more than one
factorization of $\pi$ into $\pi_1, \pi_2, \pi_3$
such that $\pi_1 \alpha  \pi_2  \beta  \pi_3$ is a valid run.
In such cases the direction $\direction$ can be introduced at different points
inside $\pi$.
In what follows we only use the notation $\pi+\direction$ to refer to a single
run $\pi'$ obtained in this way,
without specifying a particular factorization of $\pi$.
We ensure that all statements that refer to $\pi+\direction$ hold regardless of
which factorization is chosen of $\pi$.

\begin{lemma}[characterization of availability]
\label{l:avalDueToStates}
A direction $\dr{\alpha}{\beta}$ is available at
a run $\pi$ if and only if $\pi$ has two configurations
$(\initstate\alpha, c_1)$ and $(\initstate\beta, c_2)$,
occurring in this particular order, such that
$c_1 \ge \drop{\alpha}$ and
$c_2 + \effect{\alpha} \ge \drop{\beta}$.
\end{lemma}

\begin{proof}
Denote $p = \initstate\alpha$ and $q = \initstate\beta$.
It is immediate that for a direction $\dr{\alpha}{\beta}$ to be
available, it is necessary that $\pi$ have some configurations
of the form $(p, c_1)$ and $(q, c_2)$ that occur in this order.
Let $\pi=\pi_1 \pi_2 \pi_3$ where $\pi_1$
ends in $(p,c_1)$ and $\pi_2$ ends in $(q,c_2)$. 
We now show that the direction $\dr{\alpha}{\beta}$ is available
if and only if
$c_1 \ge \drop{\alpha}$ and
$c_2 + \effect{\alpha} \ge \drop{\beta}$.

We first suppose that these two inequalities hold.
Observe that
$\pi_1\alpha$ is then a run as $c_1 \geq \drop{\alpha}$;
furthermore,
$\finalcounter{\pi_1\alpha} = c_1 + \effect\alpha \geq c_1$:
by our definition of a direction, we have $\effect\alpha \ge 0$.
Hence,
$\finalcounter{\pi_1\alpha} \geq c_1 \ge \drop{\pi_2}$,
where the last inequality holds because $\pi_1 \pi_2$ is a run.
It follows that $\pi_1\alpha\pi_2$ is also a run;
we note that
$\finalcounter{\pi_1 \alpha \pi_2} =
 \finalcounter{\pi_1 \pi_2} + \effect\alpha =
 c_2 + \effect\alpha \ge \drop\beta$. This, in turn, implies
that $\pi_1\alpha\pi_2\beta$ is
a run. Moreover $\effect{\alpha}=-\effect{\beta}$ by our
definition of a direction, so
$\finalcounter{\pi_1\alpha\pi_2\beta}=\finalcounter{\pi_1\pi_2} = c_2$.
Since $\pi_1 \pi_2 \pi_3$ is a run,
we have $c_2 \ge \drop{\pi_3}$;
hence, $\pi_1\alpha\pi_2\beta\pi_3$ is also a run.

Conversely, suppose $\pi_1 \alpha \pi_2 \beta \pi_3$ is a run.
Recall that $\pi_1$ ends in $(p, c_1)$;
we conclude that $c_1 \ge \drop\alpha$.
Also recall that in the run $\pi_1 \pi_2 \pi_3$ the fragment $\pi_2$
ends in $(q, c_2)$; here $c_2 = \finalcounter{\pi_1 \pi_2}$.
But $\pi_1 \alpha \pi_2 \beta \pi_3$ is a run,
so $\finalcounter{\pi_1 \alpha \pi_2} \ge \drop\beta$;
since $\finalcounter{\pi_1 \alpha \pi_2} = \finalcounter{\pi_1 \pi_2} + \effect\alpha$,
we also conclude that $c_2 + \effect\alpha \ge \drop\beta$.
This completes the proof.
\end{proof}

By $\avail{\pi}$ we denote the set of all directions available at $\pi$.

\begin{lemma}[monotonicity of availability]
\label{l:more}
If $\pi$ is a run of a one-counter automaton
and $\direction$ is a direction available at $\pi$, then
$\avail{\pi} \subseteq \avail{\pi+\direction}$.
\end{lemma}

\begin{proof}
By Lemma~\ref{l:avalDueToStates} it suffices to show that for every pair of configurations
$(p,c_1), (q,c_2)$ that appear in $\pi$ in this particular order there is a pair of
configurations $(p,c_1'), (q,c_2')$ that appear in $\pi'=\pi+\direction$, in this particular
order, such that $c_1'\geq c_1$ and $c_2'\geq c_2$.
Now, this claim is not difficult to substantiate.
Suppose $\direction = \dr{\alpha}{\beta}$. Indeed, for any decomposition of $\pi=\pi_1\pi_2\pi_3$ 
such that there is a 
$\pi'=\pi_1\alpha\pi_2\beta\pi_3$, define 
$\pi_1'=\pi_1$, 
$\pi_2'=(\initstate{\pi_2},\initcounter{\pi_2}+\effect{\alpha})\pi_2$,
$\pi_3'=\pi_3$. Now $\pi_2'$ is simply $\pi_2$ shifted up,
as $\effect{\alpha}\geq0$.
Observe that now $\pi'=\pi_1'\alpha\pi_2'\beta\pi_3'$. 
Thus for any pair of configurations 
in $\pi$ there is a corresponding pair of configurations with the needed properties
in $\pi'$.
This completes the proof.
\end{proof}

\begin{definition}[unpumping]
\label{def:unpumping}
A run $\pi'$ \df{can be unpumped} if
there exist a run $\pi$ and a direction $\direction$
such that $\pi' = \pi + \direction$.

If additionally $\avail{\pi'} = \avail\pi$, then we say
that $\pi'$ \df{can be safely unpumped}. Note that $\avail{\pi'}$
is always a superset of $\avail\pi$ by Lemma~\ref{l:more}.
\end{definition}

\begin{lemma}[safe unpumping lemma]
\label{l:unpump-summary}
Every accepting run $\pi'$ of $\A$ 
of length greater than $\polyMinUnpLen$
can be safely unpumped.
\end{lemma}

\noindent
Lemma~\ref{l:unpump-summary} is the key lemma
in the entire Appendix~\ref{a:s:parikh-bounded};
we prove it in subsection~\ref{a:s:proof-unpump}.

Recall that a set $A \sset \N^\dimension$ is called \df{linear}
if it is of the form
$
\Lin{b}{P} \eqdef \{ b + \lambda_1 p_1 + \ldots + \lambda_r p_r \mid
        \lambda_1, \ldots, \lambda_r \in \N,
        p_1, \ldots, p_r \in P \}
$
for some vector $b \in \N^\dimension$ and some finite set $P \sset \N^\dimension$;
this vector $b$ is called the \df{base} and vectors $p \in P$ \df{periods}.
A set $S \sset \N^d$ is called \df{semilinear} if it is a finite union of linear sets,
$S = \cup_{i \in I} \Lin{b_i}{P_i}$.
Semilinear sets were introduced by FIXME in 1960s
and have since received a lot of attention in formal language theory
and its applications to verification. They are precisely the sets
definable in Presburger arithmetic, the first-order theory of natural
numbers with addition.
Intuitively, semilinear sets are a multi-dimensional analogue
of ultimately periodic sets in \N.

The following lemma characterizes the semilinear set $\parikh{\langop{\A}}$
through sets of directions available at short accepting runs.

\begin{lemma}
\label{l:parikh}
For any one-counter automaton $\A$,
it holds that
\begin{equation*}
\parikh{\langop\A} =
\bigcup_{\text{\textup{$\length{\pi}\leq small_1$}}}
\Lin{\parikh\pi}{\parikh{\avail\pi}},
\end{equation*}
where the union is taken over all runs of $\A$
of length at most $small_1\eqdef \polyMinUnpLen$.
\end{lemma}

\begin{proof}
Start with the $\supseteq$ part.
It suffices to show that, for any run $\pi$ of \A
and any vector $v \in \Lin{\parikh\pi}{\parikh{\avail\pi}}$,
the set $\langop\A$ contains at least one word with Parikh image~$v$.
Indeed, take such a $v$ and suppose that $v = v_0 + \sum_{i = 1}^{m} \lambda_i v_i$,
where the vector $v_0$ is the Parikh image of the word induced by the run $\pi$,
vectors~$v_1, \ldots, v_m$ are Parikh images of the words induced by
some directions~$\direction_1, \ldots, \direction_m$ available at $\pi$,
and $\lambda_1, \ldots, \lambda_m$ are nonnegative integers.
Lemma~\ref{l:more} ensures that we can form a run
\begin{equation*}
\pi +
\underbrace{\direction_1 + \ldots + \direction_1}_{\text{$\lambda_1$ times}}
+ \ldots +
\underbrace{\direction_m + \ldots + \direction_m}_{\text{$\lambda_m$ times}}.
\end{equation*}
This run induces a word accepted by \A, and the Parikh image of this word is $v$,
as desired.

Now turn to the $\sset$ part and take some vector $v$ in $\parikh{\langop\A}$.
This vector $v$ is the Parikh image of a word in $\langop\A$, which is induced
by some accepting run $\pi_0$ of \A. If the length of $\pi_0$ does not exceed
$\polyMinUnpLen$, there is nothing to prove, so assume otherwise.
By the main lemma~\ref{l:unpump-summary}, $\pi_0$ can be safely unpumped. This implies that
$\pi_0 = \pi_1 + \direction_1$ for some direction $\direction_1$.
Note that the length of $\pi_1$ is strictly less than the length of $\pi_0$;
if it is greater than $\polyMinUnpLen$, then we apply the safe unpumping lemma again:
$\pi_1 = \pi_2 + \direction_2$. We repeat the process until the length
of the run drops to $\polyMinUnpLen$ or below:
\begin{equation}
\label{eq:run-unpump}
\pi_0 = \pi_k + \direction_k + \direction_{k - 1} + \ldots + \direction_1,
\end{equation}
where $\length{\pi_k} \le \polyMinUnpLen$.
Take Parikh images of the words that are induced by the runs
on both sides of~\eqref{eq:run-unpump}:
on the left-hand side, we obtain $v$;
we claim that on the right-hand side
we obtain a vector from $\Lin{\parikh{\pi_k}}{\parikh{\avail{\pi_k}}}$.
Indeed, recall that Lemma~\ref{l:unpump-summary} 
guarantees that all these unpumpings are safe,
i.e., $\avail{\pi_i} = \avail{\pi_{i - 1}}$ for $0 < i \le k$.
Since each direction $\direction_i$ is available at the run $\pi_i$,
it follows that
all the directions $\direction_1, \ldots, \direction_k$ are available
at $\pi_k$, so the Parikh image of the word induced by the run on the right-hand
side of~\eqref{eq:run-unpump}
indeed belongs to $\Lin{\parikh{\pi_k}}{\parikh{\avail{\pi_k}}}$.
But, by our choice above, the run $\pi_k$ has length 
at most $\polyMinUnpLen$.
This concludes the proof.
\end{proof}

Up to now we were focused on building some semilinear
representation of $\parikh{\langop{\A}}$. Our next goal is
to improve this representation, before we start we
introduce two useful notions. For a given vector $v$
we define
$\infini{v}\eqdef \max \{ v(a) \mid a\in \AlB \}$ and
$\norm{v}\eqdef \sum_{a \in \AlB} v(a)$. 
For a set of vectors $F$ by $\infini{F}$
we denote $\max\{\infini{v} \mid v \in F\}$.
We also denote the cardinality of a finite set $F$ by
$\card{F}$.

The following lemma uses results
from Huynh~\cite{Huynh85IC}
and in Kopczy\'nski and To~\cite{KT10},
which rely on the Carath\'eodory theorem for cones in
a multi-dimensional space. Essentially, the underlying idea
is that if a vector is a linear combination of more than
$\dimension$ vectors in a $\dimension$-dimensional space, then,
by using linear dependencies, one can reduce this number
to just $\dimension$.
For non-negative integer combinations, the situation
is slightly more complicated, but the same idea
can be carried through.

\begin{lemma}
\label{l:semilinear-reduce}
Let $S \sset \N^\dimension$ be a semilinear set
with representation
$
S = \bigcup_{i \in I} \Lin{c_i}{P_i}
$
and suppose $M \in \N$ is such that
$\infnorm{c_i}, \infnorm{P_i} \le M$ for all $i \in I$.
Then $S$ also has a representation
$
S = \bigcup_{j \in J} \Lin{b_j}{Q_j}
$,
where
$\card{J} \le (M+1)^{\poly(\dimension)}$ and
for each $j \in J$ there exists an $i \in I$ such that
the following conditions hold:
\begin{itemize}
\item $\Lin{b_j}{Q_j} \sset \Lin{c_i}{P_i}$,
\item $Q_j$ is a linearly independent subset of $P_i$, and
\item
$\infini{b_j}\leq (M+1)^{\poly(\dimension)}$.
\end{itemize}
\end{lemma}

\begin{proof}
The case $\card{I} = 1$ appears in
Huynh~\cite[Lemma~2.8]{Huynh85IC}
and in Kopczy\'nski and To~\cite[Theorem~6]{KT10}.
In these results, upper bounds on $\card{J}$ are only stated implicitly and only for $I = 1$.
We will rely on results of Kopczy\'nski and To,
as it seems simpler to extract an explicit upper bound from their arguments.
Their results can be stated as follows:
If $S= \Lin{c_1, P_1}$, then
$S$ has also a representation
$
S = \bigcup_{j \in J_1} \Lin{b_j}{Q_j}
$ such that
for each $j \in J_1$ 
the following conditions hold:
\begin {itemize}
 \item $\Lin{b_j}{Q_j} \sset \Lin{c_1}{P_1}$,
\item $Q_j$ is a linearly independent subset of $P_1$, and
\item $\infini{b_j}\leq (M+1)^{\poly(\dimension)}$.
\end {itemize}
In fact, the precise bound we get from~\cite[Theorem~6]{KT10} is
as follows:
\begin{equation*}
    \infini{b_j}\leq (2M+1)^{\dimension}\cdot (M^\dimension\cdot \dimension^{\frac{\dimension}{2}})^2. 
\end{equation*}
Observe that if $S= \bigcup_{i \in I} \Lin{c_i, P_i}$, then we can transfer the above result 
directly:
$S= \bigcup_{i \in I} \bigcup_{j\in J_i}\Lin{b_j}{Q_j} = \bigcup_{j \in J} \Lin{b_j}{Q_j}$
and the above bounds are transfered as well.

So the only remaining part is to prove that $\card{J }\le (M+1)^{\poly(\dimension)}$.
To do this, we start from the observation that every element $j\in J$ is uniquely
determined by the pair $(b_j, Q_j)$, so $\card{J}$ is bounded by
the cardinality of a set of all possible pairs $(b, Q)$ where
\begin{itemize}
 \item $b\in \N^\dimension$ and $\infini{b}\leq (M+1)^{\poly(\dimension)}$
 \item $Q$ is an $r$-tuple of vectors $(v_1,\ldots v_r)$
 where $r \le \dimension$ and each $v_i\in \N^\dimension$ and $\infini{v_i}\leq M$.
\end{itemize}
Due to the above characterization, we can bound the number of
possible pairs $(b, Q)$ by 
\begin{equation*}
    ((M+1)^{\poly(\dimension)})^{\dimension} \cdot ((M+1)^{\dimension})^{\dimension},
\end{equation*}
which is again a polynomial in $M$, because $\dimension$ is fixed.
Thus, $\card{J}\leq (M+1)^{(\poly(\dimension)+\dimension)\cdot \dimension}$.
This completes the proof.
\end{proof}

\begin{lemma}
\label{l:parikh-small}
For any one-counter automaton $\A$,
it holds that
\begin{equation*}
\parikh{\langop\A} =
\bigcup_{\text{\textup{$\length{\pi}\leq small_2$}}}
\Lin{\parikh\pi}{\parikh{D_\pi}},
\end{equation*}
where
the union is taken over all runs of $\A$
of length at most $small_2\eqdef poly(\polyMinUnpLen)$
and, for each $\pi$,
$D_\pi$ is a subset of $\avail\pi$ of cardinality
at most $\dimension$.
\end{lemma}

\begin{proof}
From Lemma~\ref{l:parikh} we know that \begin{equation*}
\parikh{\langop\A} =
\bigcup_{\text{\textup{$\length{\pi'}\leq small_1$}}}
\Lin{\parikh{\pi'}}{\parikh{\avail{\pi'}}}. 
\end{equation*} 
If we apply to it Lemma~\ref{l:semilinear-reduce} 
we get the desired statement.
It is worth emphasizing
that $D_\pi\subseteq \avail\pi$ due to the following
reasoning:

\noindent First due to Lemma~\ref{l:semilinear-reduce}, 
$\parikh{\pi}\in \Lin{\parikh{\pi'}}{\parikh{\avail{\pi'}}}$ for some 
$\pi'$ and $\pi$ can be obtained from $\pi'$ by pumping some of 
directions in the set $\avail{\pi'}.$ Thus according to 
Lemma~\ref{l:semilinear-reduce} the set $\avail{\pi'}\subseteq\avail{\pi}$ and
consequently (according to Lemma~\ref{l:semilinear-reduce}) as 
$D_{\pi}\subseteq \avail{\pi'}$ we get $D_{\pi}\subseteq \avail{\pi}$.
\end{proof}

\subsection{Computing the semilinear representation}

Below we state sub-procedures used in the algorithm.

\begin{lemma}
\label{l:parikh-membership:walk}
For every fixed \AlB
there is a polynomial-time algorithm for the following task:
given a one-counter automaton \A over \AlB,
two configurations $(q_1, c_1)$ and $(q_2, c_2)$
and a vector $v \in \N^\AlB$ with all numbers written in unary,
decide if \A has a run $\pi = (q_1, c_1) \moves{} (q_2, c_2)$
with $\parikh\pi = v$.
\end{lemma}

\begin{proof}
Our algorithm leverages the unary representation
of $c_1$, $c_2$, and components of $v$.
Define $H_0 = \max(c_1, c_2) + \norm{v} + 1$. Observe that
for any run $\pi = (q_1,c_1) \moves{} (q_2, c_2)$ such that $\parikh{\pi}=v$
the counter value stays below $H_0$, as in one move the counter can not
be changed by more than $1$ and the number of moves is bounded by $\norm{v}$.

The algorithm constructs a multi-dimensional table that
for \emph{all} pairs of configurations $(q'_1, c'_1)$, $(q'_2, c'_2)$
with $c'_1, c'_2 < H_0$ and all vectors $v' \in \{0, \ldots, H_0\}^\AlB$
keeps the information whether there exists a $w'\in \AlB^*$ such that \A has
a run $(q'_1, c'_1) \moves{w'} (q'_2, c'_2)$ where
the counter value stays below $H_0$ and $\parikh{w'}=v'.$
The size of the table is at most $(\sizeA\cdot H_0)^2 \cdot (H_0 + 1)^{\dimension}$,
which is polynomial in the size of the input for a fixed \AlB.

The algorithm fills the entries of the table
using dynamic programming. %
To begin with, runs whose Parikh image is the zero vector only connect
pairs where $(q'_1, c'_1) = (q'_2, c'_2)$.
Now take a non-zero vector $v' \in \{0, \ldots, H_0\}^\AlB$;
a run $(q'_1, c'_1) \moves{w'} (q'_2, c'_2)$
with $\parikh{w'} = v'$ exists if and only if there exists
an intermediate configuration $(\bar q, \bar c)$ such that
$(q'_1, c_1)\moves{\bar w} (\bar q, \bar c)\moves{a} (q_2', c_2')$, where
$\bar w \in \AlB^*$, $a\in \AlB$, and $\parikh{\bar w} + \parikh{a}= v'$;
note that the vector $\parikh{a}$ has $1$ in exactly one component and $0$ in other components,
so $\parikh{\bar w} \in \{0, \ldots, H_0\}^\AlB$ and
$\norm{\parikh{\bar w}} = \norm{v'} - 1$.
This completes the description of the algorithm.
\end{proof}

\begin{corollary}
\label{cor:parikh-bounded:parikh-membership}
For a fixed alphabet \AlB the Parikh membership problem for (simple) one-counter
automata is in \Ptime:
there exists a polynomial-time algorithm that 
takes as an input a (simple) one-counter automaton \A over \AlB with $n$ states and
a vector $v \in \N^\dimension$ with components written in unary,
and outputs some accepting run $\pi$ of $\A$ with $\parikh\pi = v$
if such a run exists or ``none'' otherwise.
\end{corollary}

\begin{lemma}
\label{l:parikh-membership:Cwalk}
For every fixed \AlB
there is a polynomial-time algorithm for the following task:
given a one-counter automaton \A over \AlB,
a sequence of $2\dimension$ configurations $C\eqdef (q_1, c_1),(q_2, c_2)$ $\ldots 
(q_{2\dimension},c_{2\dimension})$
and a vector $v \in \N^\AlB$ with all numbers written in unary,
decide if \A has an accepting run $\pi$ such that
$\parikh\pi = v$ and $\pi$ contains $C$ as a subsequence.
\end{lemma}

\begin{proof}
Observe, that if there is such a run
$\pi$ then configurations in $C$ are cutting $\pi$ into 
$2\dimension+1$ fragments. Each of these fragments
have its own Parikh~image $u_1,u_2\dots u_{2\dimension+1}$ 
that sum up to $v$. Thus the procedure
iterates through all possible partitions of $v$ into 
$2\dimension+1$ vectors, and for each such partition
it checks if there exists a required set of runs: 
between consecutive elements of $C$, a run from the initial configuration
$(q_0,0)$ to the first element of $C$ and run form the last element of $C$ 
to some accepting configuration.
Each of those small test can be done in polynomial time 
using algorithm form Lemma~\ref{l:parikh-membership:walk}.
The crucial fact is that the number of possible 
partitions of $v$ into $2\dimension+1$ vectors
is polynomial in $\norm{v}$ as the dimension is fixed i.e. number of 
partitions is bounded by $(\norm{v}^\dimension)^{2\dimension+1}$.
This implies that the presented procedure works in polynomial time. 
\end{proof}

\begin{lemma}
\label{l:parikh-membership:whitnes}
For every fixed \AlB
there is a polynomial-time algorithm for the following task:
given a one-counter automaton \A over \AlB,
a sequence of $2\dimension$ configurations $C\eqdef  (q_1, c_1),(q_2, c_2)\ldots 
(q_{2\dimension},c_{2\dimension})$
and a vector $v$ with all numbers written in unary,
decide if  in \A there exist a direction $\dr{\alpha}{\beta}$ and 
two configurations $(p_i,c_i),(p_j,c_j)\in C$ where $i\leq j$
such that:
\begin{itemize}
	\item $\parikh{\dr{\alpha}{\beta}}=v$,
	\item $\initstate{\alpha}=p_i,\ \low{\alpha} \leq c_i$ and
	\item $\initstate{\beta}=p_j,\ \low{\beta} \leq c_j+\effect{\alpha}$ 
\end{itemize}
\end{lemma}

\begin{proof}
First observe that there are $\binom{2+(2\dimension)-1}{2}$ possibilities 
for a pair $(p_i,c_i),(p_j,c_j)$ of configurations, thus it suffice 
to provide a polynomial time procedure to check if there exist 
a direction $\dr{\alpha}{\beta}$ for a given pair of configurations. 
Indeed we can iterate 
through all possible pairs and for each of them use the procedure.

Second if there is a direction $\dr{\alpha}{\beta}$ then 
$\effect{\alpha}\leq \norm{v}$; thus in order to check 
existence of $\dr{\alpha}{\beta}$ we proceed as follows:
for every $x\in\{0\ldots \norm{v}\}$ check if there are 
runs $\alpha$ from $(p_i,c_i)$ to $(p_i,c_i+x)$ and $\beta$
from $(p_j,c_j+x)$ to $(p_j,c_j)$ such that $\parikh{\alpha}+\parikh{\beta}=v$.

This question is very close to the question considered in 
Lemma~\ref{l:parikh-membership:walk}, the problem is that we don't
know the partition of $v$ into $\parikh{\alpha}$ and $\parikh{\beta}$.

However, observe that the number of possible partitions is polynomial
(bounded by $\norm{v}^{\dimension}\cdot \norm{v}^{\dimension}$).
Thus basically for every possible partition of $v$ into $v_1$ and $v_2$
we check if there are runs $\alpha$ from $(p_i,c_i)$ to $(p_i,c_i+x)$ and $\beta$
from $(p_j,c_j+x)$ to $(p_j,c_j)$ such that $\parikh{\alpha}=v_1$
and $\parikh{\beta}=v_2$. This can be done in polynomial time 
according to Lemma~\ref{l:parikh-membership:walk}.

In conclusion we iterate though all pairs of configurations in $C$
through all possible effects of $\alpha$ and all possible splittings
of $v$ into $v_1$ and $v_2$; for each such choice we check
if there are runs $\alpha$ and $\beta$ using algorithm from 
Lemma~\ref{l:parikh-membership:walk}. The number of possible choices is
polynomial and for each choice we execute a polynomial time procedure 
so the algorithm works in polynomial time.
\end{proof}

\begin{lemma}
\label{l:parikh-membership}
For every fixed \AlB
there is a polynomial-time algorithm for the following task:
given a one-counter automaton \A over \AlB
and vectors $v, v_1, \ldots$, $v_r \in \N^\AlB$, $0 \le r \le \dimension$,
with all numbers written in unary,
decide if \A has an accepting run $\pi$
and directions $d_1, \ldots, d_r$ available at $\pi$
such that $\parikh\pi = v$ and $\parikh{d_i} = v_i$ for all~$i$.
\end{lemma}
\begin{proof}
First observe that if there is such a run $\pi$ then according 
to Lemma~\ref{l:avalDueToStates} for every 
$v_i\in \{v_1\ldots v_{\dimension}\}$ there is a pair of configurations
such that it makes available a direction $d_i$ where $\parikh{d_i}=v_i$.
Thus for a set of vectors $\{v_1\ldots v_{\dimension}\}$ there is
a sequence $C$ of $2\dimension$ configurations such that it is 
a subsequence of $\pi$ and for every vector $v_i$ there is a pair
of configurations in $C$ which makes some direction $d_i$
available in sense of Lemma~\ref{l:avalDueToStates}.

Second observation is that the counter value of any 
configuration $C$ can not exceed $\norm{v}$; thus the sequence 
$C$ has to be an element of a family of 
$2\dimension$-sequences of configurations
bounded by $\norm{v}$. The size of this family is at most
$(\sizeA\cdot \norm{v})^{2\dimension}.$ 

From above we derive a following procedure.
For every possible choice of the sequence $C$ test:
\begin{itemize}
	\item if there is an accepting run $\pi$ that contains $C$ as a subsequence
	and where $\parikh{\pi}=v$,
	\item if for every $v_i\in \{v_1\ldots v_{\dimension}\}$ there
	is a pair of configurations in $C$ which makes available
	(in sense of Lemma~\ref{l:avalDueToStates})
	some direction $d_i$ such that $\parikh{d_i}=v_i$.
\end{itemize}	
	The first item is handle by algorithm from 
	Lemma~\ref{l:parikh-membership:Cwalk}, the second by 
	Lemma~\ref{l:parikh-membership:whitnes}.

The proposed algorithm is polynomial time as the number of possible 
$C$ is polynomial in the size of the input and for each 
possible $C$ we execute small number of times polynomial
time algorithms from Lemmas~\ref{l:parikh-membership:Cwalk}~and~\ref{l:parikh-membership:whitnes}.
\end{proof}

\begin{lemma}
\label{l:parikh-bounded:construct-semilinear}
For a fixed alphabet \AlB there exists a polynomial-time algorithm that
takes as an input a simple one-counter automaton over \AlB with $n$ states and
outputs (in unary)
an integer $k \ge 0$, vectors $b_i \in \N^\dimension$, and sets of vectors
$P_i \sset \N^\dimension$, $|P_i| \le \dimension$ for $1 \le i \le k$, such that
\begin{equation*}
\parikh{\langop{\A}} = \bigcup_{1 \le i \le k} \Lin{b_i}{P_i}
\end{equation*}
and the following property is satisfied:
for each $i$ there exists an accepting run $\pi$ of \A and directions
$d_1, \ldots, d_r$ available at $\pi$ such that
$\parikh\pi = b_i$ and $\parikh{\{d_1, \ldots, d_r\}} = P_i$.
\end{lemma}

\newcommand{\consth}{poly(\polyMinUnpLen)}
\begin{proof}
Proof bases on two Lemmas~\ref{l:parikh-small}~and~\ref{l:parikh-membership}. 
From the first one we conclude that it suffices to
characterize polynomially many linear sets $l_i$ such that 
$\parikh{\langop{\A}}= \bigcup_{i} l_i$.
According to Lemma~\ref{l:parikh-small} for each linear set 
$l_i=\Lin{b_i}{Q_i}$ holds:
\begin{itemize}
\item there exist an accepting run $\pi_i$ such that $\parikh{\pi_i}=b_i$ and $\norm{b_i}\leq \consth$,
\item the number of elements of $Q_i$ is bounded by $\dimension$ 
\item for every $v_j\in Q_i$ there exist a direction $d_j\in \avail{\pi}$, such that $\parikh{d_j}=v_j$.
\end{itemize}
The last bullet point combined with Definition~\ref{def:direction} of direction gives 
upper-bound on the $v_j\in Q_i$ for any $i$, precisely $\norm{v_j}\leq \polyvi$.
   
Thus in order to compute the linear sets that characterize $\parikh{\langop{\A}}$ we 
iterate through all possible vectors for $b$, where $\norm{b}\leq \consth$,
and all possible sets $Q$,
where $|Q|\leq \dimension$, and for every
$v\in Q$ hold $\norm{v}\leq \polyvi$.
For each combination we check independently if $\Lin{b}{Q}$ satisfies three bullet points.
Second bullet point is satisfied by the definition. To check 
first and third we use the algorithm proposed in  
Lemma~\ref{l:parikh-membership}.

To show that above procedure terminates in polynomial time we observe that
\begin{itemize}
\item number of possible choices for $b$ and $Q$ is bounded by $\consth^\dimension \cdot ((\polyvi)^{\dimension})^{\dimension}$,
\item length of description of $b$ and $Q$ in unary encoding is polynomial in $\A$ (for example can be bounded by
number of possible choices). 
\item Algorithm from 
Lemma~\ref{l:parikh-membership} terminates in time polynomial in the 
input and as input is polynomial in $\sizeA$ then it terminates in time polynomial in $\sizeA$. 
\qedhere
\end{itemize}
\end{proof}

\begin{lemma}
\label{l:parikh-bounded:construct-nfa}
For a fixed alphabet \AlB there exist a polynomial time algorithm that 
takes as an input a one-counter automaton over \AlB with $n$ states and
returns a Parikh-equivalent NFA. 
\end{lemma}

\begin{proof}
Observe that it suffice to show how to change one linear set to NFA; indeed
in the end we can take union of automata designed for polynomialy many
linear sets $l_i$.
To build NFA that accepts a language Parikh equivalent to 
$l_i=\Lin{b}{Q}$ we start 
form building an automaton that accepts only one word which is 
Parikh~equivalent to $b$.
Next to the unique accepting state we add one loop for each $v_i\in Q$. Word
that can be read along $i-th$ loop is Parikh~equivalent to $v_i\in Q$.
It is easy to see that such automaton accepts a 
language Parikh~equivalent with $\Lin{b}{l_i}$.
\end{proof}

Note that DFA instead of NFA would not suffice for this construction,
because even transforming unary NFA into unary NFA induces
a super-polynomial blowup
(a standard example has several cycles whose lengths
 are different prime numbers, with lcm (least common multiple)
 of super-polynomial magnitude).

\subsection{Proof of the main lemma}
\label{a:s:proof-unpump}

In this subsection we prove Lemma~\ref{l:unpump-summary}.
We consider two cases, depending on whether the height (largest counter value)
of $\pi'$ exceeds a certain polynomial in~$n$.
The strategy of the proof is the same for both cases
(although the details are somewhat different).

\begin{lemma}
\label{l:unpump-high}
 Every accepting run $\pi'$ of height greater than $\polyv$
 can be safely unpumped.
\end{lemma}

\begin{lemma}
\label{l:unpump-low}
 Every accepting run $\pi'$ of height at most $\polyv$ and length greater than $\polyMinUnpLen$
 can be safely unpumped.
\end{lemma}

We first show that sufficiently large parts (runs or split runs) of~$\pi'$
can always be unpumped (as in standard pumping arguments).
We notice that for such an unpumping to be \emph{unsafe},
it is necessary that the part contain a configuration whose removal
shrinks the set of available directions---a reason for non-safety;
this \df{important} configuration cannot appear anywhere else in~$\pi'$.
We prove that the total number of important configurations is at most
$\poly(n)$. As a result, if we divide the run~$\pi'$ into
sufficiently many sufficiently large parts, at least one of the parts
will contain no important configurations and, therefore,
can be unpumped safely.

\subsubsection{High runs: Proof of Lemma~\ref{l:unpump-high}}
\label{s:proof-unpump:high}

The idea behind bounding the height of runs bases on two concepts.
First is that if a run is high then there is a direction in it 
that can be unpumped.
Second is that if a given run is even higher then there are a lot of
different directions which can be unpumped and among them at least
one can be unpumped in a safe way. As unpumping intuitively reduces the 
hight of a run, then iterative application of it lead to a 
path of bounded hight.

\begin{claim}
\label{c:high:existenceOfDirection}
Let \biw be a \biwalk such that $\effect{\rho}\ge \polyi$ 
and $-\effect{\sigma}\ge \polyi$, then
$\bi{\rho}{\sigma}=\bi{\rho_1\alpha\rho_2}{\sigma_1\beta\sigma_2}$ 
such that $\dr{\alpha}{\beta}$ is a direction.
\end{claim}

\begin{proof} 
 Let $\rho'$ and $\sigma'$ be a pair of sub-runs of 
 $\rho$ and $\sigma$, respectively, such that
 $\length{\rho'} = \length{\sigma'} = \polyi$;
 such sub-runs exists because
 $\effect{\rho}\ge \polyi$ and $-\effect{\sigma}\ge \polyi$.
 Consider three possibilities:
 \begin{enumerate}
 \renewcommand{\labelenumi}{\theenumi)}
 \item
 there is a non-empty walk $\alpha$ such that $\rho'=\rho_1'\alpha\rho_2'$ 
 that starts and ends in the same configuration, i.e. $\finalcounter{\rho_1'}=\finalcounter{\rho_1'\alpha}$
 and $\initstate{\alpha}=\finalstate{\alpha}$;
 \item
 there is a non-empty walk $\beta$ such that $\sigma'=\sigma_1'\beta\sigma_2'$ 
 that starts and ends in the same configuration, i.e. $\finalcounter{\sigma_1'}=\finalcounter{\sigma_1'\beta}$
 and $\initstate{\beta}=\finalstate{\beta}$;
 \item
 $\effect{\rho'}\ge n^2$ and $-\effect{\sigma'}\ge n^2$.
 \end{enumerate}
 (Note that at least one of these three statements must hold,
 because, for example, the inequality $\effect{\rho'} < n^2$
 implies that the run $\rho'$ traverses
 at most $n^2 \cdot |Q| = n^3$ different configurations;
 however, $\length{\rho'} = n^3$ implies that
 the total number of configurations that $\rho'$ traverses is
 $n^3 + 1$. Hence, by the pigeonhole principle the run $\rho'$
 should traverse some configuration at least twice---%
 which is the first possibility in the list above.)
 For each of these three possibilities,
 we now show how to find some direction $\dr{\alpha}{\beta}$
 inside the split run $\bi{\rho'}{\sigma'}$.

 Consider the first possibility, the direction $\dr{\alpha}{\eps}$ suffices for our purposes.
 Indeed, it is a direction of the first kind as
 $\effect{\alpha} = \effect{\eps} = 0$ and $\length{\eps} = 0$;
 the reader will easily check that all the conditions
 in the definition of a direction (Definition~\ref{def:direction})
 are satisfied.
 The second possibility is completely analogous:
 the direction has the form $\dr{\eps}{\beta}$.

 Now consider the third possibility.
 First, for every $i\in\{0, \ldots, n^2\}$ 
 pick one configuration
 in the run $\rho'$ with the counter value $\initcounter{\rho'}+i$;
 call these configurations \df{red}.
 As there are $n^2 + 1$ red configurations,
 at least $n+1$ of them have the same state;
 we call these $n+1$ configurations \df{blue}.
 The corresponding indices~$i$
 (for which the selected red configuration is also blue)
 are called blue too.
 Now for every blue~$i$ we pick in the other run, $\sigma'$,
 some configuration with the
 counter value equal to $\initcounter{\sigma'}-i$;
 we call these $n + 1$ configurations \df{green}.
 Among green configurations there are at least two with the same state,
 say for $i = i_1$ and $i = i_2$.
 Let $\sigma_{\beta}$ be the run between them contained in $\sigma'$;
 now $\beta$ is a walk induced by $\sigma_{\beta}$.
 But for $\alpha$ we can, in turn, take a walk induced by
 a run
 between blue configurations with indices $i = i_1$ and $i = i_2$.
 By construction, $\effect{\alpha} = - \effect{\beta} > 0$,
 and it is easy to check that $\dr{\alpha}{\beta}$ is indeed
 a direction of the second kind.
 This completes the proof.
\end{proof}

\begin{definition}[promising \biwalk]
\label{def:promisingBiwalk}
A \biwalk \biw in the run $\pi' = \pi_1 \rho \pi_2 \sigma \pi_3$
is \df{promising}
if
$\low{\rho \pi_2 \sigma} \ge \polyi$,
$\effect\rho \ge \polyi$, and
$-\effect\sigma \ge \polyi$.
\end{definition}

\begin{definition}[unpumping a split run]
A \biwalk $\bi{\rho'}{\sigma'}$
in an accepting run $\pi' = \pi_1 \cdot \rho' \cdot \pi_2 \cdot \sigma' \cdot \pi_3$
\df{can be unpumped} if
there exist a \biwalk \biw and a direction $\direction$
such that %
the following conditions hold:
\begin{itemize}
\item $\rho' = \rho_1 \alpha \rho_2$ for some runs $\rho_1, \rho_2$,
\item $\sigma' = \sigma_1 \beta \sigma_2$ for some runs $\sigma_1, \sigma_2$,
\item $\pi = \pi_1 \cdot \rho_1 \rho_2 \cdot \pi_2 \cdot \sigma_1 \sigma_2 \cdot \pi_3$ is an accepting run.
\end{itemize}
One can conclude in such a case that $\pi' = \pi + \direction$.
\end{definition}

\begin{claim}
\label{c:high:promise-unpump}
Any promising \biwalk in an accepting run $\pi'$
can be unpumped.
\end{claim}

\begin{proof}
Let $\pi'=\pi_1\rho\pi_2\sigma\pi_3$ where the split run~\biw is promising.
Note that by the definition of a promising split run,
\biw satisfies the conditions of Claim~\ref{c:high:existenceOfDirection}.
Therefore,
\begin{equation*}
\pi' =
\pi_1 \cdot \rho_1\alpha\rho_2 \cdot \pi_2 \cdot \sigma_1\beta\sigma_2 \cdot \pi_3
\end{equation*}
where $\dr{\alpha}{\beta}$ is a direction,
so it remains to prove that
\begin{equation*}
\pi = \pi_1 \cdot \rho_1\rho_2 \cdot \pi_2 \cdot \sigma_1\sigma_2 \cdot \pi_3
\end{equation*}
is, first, a run and, second, an accepting run.
It is straightforward to see that in this new concatenation
the control states match, so it suffices to check that
the counter values in~$\pi$ stay non-negative;
by Definition~\ref{def:concat}, the (necessary and) sufficient
condition for this is that the $\drop\cdot$ of each subsequent run
does not exceed the $\finalcounter\cdot$ of the prefix.
To simplify notation, we denote $\pi_2' \eqdef \rho_2\pi_2\sigma_1$.

As $\pi_1 \cdot \rho_1$ is a prefix of $\pi'$, we know that it is a run; the
first thing that has to be checked is that $\pi_1 \cdot \rho_1 \pi_2'$ is a run.
This holds if
\begin{align*}
&\finalcounter{\pi_1 \cdot \rho_1}\geq \drop{\pi_2'}, \text{ i.e., if} \\
&\finalcounter{\pi_1 \cdot \rho_1}-\drop{\pi_2'}\geq 0.
\end{align*}
We have
\begin{align*}
&\finalcounter{\pi_1 \cdot \rho_1} - \drop{\pi_2'} \\
&= \finalcounter{\pi_1 \cdot \rho_1\alpha} - \effect{\alpha} - \drop{\pi_2'} \\
&= \finalcounter{\pi_1 \cdot \rho_1\alpha} - \effect{\alpha} - \drop{\rho_2 \cdot \pi_2 \cdot \sigma_1} \\
&= \low{\rho_2 \cdot \pi_2 \cdot \sigma_1} - \effect{\alpha} \\
&\ge \low{\rho \cdot \pi_2 \cdot \sigma} - \effect{\alpha} \\
&\ge \polyi - \effect{\alpha} \\
&\ge \polyi - \polyi = 0,
\end{align*}
where the equalities and inequalities follow from
Definitions~\ref{def:attributes} and \ref{def:direction}
and from the conditions of the claim.
Hence, $\pi_1 \cdot \rho_1 \pi_2'$ is a run.

Since $\effect{\alpha} = -\effect{\beta}$,
the equality
\begin{equation*}
\finalcounter{\pi_1 \cdot \rho_1 \pi_2'} = 
\finalcounter{\pi_1 \cdot \rho_1\alpha \pi_2' \beta}
\end{equation*}
holds.
Therefore,
as $\pi_1\rho_1\alpha\pi_2'\beta\sigma_2\pi_3 = \pi'$ is a run,
$\pi_1\rho_1\pi_2'\sigma_3\pi_3 = \pi$ is a run too. Moreover,
\begin{equation*}
\finalcounter{\pi}=\finalcounter{\pi'},
\end{equation*}
so $\pi$ is an accepting run, for
$\pi'$ is accepting as well.
This completes the proof.
\end{proof}

\begin{definition}[state fingerprint]
Let $\tau$ be a run.
The \df{state fingerprint} of $\tau$
is the set of all pairs $(q_1, q_2) \in Q \times Q$
such that $\tau$ contains configurations
$c_1 = (q_1, r_1)$ and
$c_2 = (q_2, r_2)$ for some $r_1$ and $r_2$,
and, moreover, there exists at least one occurrence of $c_1$
before (possibly coinciding with) some occurrence of $c_2$.
\end{definition}

\begin{claim}
\label{c:high:fingerprint}
Suppose
\begin{align*}
\pi' &= \pi_1 \pi_2 \cdot \alpha \pi_3 \beta \cdot \pi_4 \pi_5 \text{ and } \\
\pi &= \pi_1 \pi_2 \cdot \pi_3 \cdot \pi_4 \pi_5
\end{align*}
are accepting runs and $\dr{\alpha}{\beta}$ is a direction,
so that $\pi' = \pi + \dr{\alpha}{\beta}$.
Also suppose that
\begin{equation*}
\low{\pi_2 \cdot \alpha \pi_3 \beta \cdot \pi_4} \ge \polyiiii.
\end{equation*}
If the runs $\pi_2 \cdot \alpha \pi_3 \beta \cdot \pi_4$
and $\pi_2 \cdot \pi_3 \cdot \pi_4$ have identical state fingerprints,
then $\avail{\pi'} = \avail{\pi}$.
\end{claim}

\begin{proof}
Since $\dr{\alpha}{\beta}$ is a direction,
we have $\effect{\alpha} = -\effect{\beta} \leq \polyi$.
Thus for all configurations observed along 
$\pi_2\pi_3\pi_4$ their counter values are 
at least $\polyiiii-\effect{\alpha} \ge \polyi$.
Now we can use our characterization of availability (Lemma~\ref{l:avalDueToStates}).
Let a direction $\direction$ be available in 
$\pi'$ due to a pair of configurations
$(q_1,c_1),(q_2,c_2)$. 
We have to consider three cases depending on where $(q_1,c_1),(q_2,c_2)$ are:
both configurations are in parts $\pi_1,\pi_5$, both configurations are in 
$\pi_2 \alpha \pi_3 \beta \pi_4$ and one is in $\pi_1$ or $\pi_5$ 
and the second in
$\pi_2 \alpha \pi_3 \beta \pi_4$.
In first case exactly the same pair of configurations can be found in $\pi$.
In the second case due to the assumption about equality of state fingerprints
we can find a pair $(q_1,c_1'),(q_2,c_2')$ where 
$c_1',c_2'\geq \polyi \geq \drop\alpha, \drop\beta$
so $\direction$ is available in $\pi$.
In the third case we have to combine both previous cases.
There are two symmetric situations, first if the pair of configurations
$(q_1, c_1)$ is in $\pi_1$ and $(q_2, c_2)$ in $\pi_2\alpha\pi_3\beta\pi_4$
or $(q_1, c_1)$ is in $\pi_2\alpha\pi_3\beta\pi_4$ and $(q_2, c_2)$ is in
$\pi_5$. Here we consider only the first one of them, the second is analogous.
We need to find a pair of configurations in $\pi$ that witnesses the availability of $d$.
The first element of the pair is the same $(q_1, c_1)$ that can be found in $\pi_1$
as a subrun of $\pi'$. To find the second element, consider this configuration $(q_2, c_2)$
in $\pi'$; it occurs
in $\pi_2\alpha\pi_3\beta\pi_4$, so the state fingerprint contains the pair $(q_2, q_2)$,
simply by definition.
Thus, in $\pi_2\pi_3\pi_4$ we can find a configuration $(q_2, c_3)$ such that
$c_3\geq \polyi$. Now this moves us from the pair of configurations $(q_1, c_1), (q_2, c_2)$
to the pair $(q_1, c_1), (q_2, c_3)$, which completes the proof.
\end{proof}

\begin{claim}
\label{c:high:many-unpump}
Let $\pi=\pi_1\pi_2\pi_3$ be an accepting run.
Suppose the run $\pi_2$ satisfies $\low{\pi_2} \ge \polyiiii$ and
contains $2 n^2 + 1$ pairwise disjoint
promising \biwalks. Then $\pi$ can be safely unpumped.
\end{claim}

\begin{proof}
 From Claim~\ref{c:high:promise-unpump} we know
 that each of these \biwalks can be unpumped;
 so the accepting run $\pi$ can be unpumped in at least $2 n^2 + 1$ different ways.
 Furthermore, by Claim~\ref{c:high:fingerprint},
 if the state fingerprint of $\pi_2$ after unpumping one of these
 \biwalks does not change, then this unpumping is safe.
 Hence, it remains to show that such an unpumping indeed exists.

 For any pair of states in the state fingerprint of
 $\pi_2$, \emph{mark} two configurations of $\pi_2$
 that witness this pair (in sense of Lemma~\ref{l:avalDueToStates}).
 In total,
 at most $2 n^2$ configurations of $\pi_2$ will be marked.
 Since $\pi_2$ contains at least
 $2 n^2+1$ disjoint promising \biwalks, there is a promising \biwalk in $\pi_2$ that
 contains no marked configuration.
 Our choice of this \biwalk ensures that its unpumping will not change
 the state fingerprint of $\pi_2$:
 every pair of states in the state fingerprint of $\pi_2$
 is witnessed by a pair of configurations outside this \biwalk.
 This completes the proof.
\end{proof}

\begin{claim}
\label{c:high:many-promise}
Let $\pi$ be a run of height at least $\polyv$.
Then $\pi = \pi_1\pi_2\pi_3$ for some runs $\pi_1, \pi_2, \pi_3$ such that
$\low{\pi_2} \ge \polyiiii$ and $\pi_2$ contains
at least $2 n^2 + 1$ pairwise disjoint promising \biwalks.
\end{claim}

\begin{proof}
First factorize $\pi$ as $\pi=\pi_1\pi_2\pi_3$ where
$\initcounter{\pi_2}=\polyiiii,\ \finalcounter{\pi_2}=\polyiiii$ and
$\low{\pi_2}=\polyiiii$.
We can now split $\pi_2$ into two parts,
$\pi_2=\pi_2'\pi_2''$, so that
$\finalcounter{\pi_2'}=\initcounter{\pi_2''}=\high{\pi_2}$.
Let $\rho_i$ for $0\le i\le 2n^2+1$ be the sub-run of $\pi_2'$ which
starts at the last configuration of $\pi_2'$ with counter value
$\polyiiii + i \cdot (\polyi) + 1$ and ends at the last configuration
with counter value $\polyiiii + (i + 1) \cdot (\polyi)$.
Similarly, let $\sigma_i$ be the sub-run of $\pi_2''$ which
ends at the last configuration with counter value
$\polyiiii + i \cdot (\polyi) + 1$ and starts at the last configuration
with counter value $\polyiiii + (i + 1) \cdot (\polyi)$.

It is obvious that
$\bi{\rho_i}{\sigma_i}\cap \bi{\rho_j,}{\sigma_j}=\emptyset$
for $i\neq j$ and
that $\bi{\rho_i}{\sigma_i}$ is a promising \biwalk
for any $i \le 2n^2+1$.
\end{proof}

Lemma~\ref{l:unpump-high} follows from Claims~\ref{c:high:many-unpump}
and~\ref{c:high:many-promise}.

\subsubsection{Low runs: Proof of Lemma~\ref{l:unpump-low}}
\label{s:proof-unpump:low}

The idea behind bounding the length of low runs (not high runs) 
bases on similar concept as bounding the hight of runs.
First we show that long enough low run can be unpumped and
next that if the run is even longer then there are a lot different
directions that can be unpumped independently; finally one of them 
can be unpumped in a safe way. Unpumping makes a run shorter 
so any low run that can not be unpumped anymore can not be very long.

\begin{definition}[promising run]
A run $\tau$ is \df{promising}
if $\high{\tau} < \polyv$ and $\length\tau >\polyvi$.
\end{definition}

\begin{claim}
\label{c:low:promise-unpump}
Any promising run $\tau$ can be unpumped.
\end{claim}

\begin{proof}
As the counter value in the run $\tau$ is bounded by $\polyv$ and the length
of $\tau$ is at least $\polyvi $ then among first $\polyvi$ 
configurations at least one configuration
has to repeat.
Thus $\tau=\pi_1\alpha\pi_2$ where $\initstate{\alpha} = \finalstate{\alpha}$
and $\initcounter{\alpha}=\finalcounter{\alpha}$.
As a result, the fragment $\alpha$ between two occurrences of this
repeating configuration can be unpumped, so the outcome of this
operation is a run $\pi_1\cdot \pi_2$. Note that the unpumped direction 
is of the form $\dr{\alpha}{c}$, where $c$ is path of length zero
(a configuration that occurs to the right of $\alpha$).
In this case the $\effect{\alpha} = 0$ and it does
not violates the upper bound on the length.
\end{proof}

Now, instead of the state fingerprint introduced in the previous subsubsection~\ref{s:proof-unpump:high},
we will use configuration fingerprint, defined as follows.

\begin{definition}[configuration fingerprint]
Let $\tau$ be a run.
The \df{configuration fingerprint} of $\tau$
is the set of all pairs of configurations $(c_1, c_2)$
such that in $\tau$
there exists at least one occurrence of $c_1$
before (possibly coinciding with) some occurrence of $c_2$.
\end{definition}

\begin{remark}
Suppose $\tau$ is an accepting run of $\A$ of height at most $\polyv$, then
the cardinality of its configuration fingerprint does not exceed 
$\polyvii \eqdef n^2 (\polyv)^2$.
\end{remark}

\begin{claim}
\label{c:low:fingerprint}
If runs $\tau$ and $\tau'$
have identical configuration fingerprints,
then $\avail{\tau'} = \avail{\tau}$.
\end{claim}

Claim~\ref{c:low:fingerprint} follows from Lemma~\ref{l:avalDueToStates}.

\begin{claim}
\label{c:low:many-unpump}
Any accepting run $\tau$ that satisfies $\high\tau \le \polyv$
and contains, as sub-runs, $2(\polyvii)$ pairwise disjoint
promising runs, can be safely unpumped.
\end{claim}

\begin{proof}
 Let $S$ be the configuration fingerprint of $\tau$. For every element of $x\in
 S$ we choose two configurations that witnesses $x$. Observe that set of chosen
 configurations is of size at most
 $2\cdot |S|\le 2\cdot (\polyvii)$;
 thus there exist a promising run $\sigma$ which does not contain any
 chosen configuration. Now unpumping $\sigma$ does not change set of
 configurations fingerprint, so it is safe unpunping.
\end{proof}

\begin{claim}
\label{c:low:many-promise}
In any accepting run $\tau$ which is a promising run of length at least $2\polyvii \cdot (\polyvi)$
there exist at least $2\polyvii$ pairwise disjoint promising runs.
\end{claim}
\begin{proof}
 Immediate consequence of 
Claim~\ref{c:low:promise-unpump}.
\end{proof}

Lemma~\ref{l:unpump-low} follows from Claims~\ref{c:low:many-unpump}
and~\ref{c:low:many-promise}.

\subsection{Completeness result}
In the proofs, we will use the following fact, which is easy to see.
\begin{lemma}\label{completeness:substitution}
Let $\A$ be an NFA of size $n$ over $\Sigma$ and $\sigma\colon
\Sigma\to\Powerset{\Gamma^*}$ be a substitution of size $m$. Then there is an
NFA for $\sigma(\langop{\A})$ of size $n^2\cdot m$.
\end{lemma}

\begin{proof}[Proof of Lemma~\ref{completeness:rba}]
A \emph{phase} of
$\A$ is \aprerun\  \prerun\ in which no reversal occurs, i.e. which is contained in
$\posTran^*$ or in $\negTran^*$, where
\begin{align*}
\posTran&=\{(p,a,s,q)\in \delta \mid s\in \{0,1\} \}, \\
\negTran&=\{(p,a,s,q)\in \delta \mid s\in \{0,-1\} \}.
\end{align*}
Observe that since $\A$ is $r$-reversal-bounded, every accepting run decomposes
into at most $r+1$ phases. As a first step, we rearrange phases to achieve a
certain normal form.  We call two phases $u$ and $v$ \emph{equivalent} if (i)
$\parikh{u}=\parikh{v}$, and (ii) they begin in the same state and end in the
same state.

If $u=(p_0, a_1, s_1, p_1)(p_1, a_2, s_2, p_2)\cdots (p_{m-1}, a_m, s_m, p_m)$
is a phase, then we write $\Theta(u)$ for the set of all phases 
\begin{equation} (p_0, a_1, s_1, p_1)v_1(p_1, a_2, s_2, p_2)\cdots v_{m-1}(p_{m-1}, a_m, s_m, p_m)\label{flatinsertion} \end{equation}
where for each $1\le i\le m$, we have $v_i=w_{i,1}\cdots w_{i,k_i}$ for some
simple $p_i$-cycles $w_{i,1},\ldots,w_{i,k_i}$.

\emph{Claim:} For each phase $v$, there is a phase $u$ with $|u|\le B:=2n^2+n$
and a phase $v'\in\Theta(u)$ such that $v'$ is equivalent to $v$.

Observe that it suffices to show that starting from $v$, it is possible to
successively delete  factors that form simple cycles such that (i) the set of
visited states is preserved and (ii) the resulting phase $u$ has length $\le
B$: If we collect the deleted simple cycles and insert them like the
$w_{i,j}$ from \ref{flatinsertion} into $u$, we obtain a phase
$v'\in\Theta(u)$, which must be equivalent to $v$.

We provide an algorithm that performs such a successive deletion in $v$. During
its execution, we can mark positions of the current phase with states, i.e.
each position can be marked by at most one state. We maintain the following
invariants. Let $M\subseteq Q$ be the set of states for which there is a marked
position. Then (i) deleting all marked positions results in \aprerun\  \prerun\ (hence a
phase), (ii) If $p\in M$, then there is some position marked with $p$ that
visits $p$, (iii) for each state $p$, at most $n$ positions are marked by $p$,
and (iv) the unmarked positions in $v$ form at most $|M|+1$ contiguous blocks.

The algorithm works as follows. In the beginning, no position in $v$ is marked.
Consider the phase $\bar{v}$ consisting of the unmarked positions in $v$.
If
$|\bar{v}|\le n(n+1)$, the algorithm terminates. Suppose $|\bar{v}|>n(n+1)$.  Since
the unmarked positions in $v$ form at most $n+1$ contiguous blocks, there has to
be a contiguous block $w$ of unmarked positions with $|w|>n$. Then then $w$
contains a simple $p$-cycle $f$ as a factor. 
Note that $f$ is also a factor of
$v$. We distinguish two cases:
\begin{itemize}
\item If deleting $f$ from $v$ does not reduce the set of visited states, we
delete $f$.
\item If there is a state $p$ visited only in $f$. Then, $p\notin M$:
Otherwise, by invariant (ii), there would be a position that visits $p$ and is
marked by $p$ and hence lies outside of $f$. Therefore, we can mark all
positions in $f$ by $p$.
\end{itemize}
These actions clearly preserve our invariants.

Each iteration of the algorithm reduces the number of unmarked positions, which
guarantees termination.  Upon termination, we have $|\bar{v}|\le n(n+1)$, meaning
there are at most $n(n+1)$ unmarked positions in $v$.  Furthermore, by invariant
(iii), we have at most $n^2$ marked positions. Thus, $v$ has length $\le B=2n^2+n$
and has the same set of visited states as the initial phase.  This proves our
claim.

We are ready to describe the OCA $\B$. It has states $Q'=Q\times [0,B]\times
[0,r]$.  For each $j\in [1,r]$, let $\tran_j=\posTran$ is $j$ is even and
$\tran_j=\negTran$ if $j$ is odd. The initial state is $(q_0,0,0)$ and
$(q,B,r)$ is final, where $f$ is the final state of $\A$.  
For each $(i,j)\in
[0,B-1]\times [0,r]$, and each transition $(p,a,s,q)\in\tran_j$, we add the
transitions
\begin{align} 
&((p,i,j),a,s,(q,i+1,j)), \label{rba:plusia} \\
&((p,i,j),\varepsilon,0,(p,i+1,j)). \label{rba:plusib} 
\end{align}

Moreover, for each $p\in Q$, $i\in [0,B]$, and $j\in [0,r-1]$, we include
\begin{equation} ((p,i,j),\varepsilon,0,(q,0,j+1)). \label{rba:plusj} \end{equation}
The input alphabet $\AlB'$ of $\B$ consists of the old symbols $\AlB$ and
the following fresh symbols.  For each $p\in Q$ and $z\in[-n,n]$, we include a
new symbol $a_{p,z}$. Moreover, for each $p\in Q$ and $k\in[0,n]$, we add a
loop transition
\begin{equation} ((p,i,j), a_{p, s\cdot k}, s\cdot k, (p,i,j)), \label{rba:loop} \end{equation}
where $s=(-1)^{j+1}$.
In other words, we add a loop in $(p,i,j)$ that reads $a_{p,s\cdot k}$ and adds
$s\cdot k$ to the counter, where the sign $s$ depends on which phase we are
simulating.  Let us estimate the size of $\B$. It has $n\cdot B\cdot (r+1)$
states. Moreover, for each $k\in[1,n]$, it has $2\cdot n\cdot B\cdot (r+1)$ transitions with
an absolute counter value of $k$. This means, $\B$ is of size
\begin{align*}
 &nB(r+1) + \sum_{k=1}^n (k-1)\cdot 2\cdot nB(r+1) \\
=~&nB(r+1) \cdot (1+2\cdot \frac{1}{2}n(n-1)) \\
=~&n(2n^2+n)(r+1) \cdot (1+n(n-1)) \\
\le~& 3n^3 (r+1)\cdot 2n^2=6n^5 (r+1)
\end{align*}
Furthermore, $\B$ is \anrba\ \rba: If we set $(p,i,j)<(q,\ell,m)$ iff (i) $j<m$
or (ii) $j=m$ and $i<\ell$, then this is clearly a strict order on the states
of $\B$ and every transition of type \eqref{rba:plusia}, \eqref{rba:plusib} or
\eqref{rba:plusj} is increasing with respect to this order.  Hence, a cycle
cannot contain a transition of type \eqref{rba:plusia}, \eqref{rba:plusib} or
\eqref{rba:plusj},  and all other transitions are loops. Thus, $\B$ is \anrba\ \rba.

The idea is now to substitute each symbol $a_{p,z}$ by the regular language of
of $p$-cycles without reversal that add $z$ to the counter. To this end, we
define the NFA $\B_{p,z}$ as follows.  
\begin{itemize}
\item If $z\ge 0$, then $I_{z}=[0,z]$ and let $\tran'_z=\posTran$.
\item If $z< 0$, then $I_{z}=[z,0]$ and let $\tran'_z=\negTran$.
\end{itemize}
$\B_{p,z}$ has states $Q_{p,z}=Q\times I_{z}$, $(p,0)$ is its initial state and
$(p,z)$ its final state. It has the following transitions: For each transition
$(q,a,s,q')\in\tran'_z$, we include $((q,y), a, (q', y+s))$ for each $y\in I_z$
with $y+s\in I_z$. Now indeed, $\langop{\B_{p,z}}$ is the set of inputs of
$p$-cycles without reversal that add $z$ to the counter. Note that $\B_{p,z}$ has
at most $n(n+1)$ states.

Let us now define the substitution $\sigma$. For each $a\in\AlB$, we set
$\sigma(a)=\{a\}$.  For the new symbols $a_{p,z}\in\AlB'\setminus\AlB$, we
define  $\sigma(a_{p,z})=\langop{\B_{p,z}}$. Since each $\B_{p,z}$ has at most  
$n(n+1)$ states, $\sigma$ has size at most $n(n+1)$.

It remains to be shown that
$\parikh{\sigma(\langop{\B})}=\parikh{\langop{\A}}$.  It is clear from the
construction that $\sigma(\langop{\B})\subseteq \langop{\A}$ and in particular
$\parikh{\sigma(\langop{\B})}\subseteq \parikh{\langop{\A}}$. For the other
inclusion, we apply our claim. Suppose $v$ is an accepting \prerun\ of $\A$.
Since $\A$ is $r$-reversal-bounded, $v$ decomposes into $r+1$ (potentially
empty) phases: We have $v=v_0\cdots v_{r}$, where each $v_j$ is a phase.

For each $v_j$, our claim yields phases $u_j$ and $v'_j\in\Theta(u_j)$ such
that $|u_j|\le B$ and $v'_j$ is equivalent to $v_j$. Now each $v'_j$ gives rise
to \aprerun\  \prerun\ $w_j$ of $\B$ as follows. First, each $u_j$ induces a run from
$(p,0,j)$ to $(q,|u_j|,j)$, where $p$ and $q$ is the first and last state of
$u_j$, respectively, via transitions \eqref{rba:plusia}. Now for each simple
$p$-cycle added to $u_j$ to obtain $v'_j$, we insert a transition of type
\eqref{rba:loop}, whose input can later be replaced with the $p$-cycle by
$\sigma$. Now we can connect the \preruns\ $w_0,\cdots,w_r$ via the transitions
\eqref{rba:plusib} and \eqref{rba:plusj} and thus obtain \aprerun\  \prerun\  $w$ of
$\B$. Clearly, applying $\sigma$ to the output of $w$ yields an word that is
Parikh-equivalent to the ouput of $v$. This proves
$\parikh{\langop{\A}}\subseteq \parikh{\sigma(\langop{\B})}$.
\end{proof}

\begin{proof}[Proof of Lemma~\ref{completeness:dyck}]
Assume there is a Dyck sequence $x_1,\ldots,x_n$ for which the statement fails.
Furthermore, assume that this is a shortest one.  We may assume that $x_i\ne 0$
for all $i\in[1,n]$.  Of course we have $n>2r(2N^2+N)$: Otherwise, we could
choose $I=[1,n]$. We define $s_i=\sum_{j=1}^i x_j$ for each $j\in[1,n]$.

Consider 
\[ s_+=\sum_{i\in[1,n],~x_i>0} x_i,~~~~~s_-=\sum_{i\in[1,n],~x_i< 0} x_i.\]
A contiguous subsequence of $x_1,\ldots,x_n$ is called a \emph{phase} if all
its numbers have the same sign.  Suppose we had $s_i\le 2N^2+N$ for every
$i\in[1,n]$. Then every positive phase contains at most $2N^2+N$ elements.
Since we have at most $r$ phases, this means $s_+\le r(2N^2+N)$. However, we
have $s_++s_-\ge 0$ and thus $|s_-|\le |s_+|\le r(2N^2+N)$.  This implies $n\le
2r(2N^2+N)$, in contradiction to above.

Hence, we have $s_i>2N^2+N$ for some $i\in[1,n]$. Choose $r\in[1,i]$ maximal
such that $s_{r+1}\le N^2+N$. Then $s_r\ge N^2$. Similarly, choose $t\in[i,n]$
minimal such that $s_{t-1}\le N^2+N$. Then $s_t\ge N^2$.  Note also that
$s_j\ge N^2$ for $j\in[r,t]$.  Now we have $\sum_{j=r+2}^i x_j \ge N^2$ and
$\sum_{j=i+1}^{t-2} x_j\le -N^2$. Therefore, there is a $u\in[0,N]$ that
appears at least $N$ times among $x_{r+2},\ldots,x_i$ and there is $v\in[-N,0]$
that appears at least $N$ times among $x_{i+1},\ldots,x_{t-2}$.

We can remove $v$-many appearances of the $u$ and $u$-many appearances of the
$v$. Since $s_j\ge N^2$ for $j\in[r,t]$ and we lower the partial sums by at
most $N^2$, this remains a Dyck sequence.  We call it $y_1,\ldots,y_m$.
Moreover, it has the same sum as $x_1,\ldots,x_n$ since we removed $v\cdot u$
and added $u\cdot v$. Finally, it is shorter than our original sequence and
thus has a removable subset $I$ with at most $2r(2N^2+N)$ elements. We denote
sequence remaining after removing the $y_i$, $i\in I$, by $z_1,\ldots,z_p$.
Note that it has sum $0$.
Now, we add the removed appearances of $u$ and $v$
back at their old places into $z_1,\ldots,z_p$, we get a Dyck sequence with sum
$0$ that differs from $x_1,\ldots,x_n$ only in the removed $y_i$, $i\in I$.
Thus, we have found a removable subset with at most $2r(2N^2+N)$ elements, in
contradiction to the assumption.
\end{proof}

\begin{proof}[Proof of Lemma~\ref{completeness:loop-counting}]
Suppose $\A=(Q, \AlB, \tran, q_0, F)$ is \anrba\ \rba\ of size $n$. The idea is
to keep the counter effect of non-loop transitions in the state. This is
possible since $\A$ is \anrba\ \rba\  and thus the accumulated effect of all non-loop
transitions is bounded by $K$.  This means, however, that the counter value we
simulate might be smaller than the value of $\B$'s counter. This means, if the
effect stored in the state is negative and $\B$'s counter is zero, we might
simulate quasi runs with a negative counter value.  That is why faithfully,
we can only smulate runs that start and end at counter value $K$.

We will use the following bounds:
\begin{align*}
N=n^2,~~~M=2N^2+N,~~~K=N+M\cdot n.
\end{align*}
We construct the automaton $\B$ as follows. It has the state set $Q'=Q\times
[-K,K]\times [0,M]$.  Its initial state is $(q_0, 0, 0)$ and all states
$(q,0,m)$ with $q\in F$ and $m\in[0,M]$ are final.
For each non-loop transition $(p,a,s,q)\in\tran$, we include
transitions 
\begin{equation} ((p,k,m), a, 0, (q,k+s,m)), \label{completeness:lc:nonloop} \end{equation}
for all $k\in [-K,K]$ with $k+s\in [-K,K]$ and $m\in[0,M]$. 
In contrast to non-loop transitions,
loop transitions can be simulated in two ways. For each loop transition
$t=(p,a,s,p)$, we include the loop transition:
\begin{equation} ((p,k,m), a, s, (p,k,m)) \label{completeness:lc:loopa} \end{equation}
for each $k\in [-K,K]$ and $m\in[0,M]$, but also the transition
\begin{equation} ((p,k,m), a, 0, (p,k+s,m+1)) \label{completeness:lc:loopb}\end{equation}
for each $k\in[0,K]$ and $m\in[0,M-1]$ with $k+s\in[-K,K]$.

First, we show that this OCA is in fact acyclic. By assumption, $\A$ is
acyclic, so we can equip $Q$ with a partial order $\le$ such that every
non-loop transition of $\A$ is strictly increasing. We define a partial order
$\le'$ on $Q'$ as follows. For $(p,k,m),(p',k',m')\in Q'$, we have
$(p,k,m)\le'(p',k',m')$ if and only if (i) $p<q$ or (ii) $p=q$ and $m<m'$.
Then, clearly, all transitions in $\B$ are increasing with respect to $\le'$.

Note that $\B$ has $n\cdot (2K+1)\cdot (M+1)$ states. Moreover, for each of the
transitions of $\A$ that have positive weight (of which there are at most $n$),
it introduces $(2K+1)\cdot M$ edges, each of weight at most $n$.  In total,
$\B$ has size at most 
\begin{align*} 
& n\cdot (2K+1)\cdot (M+1) + n\cdot (2K+1)\cdot M\cdot n,
\end{align*}
which is polynomial in $n$.

It remains to be shown that $\langop{\A}\subseteq
\langop{\B}\subseteq\langop[K]{\A}$. We begin with the inclusion
$\langop{\A}\subseteq \langop{\B}$. Consider an accepting \prerun\ in $\A$.  It
can easily be turned into a \prerun\ of $\B$ as follows. Instead of a non-loop
transition, we take its counterpart of type \eqref{completeness:lc:nonloop}.
Instead of a loop-transition, we take its counterpart of type
\eqref{completeness:lc:loopa}. Consider the final configuration $((q,k,m),x)$
reached in $\B$.  The sum of $k$ and $x$ is the counter value at the end of
\prerun\ in $\A$, so $k+x=0$. Now it could happen that $x>0$ and $k<0$, in
which case this is not an accepting \prerun\ of $\B$. However, we know that $k$
results from executing non-loop transitions and $\A$ is acyclic, which means
there are at most $n$ of them in a \prerun.  Hence, we have $x=|k|\le n^2=N$.

Consider the loop transitions executed in our \prerun\ and let $x_1,\ldots,x_m$
be the counter values they add to $\B$'s counter. Note that $\sum_{i=1}^m
x_i=x\in [0,N]$.  Now we want to \emph{switch} some of the loops, meaning that
instead of taking a transition of type \eqref{completeness:lc:loopa}, we take
the corresponding transition of type \eqref{completeness:lc:loopb} (including,
of course, the necessary updates to the rightmost component of the state).
Observe that the loop-transitions contain at most $n$ reversals (the
loop-transitions on each state have the same sign in their counter action,
otherwise, the automaton would not be reversal-bounded).

According to Lemma~\ref{completeness:dyck}, there are $\le 2n(2N^2+N)=M$
occurences of loops we can switch such that (i) the resulting \prerun\ still
leaves $\B$'s counter non-negative at all times and (ii) the new \prerun\
leaves $\B$'s counter empty in the end. Since we do this with at most $M$
loop-transitions, the rightmost component of the state has enough capacity. 

Moreover, we do not exceed the capacity of the middle component: Before the
switching, this component assumed values of at most $N$, because there are at
most $n$ non-loop transitions in a run of $\A$. Then, we add at most
$M$ times a number of absolute value $\le n$. Hence, at any point, we have an
absolute value of at most $N+M\cdot n=K$. Thus, we have found an accepting
\prerun\ in $\B$ that accepts the same word. We have therefore shown
$\langop{\A}\subseteq\langop{\B}$.

The other inclusion, $\langop{\B}\subseteq \langop[K]{\A}$, it easy to show:
Whenever we can go in one step from configuration $((q,k,m),x)$ to
$((q',k',m'),x)$ in $\B$ for $q,q'\in Q$, then we can go from $(q, K+x+k)$ to
$(q', K+x'+k')$ in $\A$. Note that then, $K+x+y$ and $K+x'+y'$ are both
non-negative. This implies that $\langop{\B}\subseteq \langop[K]{\A}$.
\end{proof}

\begin{proof}[Proof of Theorem~\ref{completeness:hardaut}]
Suppose $\A$ is a loop-counting \rba\  of size $n$.  By possibly adding a
state and an $\varepsilon$-transition, we may assume that in $\A$, every
run involves at least one transition that is not a loop. After this
transformation, $\A$ has size at most $m=n+1$.

Let $f\in Q$ be the final state of $\A$.  Since $\A$ is acyclic, we may define
a partial order on $Q$ as follows. For $p,q\in Q$, we write $p\le q$ is there
is a (possibly empty) \prerun\ starting in $p$ and arriving in $q$.  Then
indeed, since $\A$ is acyclic, $\le$ is a partial order. We can therefore sort
$Q$ topologically, meaning we can find an injective function $\varphi\colon
Q\to [1,m]$ such that $p\le q$ implies $\varphi(p)\le\varphi(q)$ and
$\varphi(q_0)=1$ and $\varphi(f)=m$.
The function $\varphi$ will help us map
states of $\A$ to states of $\HardAutomaton{2m}$, but it is not quite enough in its current
form: We want to map a state $p$ in $\A$ to a state $q$ in $\HardAutomaton{2m}$ such that
the signs of the counter actions of loops in $p$ and in $q$ coincide. To this
end, we have to modify $\varphi$ slightly.

Consider a state $p$.  Since $\A$ is $r$-reversal-bounded, either all $p$-loops
are non-incrementing or all $p$-loops are non-decrementing. Hence, we may
define $\tau\colon Q\to\{0,1\}$ by
\[ \tau(p) = \begin{cases} 1 & \text{if all $p$-loops are non-decrementing} \\ 0 & \text{if all $p$-loops are non-incrementing} \end{cases} \]
Using $\tau$, we can construct our modification $\chi\colon Q\to[1,2m]$ of $\varphi$. For $p\in Q$, let
\[ \chi(p)=2\cdot \varphi(p) - \tau(p). \]
Note that we may assume that $\tau(q_0)=1$ (otherwise, we could delete the
loops in $q_0$, they cannot occur in a valid run) and thus
$\chi(q_0)=1$. By the same argument, we have $\tau(f)=0$ and hence
$\chi(f)=2m$. Moreover, we still have that $p\le q$ implies
$\chi(p)\le\chi(q)$.

The idea is now to let each symbol $\sCount_{s,k}$ be substituted by all labels
of loops in the state $\chi^{-1}(s)$ that change the counter by $k$. Moreover,
we want to substitute $\sConn_{s,t}$ by all labels of transitions from
$\chi^{-1}(s)$ to $\chi^{-1}(t)$. However, since the loops in $\HardAutomaton{2m}$ always
modify the counter, those loops on $\chi^{-1}(s)$ (or on $\chi^{-1}(t)$) that
do not modify the counter, are generated in the images of $\sConn_{s,t}$.

We turn now to the definition of $\sigma$. Let $S\subseteq [1,2m]$ be the set
$\chi(Q)$ of all $\chi(q)$ for $q\in Q$. Then $\chi\colon Q\to S$ is a
bijection. We begin by defining subsets $\Gamma_{i,j}$ and $\Omega_{i,j}$ of
$\Sigma\cup\{\varepsilon\}$ for $i,j\in[1,2m]$.  Consider $s\in S$ and let
$k\in[0,m]$. Note that all counter actions on transitions in $\A$ have an
absolute value of $\le m$. By $\Gamma_{s,k}$ we denote the set of all
$a\in\Sigma$ such that there is a loop $(\chi^{-1}(s), a, u, \chi^{-1}(s))$ in
$\A$ with $|u|=k$.  For all other indices $i$, $j$, $\Gamma_{i,j}$ is empty.
Furthermore, for $s,t\in S$ with $s<t$, let
\[ \Omega_{s,t} = \{ a\in\Sigma \mid (\chi^{-1}(s),a,0,\chi^{-1}(t))\in\tran \}. \]
Note that if $s\ne t$, all transitions from $\chi^{-1}(s)$ to $\chi^{-1}(t)$
leave the counter unchanged ($\A$ is loop-counting). Again, for all other
choices of $i,j\in[1,2m]$, $\Omega_{i,j}$ is empty. 

We define $\sigma$ as follows.  Let 
\[ \Sigma_{2m}=\{\sCount_{i,j} \mid i,j\in[1,2m] \}\cup \{\sConn_{i,j} \mid i,j\in [1,2m],~i<j\} \]
be the input alphabet of $\HardAutomaton{2m}$.  For
$s,k \in[1,2m]$, let 
\begin{equation} \sigma(\sCount_{i,j})=\Gamma_{i,j}\label{completeness:rega} \end{equation}
and for $s,t\in [1,2m]$, $s<t$, we set
\begin{equation} \sigma(\sConn_{s,t})=(\Gamma_{s, 0})^* \Omega_{s,t} (\Gamma_{t, 0})^* \label{completeness:regb}. \end{equation}
Now it is easy to verify that
$\parikh{\sigma(\langop{\HardAutomaton{2n+2}})}=\parikh{\langop{\A}}$ (recall that
$m=n+1$). Note that loops in $\A$ that do not modify the counter are
contributed by the images of the $\sConn_{s,t}$. This will generate all inputs
because we assumed that in $\A$, every run involves at least one
non-loop transition. Observe that since the regular languages
\eqref{completeness:rega} and \eqref{completeness:regb} each require at most
$2$ states, $\sigma$ has size at most $2$.
\end{proof}

\begin{proof}[Proof of Theorem~\ref{completeness:result}]
Suppose for each $n$, there is a Parikh-equivalent NFA for $\HardAutomaton{n}$ of
size at most $h(n)$. Our proof strategy is the following. The preceding lemmas
each allow us to restrict the class of input automata further. It is therefore
convenient to define the following. Let $\C$ be a class of one-counter
automata. We say that $\C$ is \emph{polynomial modulo $h$} if there are
polynomials $p$ and $q$ such that for each OCA $\A$ in $\C$, there is a
Parikh-equivalent NFA $\B$ of size at most $q(h(p(n)))$. Of course, we want to
show that the class of all OCA is polynomial modulo $h$.

First, in section~\ref{sec:reversal-bounding}, we have seen that there is a
polynomial $p_1$ such that the following holds: if the class $\C_{p_1}$ of all
automata $\A$ that are $p_1(|\A|)$-reversal-bounded, is polynomial modulo $h$,
then so is the class of all OCA. Therefore, it remains to be shown that
$\C_{p_1}$ is polynomial modulo $h$.

Next, we apply Lemma~\ref{completeness:rba}. Together with
Lemma~\ref{completeness:substitution}, it yields that if the class of \rbas
is polynomial modulo $h$, then so is the class $\C_{p_1}$. Hence, it remains to
be shown that the class of \rbas\ is polynomial modulo $h$.

Furthermore, Lemma~\ref{lem:simple-approx} and
Lemma~\ref{completeness:loop-counting} together imply that if the class of
loop-counting \rbas\ is polynomial modulo $h$, then so is the class of all
\rbas. Hence, we restrict ourselves to the class of loop-counting \rbas.

Finally, Lemma~\ref{completeness:hardaut} tells us that the class of
loop-counting \rbas\ is polynomial modulo $h$.
\end{proof}

\end{document}